\documentclass[lettersize,journal,10pt]{IEEEtran}
\usepackage{amssymb}
\usepackage{amsthm}
\usepackage{tcolorbox}
\usepackage{latexsym}
\usepackage[hyperref]{knowledge}
\usepackage{color}

\usepackage[font={footnotesize}]{caption}
\usepackage[font={footnotesize}]{subcaption}
\usepackage{csquotes}
\usepackage{cite}
\usepackage{hyperref}
\usepackage{url}
\usepackage{epsf}
\usepackage{dcolumn}
\usepackage{array,booktabs,arydshln,xcolor}
\usepackage{float}
\usepackage{lipsum}
\usepackage{widetext}
\usepackage{cuted} 
\usepackage{textcomp}
\usepackage{enumitem}
\usepackage{algorithm}
\usepackage{algorithmicx}
\usepackage{algpseudocode}

\usepackage{wrapfig}
\usepackage{tikz}
\usetikzlibrary {arrows.meta}
\usepackage{graphicx}

\ifCLASSINFOpdf
\else
\fi
\usepackage{amsmath}
\usepackage[cmintegrals]{newtxmath}
\theoremstyle{definition}
\newtheorem{thm}{Theorem}
\newtheorem{lem}[thm]{Lemma}
\newtheorem{prop}[thm]{Proposition}

\newenvironment{pf}{{\noindent\sc Proof. }}{\qed}

\theoremstyle{definition}

\theoremstyle{remark}
\newtheorem*{rem}{Remark}

\newcommand{\argmax}{arg~max_}

\begin{document}
%
\title{A Low-Delay MAC for IoT Applications: Decentralized Optimal Scheduling of Queues without Explicit State Information Sharing}

%
%
%

\author{Avinash~Mohan,~\IEEEmembership{Member,~IEEE,}
        Arpan~Chattopadhyay,\IEEEmembership{}
        Shivam~Vinayak~Vatsa,\IEEEmembership{}
        and~Anurag~Kumar,~\IEEEmembership{Fellow,~IEEE.}
\thanks{E-mail: avimohan@bu.edu, arpanc@ee.iitd.ac.in, shivamv@iisc.ac.in and anurag@iisc.ac.in, respectively.
\emph{Avinash Mohan is the corresponding author.}}
\thanks{Avinash Mohan is a postdoctoral fellow at Boston University, Boston, Massachusetts. Arpan Chattopadhyay is with the Indian Institute of Technology, Delhi, India. Shivam Vinayak and Anurag Kumar are with the Indian Institute of Science, Bangalore, India.}
\thanks{This work was presented, in part, at the 13\textsuperscript{th} IEEE International Conference on Mobile Ad Hoc and Sensor Systems \cite{mohan-etal16hybrid-macsMASSversion}. This research was supported by the Ministry of Human Resource Development, via a graduate fellowship for the first author, and the Department of Science and Technology, via a J.C. Bose Fellowship awarded to the last author.
}
}

%
%

\markboth{QZMAC 2 column}
{Version controlled}
%
\date{June 6\textsuperscript{th}, 2020}



\maketitle

\begin{abstract}
We consider a system of several collocated nodes sharing a time slotted wireless channel, and seek a MAC (medium access control) that (i) provides low mean delay, (ii) has distributed control (i.e., there is no central scheduler), and (iii) does not require explicit exchange of state information or control signals. 
The design of such MAC protocols must keep in mind the need for contention access at light traffic, and scheduled access in heavy traffic, leading to the long-standing interest in hybrid, adaptive MACs.

Working in the discrete time setting, for the distributed MAC design, we consider a practical information structure where each node has local information and some common information obtained from overhearing. In this setting, \enquote{ZMAC} is an existing protocol that is hybrid and adaptive. We approach the problem via two steps (1) We show that it is sufficient for the policy to be \enquote{greedy} and \enquote{exhaustive.} Limiting the policy to this class reduces the problem to obtaining a queue switching policy at queue emptiness instants. (2) Formulating the delay optimal scheduling as a POMDP (partially observed Markov decision process), we show that the optimal switching rule is Stochastic Largest Queue (SLQ). 

Using this theory as the basis, we then develop a practical distributed scheduler, QZMAC, which is also tunable. We implement QZMAC on standard off-the-shelf TelosB motes and also use simulations to compare QZMAC with the full-knowledge centralized scheduler, and with ZMAC. We use our implementation to study the impact of false detection while overhearing the common information, and the efficiency of QZMAC. Our simulation results show that the mean delay with QZMAC is close that of the full-knowledge centralized scheduler.

\end{abstract}


%
\IEEEpeerreviewmaketitle

\section{Introduction}
\IEEEPARstart{I}n the Internet of Things (IoT), wireless access networks will connect embedded sensors to the infrastructure network. 
Since these embedded devices will be resource challenged, the wireless medium access control (MAC) protocols will need to be simple, and decentralized, and not require explicit exchange of state information and control signals. \emph{However, some of the emerging applications over IoT networks might expect low packet delivery delays as well} \cite{dujovne-etal14ip-enabled-industrial-iot}. In this paper, we report our work on developing a low mean delay  MAC protocol for $N$ \emph{collocated} nodes sharing a \emph{time slotted} wireless channel, such that there is no centralized control and no \emph{explicit} exchange of state information. Networks with collocated nodes regularly arise in the industrial IoT setting, specifically, machine health monitoring \cite{zhong-etal17iot-real-time-machine-status}. Emerging standards for IoT applications, such as the DetNet 
and 6TiSCH \cite{thubert-etal15sdn-meets-iot}, have shown considerable interest in systems with a 
synchronous time-slotted framework. 
It may be noted that, for this setting, a centralized scheduler \emph{with full queue length information} can just schedule any nonempty queue in each slot. The challenge we address in this paper is to develop a distributed mechanism, without explicit exchange of queue length information, that achieves mean delay very close to that of the centralized scheduler.


\indent It is well known that, while contention access (ALOHA, CSMA etc.) performs well at low contention, it can result in very large delays
and possibly instability under high contention \hspace{-0.05cm}\cite{lam84principles-comm-networking-prtcls}. While attempts have been made to stabilize CSMA, the delay of these algorithms still remains prohibitively high \cite{jiang2012stability,jiang2011approaching,rajagopalan09network-adiabatic}. Polled access (e.g., 1-limited cyclic service \cite{takagi-kleinrock85tutorial-polling-systems}, which we will call TDMA in this paper)
 on the other hand, shows the opposite behavior. 
It is, hence, desirable to have protocols that can behave like TDMA
under high contention and CSMA under low contention. The qualitative sketch in Fig.~\ref{figHybridAndDesiredMACPerformance} illustrates this behavior. The figure compares the delay performance of polling and contention MACs (orange and blue, respectively) showing the phenomenon discussed above. Also shown are two illustrations of \enquote{load adaptive} performance expected of hybrid MACs, the one in red being uniformly better than the one in green. The figure also shows the lowest possible delay attainable in this setting (black curve). This curve is discussed in more detail in Sec.~\ref{secNumRes}, where we show that the hybrid MACs proposed in this paper come extremely close to this curve and, hence, are nearly delay optimal.
\begin{figure}[bt]
 \includegraphics[scale = 0.4]{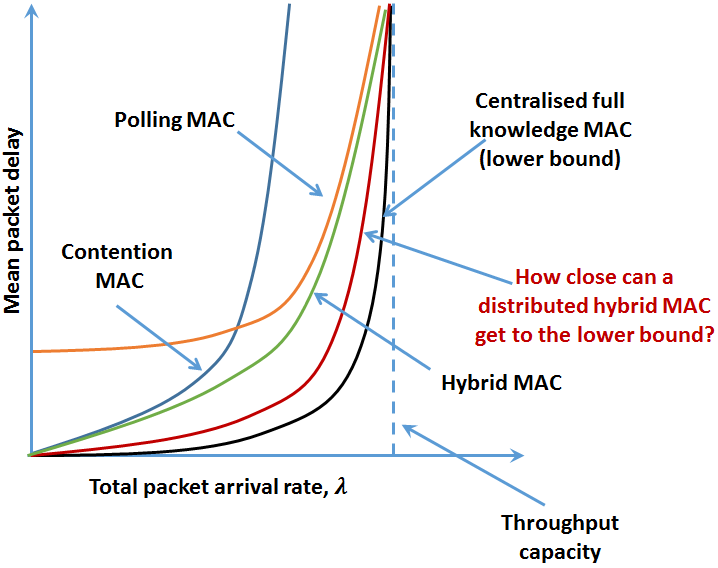}
\caption{A sketch of the delay performance of various MAC protocols. Arrival rates to all queues are assumed equal. Throughput capacity = 1, i.e. $\lambda\in[0,1)$.}
\label{figHybridAndDesiredMACPerformance}
\vspace{-0.50cm}
\end{figure}

There have been many attempts at proposing protocols that achieve this, 
especially in the context of wireless sensor networks. Important examples include \cite{nazib-etal21energy-efficient-fast-data-collection,liu-etal18qtsac-energy-efficient-delay-minimizaion,ephremides82analysis,doerr2005multimac,huang-etal18low-latency-mobile-sensor-control-systems,rhee08zmac,ahn06funneling} and \cite{sitanayah2010er-mac}. Among these, ZMAC \cite{rhee08zmac} has been observed to perform the best  in terms of delay and channel utilization (see \cite{warrier2005stochastic}), and will be an important point of comparison in this paper. 
ZMAC's delay performance, however, is quite poor compared to centralized scheduling with full queue length knowledge (see Fig.~\ref{figDelAll4Qs30}). The question, therefore, remains whether we can achieve low delay, in a decentralized setting, in which the information available to each node is just what it can acquire by \enquote{listening} over the channel.
\\
\indent Moreover, the design of the aforementioned protocols has been heuristic and without theoretical explanation of optimal scheduling of queues. Further, little attention appears to have been paid to distributed scheduling in the TDMA context, such that each node uses only locally available information.  If all queue lengths are known at a central scheduler, then ideal performance can be achieved. In the context of partial information about queue occupancies, the aim, at each scheduling instant, is to minimize the time it takes for the system to find a nonempty queue and allow it to transmit. How this can be achieved in the setting of collocated nodes with a common receiver, by using stochastic optimal control (with partial information) in conjunction with extensions of known results on optimal scheduling, is the main contribution of this paper.\\
\vspace{-8mm}
\subsection{Our Contributions}
We consider $N$ nodes sharing a slotted wireless channel to transmit packets (each of which fits into one slot) to a common receiver. Each node receives a stochastic arrival process embedded at the slot boundaries.
    We begin by describing in detail the structure of the (partial) information that nodes in the network will be assumed to possess, in Sec.~\ref{secMslot}. This structure will inform the development of optimal scheduling policies and distributed scheduling protocols. 
\begin{itemize}[leftmargin=0.15cm]
    \item 
    We then prove the delay optimality of greedy and exhaustive service policies in this novel partial information setting (in Sec.~\ref{secOptPollCont}). While our analysis employs proof techniques in \cite{liu-nain92optimal-polling} to show sample pathwise dominance, the added complexity of \emph{state dependent} server switchover delays complicates the analysis in this (and the next) section considerably.
    \item 
    Focusing on greedy and exhaustive service policies, with no explicit queue length information being shared between nodes, we derive a mean delay optimal policy by formulating the problem as a Markov decision process with partial information (Sec.~\ref{secDelayOptimalityInSymmetricSystems}).
    We initially cast the problem as an $\alpha$-discounted cost MDP, obtain the optimal policy, and extend the result to time-average costs (Sec.~\ref{secTheAverageCostCriterion}).
    We also use Foster Lyapunov theory \cite{fayolle-etal95constructive-theory-markov-chains} to modify and extend our proposed optimal policy 
    to solve the problem of unfairness that it gives rise to in our Technical Report \cite[Sec.~4.6]{mohan-etal21low-delay-iot}. Further, we discuss how our results can also handle channel errors and fading.

    \item
    We then use our delay optimality results to design two Hybrid MAC protocols ({\color{blue}EZMAC} and {\color{blue}QZMAC}) and show that the delays achieved by them are much lower than that achieved by ZMAC (Sec.~\ref{secProtocolDesign}).
    We then present modifications to QZMAC to handle unequal arrival rates, alarm traffic and Clear Channel Assessment (CCA) errors. To the best of our knowledge, this is the first work that deals with hybrid MAC scheduling for systems with \emph{unequal} arrival rates (Prop.~\ref{propSLEQunequal} and Thm.~\ref{thmCycExhStable}). 
    \item 
    We report the results of implementing QZMAC over a collocated network comprising CC2420 based Crossbow telosB motes running the 6TiSCH \cite{dujovne-etal14ip-enabled-industrial-iot} communication stack, in Sec.~\ref{sec:Experiments}. In Sec.~\ref{secNumRes}, we present simulation results comparing the delay (both mean and delay CDF) performance of QZMAC and EZMAC with that of ZMAC.  We also show that delay with QZMAC is very close to the minimum delay that can be obtained in this scenario (Sec.~\ref{secSymmetricSystemsSimulations}).
    We discuss techniques to \emph{tune} QZMAC to modify its performance over different portions of the network capacity region (see Eqn.\eqref{eqnBasicCapacityRegion}) and also compare the Channel Utilization of QZMAC with that of ZMAC.
    Finally, we conclude the paper and present directions for future research.
    
\end{itemize}

\vspace{-4mm}
\section{Frame Structure and System Processes}
\label{secSysMod}
We consider a wireless network comprising several source nodes (e.g., sensor nodes) transmitting to a common receiver node (e.g., a base station). The nodes are \emph{collocated} in the sense that all nodes can \emph{hear} each other's transmissions. This could mean, at one extreme, that they can \emph{decode} each other's transmissions, or just sense each other's transmissions. We will comment on this further in Sec.~\ref{secMechanismsForDecentralizedScheduling}. 
 Time is assumed slotted; the slots are indexed $t= 0, 1, 2, \cdots$, with slot $t$ being bounded by the epochs $t$ and $t+1$. In each slot $t$, a single node can transmit successfully (note that there is a common receiver). If a node transmits in slot $t$, at the end of the slot, i.e., at instant $t+1$, that node is viewed as the 
\enquote{node under service,} or the \enquote{incumbent,} the identity of which is assumed known to all the nodes, a property that is ensured by the information structure and our distributed algorithms.

We model the system as a network of $N$ parallel queues with a single shared server, whose service has to be scheduled between the $N$ queues. 
We denote by $A_j(t)$, the number of arrivals to Queue $j$ \emph{at} time slot boundary $t~(t=0,1,2,\cdots)$. $A_j(t)$ is assumed i.i.d Bernoulli with $P\{A_j(t)=1\}=\lambda_j$. 
The arrival processes are assumed independent of each other and the system backlog. The backlog at Queue $j$ at the beginning of time slot $t$ (i.e., at $t+$) is denoted by $Q_j(t)$ (see Fig.~\ref{figEmbeddingQueueLengthArrivalAndDepartureProcesses}). We assume that each slot can carry exactly one data packet, after allowing time for any protocol overhead.
A packet transmission from Node~$j$ in slot $t$ leads to a \enquote{departure} from the corresponding queue, which is viewed as occurring at $(t+1)-$, i.e., just before the end of slot $t$. $D_j(t)\in\{0,1\}$ indicates whether Queue~$j$ is scheduled for service in slot $t.$ 
We assume packet transmission success probability to be $1$ and remove this assumption in \cite[Sec.~11.12]{mohan-etal21low-delay-iot}. It follows from the embedding of the processes described, that the evolution of the queue-length process at Queue $j$ can be described by. 
\begin{eqnarray}
 Q_j(t+1)&=&(Q_j(t)-D_j(t))^++A_j(t+1),\\
 Q_j(0)&=&A_j(0),\nonumber
\label{eqn:qEvolution}
\end{eqnarray}
where for all $x\in\mathbb{R},(x)^+=$max$(x,0)$. \\
\textbf{Stability and Delay. }Since at most one queue (equivalently, \enquote{node}) can be scheduled for transmission in a slot and at most one packet can be transmitted in one slot, the capacity region \cite{tassiulas92stability} of this system is given by
\begin{equation}
\boldsymbol{\Lambda}=\left\lbrace\boldsymbol{\lambda}\in\mathbb{R}^N_+:\sum_{j=1}^N\lambda_j<1\right\rbrace.
\label{eqnBasicCapacityRegion}
\end{equation}
\looseness=-1 Define $\mathbf{Q}(t)=\left[Q_1(t),\cdots,Q_N(t)\right],~t\geq0.$ We say that the system is \emph{stable} if the system backlog Markov chain, $\left\lbrace\mathbf{Q}(t),~t\geq0\right\rbrace$, is positive recurrent. A protocol that is capable of stabilizing any vector in $\mathbf{\Lambda}$ is said to be \textit{Throughput Optimal} \cite{tassiulas92stability}.
Delay is defined as the number of slots between the instant a packet enters a queue and the instant it leaves the queue. 
Note that this along with the embedding described so far means that a packet experiences a delay of at least one slot.
\begin{figure}[tb]
\centering
\tikzset{every picture/.style={line width=0.75pt}} 
\resizebox{7.00cm}{3.2500cm}{
\begin{tikzpicture}[x=0.75pt,y=0.75pt,yscale=-1,xscale=1]

\draw [line width=3.75]    (5.8,200.6) -- (469.8,200.6) ;
\draw [line width=3]    (34.8,102.6) -- (34.8,200.6) ;
\draw [line width=3]    (275.8,107.6) -- (275.8,201.6) ;
\draw [color={rgb, 255:red, 24; green, 15; blue, 202 }  ,draw opacity=1 ][line width=2.25]    (39.8,236) -- (270.8,236) ;
\draw [shift={(275.8,236)}, rotate = 180] [fill={rgb, 255:red, 24; green, 15; blue, 202 }  ,fill opacity=1 ][line width=0.08]  [draw opacity=0] (14.29,-6.86) -- (0,0) -- (14.29,6.86) -- cycle    ;
\draw [shift={(34.8,236)}, rotate = 0] [fill={rgb, 255:red, 24; green, 15; blue, 202 }  ,fill opacity=1 ][line width=0.08]  [draw opacity=0] (14.29,-6.86) -- (0,0) -- (14.29,6.86) -- cycle    ;
\draw [color={rgb, 255:red, 24; green, 15; blue, 202 }  ,draw opacity=1 ][line width=2.25]    (69.8,65.6) .. controls (78.3,90.17) and (70.74,104.92) .. (44.56,112.37) ;
\draw [shift={(39.8,113.6)}, rotate = 346.87] [fill={rgb, 255:red, 24; green, 15; blue, 202 }  ,fill opacity=1 ][line width=0.08]  [draw opacity=0] (14.29,-6.86) -- (0,0) -- (14.29,6.86) -- cycle    ;
\draw [color={rgb, 255:red, 24; green, 15; blue, 202 }  ,draw opacity=1 ][line width=2.25]    (332.8,86.6) .. controls (328.15,119.15) and (306.2,121.43) .. (285.46,124.81) ;
\draw [shift={(280.8,125.6)}, rotate = 349.7] [fill={rgb, 255:red, 24; green, 15; blue, 202 }  ,fill opacity=1 ][line width=0.08]  [draw opacity=0] (14.29,-6.86) -- (0,0) -- (14.29,6.86) -- cycle    ;
\draw [color={rgb, 255:red, 92; green, 166; blue, 7 }  ,draw opacity=1 ][line width=2.25]    (273.8,196.6) .. controls (237.14,199.42) and (237.59,186.34) .. (226.15,164.82) ;
\draw [shift={(223.8,160.6)}, rotate = 59.74] [fill={rgb, 255:red, 92; green, 166; blue, 7 }  ,fill opacity=1 ][line width=0.08]  [draw opacity=0] (14.29,-6.86) -- (0,0) -- (14.29,6.86) -- cycle    ;
\draw [color={rgb, 255:red, 208; green, 2; blue, 27 }  ,draw opacity=1 ][line width=2.25]    (276,29) -- (276,44.6) -- (276,94.6) ;
\draw [shift={(276,99.6)}, rotate = 270] [fill={rgb, 255:red, 208; green, 2; blue, 27 }  ,fill opacity=1 ][line width=0.08]  [draw opacity=0] (14.29,-6.86) -- (0,0) -- (14.29,6.86) -- cycle    ;

\draw (46,33.4) node [anchor=north west][inner sep=0.75pt]  [font=\Large]  {${\textstyle \mathbf{Q_{j}( t)}}$};
\draw (234,0.4) node [anchor=north west][inner sep=0.75pt]  [font=\Large]  {${\textstyle \mathbf{A_{j}( t+1)}}$};
\draw (189,125.4) node [anchor=north west][inner sep=0.75pt]  [font=\Large]  {${\textstyle \mathbf{D_{j}( t)}}$};
\draw (287,49.2) node [anchor=north west][inner sep=0.75pt]  [font=\Large]  {${\textstyle \mathbf{Q_{j}( t) \ =\ ( Q_{j}( t) -D_{j}( t))^{+} +\ A_{j}( t+1)}}$};
\draw (119,244.4) node [anchor=north west][inner sep=0.75pt]  [font=\Large]  {$\mathbf{\text{Slot} \ t}$};
\draw (31,204.4) node [anchor=north west][inner sep=0.75pt]  [font=\Large]  {${\textstyle \mathbf{t}}$};
\draw (252,203.4) node [anchor=north west][inner sep=0.75pt]  [font=\Large]  {${\textstyle \mathbf{t+1}}$};


\end{tikzpicture}
}
\caption{Figure showing how the queue length, arrival and service processes are embedded.}
\label{figEmbeddingQueueLengthArrivalAndDepartureProcesses}
\vspace{-0.50cm}
\end{figure}
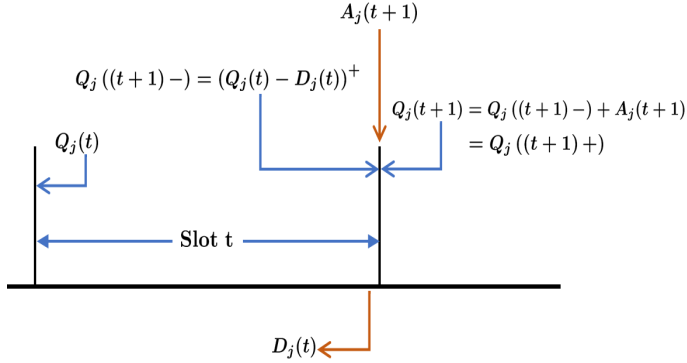

Let 
$Q_j\stackrel{\Delta}{=}\lim_{t\rightarrow\infty}\frac{1}{t}\sum_{\tau=0}^{t-1}Q_j(\tau)$, let $W_j(k)$ denote the delay experienced by the $k^{th}$ packet in Queue $j$ and $W_j\stackrel{\Delta}{=}\lim_{L\rightarrow\infty}\frac{1}{L}\sum_{k=0}^{L-1}W_j(k)$. Under stability, $Q_j$ and $W_j$ are constant with probability one; we denote these constants also by $Q_j$ and $W_j$. 
Using Little's Theorem, $W_j=\frac{Q_j}{\lambda_j}$. The delay experienced by a packet randomly chosen from the arriving stream, therefore, is given by 
$W=\sum_{j=1}^N\left(\frac{\lambda_j}{\sum_{k=1}^N\lambda_k}\right)W_j=\frac{1}{{\sum_{k=1}^N\lambda_k}}\sum_{j=1}^NQ_j$.
Our objective in the paper is to develop decentralized protocols that minimize $W.$
\\
\textbf{Centralized vs Distributed Scheduling.}\label{secMslot} In the setting described in Sec.~\ref{secSysMod}, if, at the beginning of every slot, the queue lengths are all known to a central scheduler, then (assuming zero queue-switching overheads) it suffices to simply schedule \emph{any} nonempty queue. 
The sum queue length process would then be stochastically equivalent to a discrete time, work conserving, single server queue, whose arrival process is the superposition of the $N$ arrival processes. The mean delay (equivalently, the mean queue length) would then be the smallest possible. We, however, wish to {\color{blue}avoid any explicit dissemination} of queue length information over the network, and aim to develop a scheduler that the nodes implement in a \emph{decentralized} fashion. 
\\
%
%
\vspace{-4mm}
\section{Delay Optimal Multiqueue Scheduling with Local and Common Queue Length Information}\label{secRefineModel}
\textbf{Information Structure.} 
Assuming that the collocated setting is such that all nodes can overhear and decode each other's transmissions, we consider the following natural information structure for the multiqueue scheduling problem. If each transmitted packet carries the queue length of the node at the instant it is transmitted, at the beginning of each slot, every node knows the length of every queue at the beginning of every slot in which the queue was allowed to transmit. 
At the beginning of slot $t$, each queue knows $\left\lbrace Q^\pi_j\left(\left(t-V_{j}^\pi(t)\right)-\right),~j\in I\right\rbrace$, where $V_{j}^\pi(t)$ is the number of slots prior to slot $t$ in which Queue~$j$ was allowed to transmit under a generic scheduling policy $\pi$ (the term \enquote{policy} is formally defined in Sec.~\ref{secTheCentralizedScheduler} below). As an example, if the incumbent is Queue~$k$, the packet transmitted in Slot $t-1$ would have carried the queue length $Q^{\pi}_k(t-1)$; since one packet was transmitted, $Q^\pi_k(t-) = Q^\pi_k(t) - 1$, and, since $V^\pi_k(t) = 0$,  this is also $Q^\pi_k \left(\left(t -  V^\pi_k(t)\right)-\right)$. 

In the literature, $V_{j}^\pi(t)$ is also called \enquote{Time Since Last Service} (TSLS) \cite{li-etal17emulating-round-robin}. 
In addition, at the beginning of slot $t$, every node evidently knows its own queue length; this can be viewed as local information at Queue~$j$. 
We seek a distributed scheduling policy where each node acts autonomously based on the queue information structure above, while aiming to achieve a global performance objective (say, minimizing the time average total queue length in the system). 

\vspace{-4mm}
\subsection{Our approach to developing an optimal policy} 
Our approach is via two optimal scheduling problems that each requires more than the common information.\\ 
The \emph{\color{blue}first} problem, formulated in Sec.~\ref{secTheCentralizedScheduler}, uses the common information (outlined above) and the incumbent's queue length (which is actually known only to the incumbent), to establish that, for every scheduling policy, there is a  \emph{greedy and exhaustive} scheduler that yields stochastically smaller queue lengths (Sec.~\ref{secOptPollCont}). Due to this result, we can limit ourselves to greedy and exhaustive schedulers. \\
Once we limit to a greedy and exhaustive scheduler, effectively, the decision instants are those at which the incumbent queue becomes empty, at which instants a decision has to be made to switch to one of the other queues. We formulate this \emph{\color{blue}second} problem via an average cost Markov decision process (Sec.~\ref{secMeanDelayOptimalityOfCyclicExhaustiveService}), for which the available information is the common information and some additional information. We show that, for equal arrival rates ($\lambda_k = \lambda , k \in I$), an optimal policy is Stochastic Longest Queue, i.e., at an instant at which an incumbent becomes empty, the queue to which the service switches must not be smaller than any other queue in the stochastic ordering sense. This is also equivalent to serving the {\color{blue}Longest Expected Queue} (LEQ), albeit with equal arrival rates.
Finally, the implementation of this policy only requires a distributed means for all nodes to realize that the incumbent has become empty. In Sec.~\ref{secMechanismsForDecentralizedScheduling}, we use a technique from \cite{rhee08zmac} for providing this additional information to all the nodes, thus implementing the distributed scheduler, which works only with a single bit of  common information known to all the nodes, by virtue of overhearing.\\
\vspace{-9mm}

\subsection{{\color{blue}Problem~1}: Centralized Scheduling with an Augmented Information Structure}\label{secTheCentralizedScheduler}
A centralized scheduler is defined by a scheduling policy say $\pi$, which (informally speaking), at the beginning of slot $t$, i.e., at $t+$, uses the available information and past actions to determine which queue (if any) should be allowed to transmit in slot $t$. We indicate the use of a given policy $\pi$ by placing the superscript $\pi$ on all processes associated with the system; for example, $Q^{\pi}_j(t),~t \in \{0, 1, \cdots\} $ is the queue length process of Node~$j$ under the policy $\pi$.
At the beginning of slot $t,~t\geq1$, we denote by $\pi^q_{t-1}$, the queue that had been scheduled in the previous slot, or, equivalently the \enquote{incumbent} at the beginning of slot $t.$ \\
\textbf{Notation for policies. }Under policy $\pi$, the action at $t+$ is denoted by $\mathcal{A}^\pi_{t}:=(\pi^a_t,\pi^q_t),$ where $\pi^a_t$ represents the action to be taken during slot $t$ and $\pi^q_t,$ the queue upon which the action is to be performed. $\pi^a_t\in\{0,1,s_0,s_1\}$ which mean respectively, idle at $\pi^q_t=\pi^q_{t-1}$, serve $\pi^q_t=\pi^q_{t-1}$ if nonempty, switch (away from queue $\pi^q_{t-1}$) to $\pi^q_t$ and idle there, and switch to $\pi^q_t$ and serve it if nonempty. 
With the above in mind, and anticipating a result similar to \cite [Prop.~3.2]{liu-nain92optimal-polling}, we consider a centralized scheduler that has the following information.
%
%
%
For every $t\geq1$ and policy $\pi$, the \enquote{history} of the policy contains (i) all the actions taken until $t-1,$ (ii) the backlog of each queue at the last instant it was allowed to transmit, 
 (iii) the instants at which this backlog was revealed, and (iv) the backlog of the incumbent
 at $t$, i.e.,
\begin{eqnarray}\label{eqn:DefnHistoryOfPolicyPi}
H^\pi_t&:=&\left\lbrace\left(\mathcal{A}^\pi_m\right)_{m=0}^{t-1};\left(Q^\pi_j\left(\left(t-V_{j}^\pi(t)\right)-\right)\right)_{j=1}^N; \right. \nonumber\\
&& \left. \left(V_{j}^\pi(t)\right)_{j=1}^N; Q^\pi_{\pi^q_{t-1}}(t)\right\rbrace.
\label{eqn:historyHybridInformation}
\end{eqnarray}
With reference to the information structure defined at the beginning of Sec.~\ref{secRefineModel}, the common information available to all nodes at $t$ is in $H^\pi_t$. 
Let $\mathbb{H}^\pi_t$ denote the set of all histories under $\pi$ up to time $t$. 
A \textit{deterministic admissible policy} is defined as a sequence of measurable functions from 
$\mathbb{H}^\pi_t$ into the 
action space $\mathbb{A}=\left\lbrace 0,1,s_0,s_1\right\rbrace\times I$. 
For policy $\pi$, define
\begin{equation}
 Q^\pi(t)=\sum_{i=1}^NQ^\pi_i(t),
\end{equation}
as the total backlog in the system at the beginning of time slot $t,~\forall t\geq0$.
Let the space of all admissible policies be denoted by $\Pi$. 
A policy $\pi\in\Pi$ is said to be
\begin{itemize}[leftmargin=0.1cm]
\item {\color{blue}\emph{greedy} or \emph{non-idling}} if the server never idles at a nonempty queue, i.e., at $t,$  if the incumbent queue is nonempty, it must be served again, if the decision is not to switch away. The set of all such policies is denoted by $\Pi_g\subset\Pi,$ so  \\$\Pi_g=\left\lbrace\pi\in\Pi\mid Q_{\pi^q_{t-1}}(t)>0\Rightarrow \mathcal{A}^\pi_t \in \left\lbrace 1,s_0,s_1 \right\rbrace\times I,~\forall t\geq0 \right\rbrace$.
\item {\color{blue}\emph{exhaustive}} if the server never switches away from a nonempty queue. This set is denoted by $\Pi_e\subset\Pi$, i.e., \\$\Pi_e=\left\lbrace\pi\in\Pi\mid Q_{\pi^q_{t-1}}(t)>0\Rightarrow \mathcal{A}^\pi_t\in\left\lbrace0,1\right\rbrace\times \left\lbrace\pi^q_{t-1} \right\rbrace,~\forall t\geq0\right\rbrace$.
\end{itemize}
We will now present several results with regards to scheduling in this system that will aid our design process.
\vspace{-4.5mm}
\subsection{Optimality of Non-idling and Exhaustive Policies}\label{secOptPollCont}
On applying the next proposition to the centralized problem, we need only restrict attention to policies that let an incumbent continue to transmit its packets until a time $t$ at which its queue is empty. As long as the incumbent has packets, it is optimal not to idle (i.e., transmit in every slot).
\begin{prop}\label{propNonIdlingAndExhOpt}
For the system defined in Sec.~\ref{secRefineModel} and for any policy $\pi\in\Pi$, 
$\exists~\xi\in\Pi_g\cap\Pi_e$, such that 
\begin{equation}
 Q^\xi(t)\stackrel{st}{\leq} Q^\pi(t), ~\forall~t\geq0,
\end{equation}
\label{secExhOpt}
where \enquote{$st$} denotes stochastic ordering.
\end{prop}


\begin{IEEEproof}[Proof Sketch] The structure of the model is similar to the one in \cite{liu-nain92optimal-polling}, and the result we seek is the same as \cite[Prop.~4.2]{liu-nain92optimal-polling}. There, the model is a centrally scheduled system of queues, with i.i.d. service times, \emph{nonzero} i.i.d. queue switchover times
, and knowledge of $\left\lbrace Q_i(t-V_{i,t}^\pi),~1\leq i\leq N\right\rbrace$. Since switching times are nonzero, by exhaustively serving a queue the per packet switching time is reduced, rather than switching away from a nonempty queue and then switching back to it in order to complete the service of the remaining packets. However, a {\color{blue}crucial assumption} in \cite{liu-nain92optimal-polling} is that the scheduler instantaneously obtains the queue length of every queue to which it switches. This is important since this limits their study to schedulers that never waste time attempting to serve empty queues, since, upon finding a queue empty, the scheduler is assumed to immediately switch away to another queue. 

In our system, all queues can observe the empty or nonempty status of a transmitter whenever it is allowed to transmit. However, while there are \emph{no explicit switching times}, due to decentralized scheduling, at switching instants, time could be wasted by switching to empty queues, although other nonempty queues exist. In \cite[Sec.~4.2]{mohan-etal21low-delay-iot}, we extend the coupling argument in \cite{liu-nain92optimal-polling} to establish this result.
\end{IEEEproof}
\textbf{Implication for a distributed scheduling policy.} Prop.~\ref{propNonIdlingAndExhOpt} leads to the following simplification for the distributed scheduling problem. Since it is optimal for an incumbent to serve its queue in a non-idling and exhaustive manner, a switching decision needs to be made only when the incumbent's queue becomes empty. Thus, once an incumbent is chosen, we only need a distributed mechanism for all nodes to determine the time $t$ at which the incumbent's queue becomes empty. At such a $t$, based on the common information $V^{\pi}_j(t), 1 \leq j \leq  N$ (where, of course, the element corresponding to the just past incumbent would be 0) the next incumbent needs to be selected in a distributed manner.
We answer this question in Sec.~\ref{secDelayOptimalityInSymmetricSystems} and Prop.~\ref{propSLEQunequal}. 
\begin{rem}
We are now working with \emph{truly common} information unlike \eqref{eqn:historyHybridInformation} where $Q^\pi_{\pi^q_{t-1}}(t),$ was known only to the incumbent.
\end{rem}
\vspace{-5mm}
\subsection{Equal Arrival Rates: Mean Delay Optimality of SLQ}\label{secDelayOptimalityInSymmetricSystems}

In this section we show that, when the arrival rates at all queues are equal, SLQ is mean delay optimal. The stability region with equal arrival rates is obviously the interval $\lambda=[0,1/N)$. When the incumbent queue becomes empty at some $t$, the next incumbent has to be chosen in a distributed way, and if the chosen queue turns out to be empty as well then a slot is wasted.

%
We call a policy $\pi\in\Pi$ an \textit{SLQ} (stochastically largest queue) policy if, at every instant $t$, in which a switchover takes place, the new queue chosen is not stochastically smaller than any other queue in the system. Formally, this means that whenever $\pi_t^q\neq\pi_{t-1}^q$, $V_{\pi_t^q}(t)\geq V_k(t)$ for every $k\in I$. The name SLQ arises from the fact that if  $V_l>V_k$, and the two queues were served exhaustively, their current queue lengths, with Bernoulli arrivals, are Binomial($V_l,\lambda$) and Binomial($V_k,\lambda$) random variables 
respectively. Under these conditions, we know that Binomial($V_l,\lambda$) is \emph{stochastically larger} than Binomial($V_k,\lambda$). Hence, the name. \\
 
%
%
%
\vspace{-3mm}
\begin{rem} Recall that all scheduling decisions are made at the beginnings of slots and the only information available then is 
whether the incumbent queue is empty and the number of slots since each queue in the system was last allowed to transmit 
\footnote{Note that, by Prop.~\ref{secExhOpt}, the additional knowledge that the incumbent queue has of its own queue length is not useful, since it is optimal for the incumbent to continue to transmit until its own queue is empty.}.
It is this slot-wastage, 
 which happens if the chosen queue is empty, that prevents us from using Prop.~5.1 of \cite{liu-nain92optimal-polling} by simply reinterpreting such wasted slots as a part of \emph{switchover} times.
\end{rem}
We now explicitly derive \emph{the} delay optimal policy
 for equal arrival rate into all the queues, by formulating the problem as a Markov decision process (MDP).
\subsubsection{Mean Delay Optimality of SLQ}\label{secMeanDelayOptimalityOfCyclicExhaustiveService}
Motivated by Prop.~\ref{secExhOpt}, we restrict ourselves to policies in which the channel is 
acquired by another queue only if the queue under service is empty at the beginning of a slot. 
Suppose the queue scheduled in slot $t$ is denoted by $\{u_t, t\geq0\}$. We seek a policy, say $\pi$, to choose 
$u_0, u_1, \cdots$ so as to minimize the long term average cost function

{
\small
\begin{equation}
 \limsup_{T\rightarrow\infty}\frac{1}{T}\mathbb{E}_{s(0)}\sum_{t=0}^{T-1}\left(\sum_{i=1}^NQ_i(t)\right),
 \label{eqnTACost}
\end{equation}
}

where $s(0)=[Q_1(0),\cdots,Q_N(0)]$ is the initial state of the system. This is the mean
long-term backlog in the system and, as described at the end of Sec.~\ref{secSysMod} can be used as a proxy
for mean delay. 
We first cast the problem as an $\alpha$-discounted cost minimization problem which
involves minimizing the total discounted cost 
\begin{equation}
 \mathbb{E}_{s(0)}\sum_{t=0}^\infty\alpha^t\left(\sum_{i=1}^NQ_i(t)\right),
\end{equation}
where $\alpha\in(0,1)$ is the discount factor. 
We will later arrive at the solution to the long-term time averaged cost, (\ref{eqnTACost}), using a limiting
procedure with a sequence of optimal policies $\pi_k$ corresponding to a sequence of discount
factors $\alpha_k\uparrow1$. 
In what follows, we invoke the results of Prop.~\ref{secExhOpt}, and focus only on $\Pi_g\cap\Pi_e$. 

\subsubsection{Formulating The Discounted Cost MDP}
\label{secOptSlq}
We will consider an information structure that includes the common information, and some local information. 
At the beginning of slot $t$ all the queues know (i) the index of the incumbent (ii) the number of slots since every queue was last served $\{V_k(t),k\in I\}$ (iii) the queue-length of the incumbent (at $t$), and (iv) residual known queue lengths in the queues (if any). Although we begin with this elaborate information structure, which has common information (known at all nodes), and local information (known only to each node), it will turn out that the optimal policy is such that (a) the action it takes at $t$ is a function only of $\{V_k(t), k \in [N]\}$ and the empty-nonempty status of the incumbent, and (b) it serves queues exhaustively, so the residual known queue lengths in the queues are all zero. Thus, the information structure we have in our problem (plus, an item of information that easily shared with all the nodes by using a simple, known technique (Sec.~\ref{secMechanismsForDecentralizedScheduling})) is sufficient to implement the optimal policy we derive and the \emph{generalization}  of the information structure (above) is only required for the proof of delay optimality to go through. 

The MDP, under the information structure wherein the queues know the backlog of the incumbent and $\mathbf{V}(t)$, has $s(t)=\left[Q_{u_{t-1}}(t);\mathbf{V}(t);\mathbf{r}(t);u_{t-1}\right]^T$ as its state. 
The new coordinate $\mathbf{r}(t)$ is explained as follows. At the beginning of slot $t$ (just after arrivals occur), the server knows that Queue~$j$ ($j\neq u_{t-1}$) has length $r_j(t)+U_j,$ where $U_j$ is a Binomial$(V_j(t),\lambda)$ random variable.
To accommodate the new information $\mathbf{r}(t)$, we expand our focus to the class of policies that always serve any queue $j$ such that $r_j(t)>0$ in time slot $t$. This means that if some queue $j$ has $r_j(t)>0$ at the beginning of time slot $t$, it is served in that slot and not the incumbent ($u_{t-1}$), even if $Q_{u_{t-1}}(t)>0.$ The underlying principle is that we can serve any queue that is \emph{known to be nonempty} at the beginning of the slot.

But these policies are still greedy and exhaustive, in the (restricted) sense that whenever $\mathbf{r}(t)=\mathbf{0}$ and $Q_{u_{t-1}}(t)>0$, $u_t=u_{t-1}$, i.e., if no other queue is known to have packets and the incumbent is nonempty, it \emph{is} served. This is a purely technical construction required for the proof of Thm.~\ref{thmCycExhDelOpt} and, as we will show later, the optimal policy (i.e., cyclic exhaustive service) will not require the knowledge of $\mathbf{r}(t)$ at all. 
The action in every slot $t$ with an empty incumbent, involves choosing a queue $u_t\in I$ and the single-step cost is the expected sum of the current queue lengths conditioned on the
current state, given by
\vspace{-0.2cm}
{
\small
\begin{eqnarray}
 c(s(t),u_t)=\mathbb{E}\left[\sum_{i=1}^N Q_i(t)|s(t)\right]
 =Q_{u_{t-1}}(t)+\lambda\sum_{j\neq u_{t-1}}V_j(t)+\sum_{i\in I}r_i(t).
 \label{eqnCOMDPssCost}
\end{eqnarray}
}

Define the state space as $\mathbb{S}:=\mathbb{N}\times\mathbb{N}^N\times\mathbb{N}^N\times I$ and optimal discounted cost $J^*:\mathbb{S}\longmapsto\mathbb{R}_+$ starting in state $s(0)\in\mathbb{S}$ as 
\begin{equation}
 J^*(s(0))=min_{\xi}~\mathbb{E}^\xi_{s(0)}\sum_{t=0}^\infty\alpha^t\mathbb{E}\left[\sum_{i=1}^NQ_i(t)|s(t)\right],
\label{eqnDiscountedCostOfCOMDP}
\end{equation}
and let $q$ represent the backlog of Queue~$i$, the incumbent. 
The Bellman Optimality equations \cite{bertsekas95mdp-control} associated with the MDP formulation as
described above are as given in (\ref{eqnBellmanOpt})
(we denote $J^*(q; \mathbf{V}; \mathbf{r}; i)$ by $J_i^*(q; \mathbf{V}; \mathbf{r})$).
\vspace{-0.2cm}
{
\small
\begin{align}
J^*_i(q>0; \mathbf{V}; \mathbf{r}=\mathbf{0})&=q+\lambda \sum_{k \neq i} V_k +\alpha \mathbb{E} J^*_i\bigg(q-1+A;\nonumber\\
V_i=0,\mathbf{V}_{-i}+\mathbf{1}; \mathbf{r}=\mathbf{0}\bigg),\hspace{-5.5cm}\nonumber
\end{align}
\vspace{-0.2cm}
\begin{align}
J^*_i(q=0; \mathbf{V}; \mathbf{r}=\mathbf{0})=\lambda \sum_{k \neq i} V_k +\alpha\min_{j\neq i}\mathbb{E} J^*_j\bigg( B(V_j);\hspace{2cm}\nonumber\\%
V_j=0,\mathbf{V}_{-j}+\mathbf{1}; \mathbf{r}=\mathbf{0} \bigg),\hspace{2cm}
\label{eqnBellmanOpt}
\end{align} 
}


In \eqref{eqnBellmanOpt}, $A$ is a generic Bernoulli($\lambda$) random variable, $\mathbf{V}_{-i}=[V_1,\cdots,V_{i-1},V_{i+1},\cdots,V_N]$ and $\mathbf{1}\in\mathbb{R}^{(N-1)}$ is the vector with 1's at all coordinates. Finally, if random variable $C$ is distributed Binomial($V_j,\lambda$), $B(V_j)$ is a random variable whose distribution is the same as that of $(C-1)^++A$. 
The detailed formulation of the MDP and the Bellman Optimality equations for the case where $r_j\geq1$ for some $j\in I-\{i\}$ can be found in \cite[Sec.~11.6]{mohan-etal21low-delay-iot}. 
\subsubsection{Solution to the Discounted Cost MDP}\label{secPlcyStr}
Clearly, no decision needs to be taken when either $q>0$ or (from Prop.~\ref{propNonIdlingAndExhOpt},)
 when $\mathbf{r}\neq\mathbf{0}$. Since the policies
we seek are non-idling and exhaustive, they will simply continue to serve the incumbent when $q>0$ and therefore, states with $\mathbf{r}\neq\mathbf{0}$ never actually arise. 
We will now show that if $q=0$ and $\mathbf{r}=\mathbf{0}$, the system must choose 
$n:=\arg \max_j V_j$. In what follows, let $V_m < V_n$ for some $m\in I,~\text{and}~n\neq m$. 
We now prove that choosing $n$ results in the least cost.
\begin{thm}\label{thmCycExhDelOpt}
When\footnote{In what follows, by \enquote{$V'_m=0, \mathbf{V}_{-m}+\mathbf{1}$,} we mean the $m$\textsuperscript{th} coordinate of $\mathbf{V}$ is $0$ and the other $N-1$ coordinates increase by $1.$} $V_m<V_n$ and $\lambda_i=\lambda,~\forall i\in[N],$
{\small
\begin{eqnarray}
\mathbb{E} J^*_n\bigg( B(V_n); V'_n=0, \mathbf{V}_{-n}+\mathbf{1};\mathbf{r}_{-n}=\mathbf{0}\bigg)\leq\nonumber\hspace{1cm}\\
\mathbb{E} J^*_m \bigg( B(V_m); V'_m=0, \mathbf{V}_{-m}+\mathbf{1};\mathbf{r}_{-m}=\mathbf{0}\bigg).
\label{eqn:cyclicExhaustiveOptimal}
\end{eqnarray}
}
\end{thm}

\begin{proof}The expectation in \eqref{eqn:cyclicExhaustiveOptimal} is with respect to arrivals.
Intuitively, with equal arrival rates, the queue that has not received service the longest is also the most likely to be non-empty. Hence, attempting to serve this queue will result in the slot being wasted with the least probability. The proof of the theorem can be found in \cite[Sec.~11.7]{mohan-etal21low-delay-iot}. 
\end{proof}
\vspace{-2mm}
\textbf{Implication for a Distributed Scheduler.} The policy above, denoted $\pi^*$ in what follows, chooses $\arg \max_j V_j(t)$ in every slot $t$ with an empty incumbent; in addition, by exhaustive service $\mathbf{r}(t)=\mathbf{0}$ for all $t$. 
In other words, $\pi^*$ is a stationary Markov policy that, at the beginning of each slot, chooses the incumbent, if nonempty, and the queue that has not been served the longest until that epoch, if the incumbent is found empty.
\\
\indent Now, before moving on to our original average cost problem, we make an important observation about the LEQ scheduling policy. Recall that the LEQ policy chooses $\argmax{1\leq i\leq N}\lambda_iV_i(t)$ at every decision instant.
\prop\label{propSLEQunequal}
The LEQ policy stabilizes all arrival rate vectors in the set

{
\small
\begin{equation}
\Lambda_{LEQ}:=\bigg\{\boldsymbol{\lambda}\in\mathbb{R}^N_+\bigg|\sum_{i=1}^N\lambda_i<1~\text{and}~\min_{1\leq i\leq N}\lambda_i>0\bigg\}.
\label{eqnCapacityOfFullyConnected}
\end{equation}
}

\begin{IEEEproof} The proof essentially involves showing that the expected time to find the next nonempty queue (after serving some queue exhaustively) is bounded. Refer to \cite[Sec.~11.3]{mohan-etal21low-delay-iot} for details of the proof.
%
\end{IEEEproof}
%
\begin{rem}
    \emph{With $\lambda_i = \lambda, \forall i$, we have and the policy is easily seen to be a \emph{cyclic} exhaustive policy, i.e., queues are served in a fixed cycle, with each queue being served to exhaustion (for a quick example, see \cite[Sec.~11.4]{mohan-etal21low-delay-iot}).
    $\argmax{1\leq i\leq N}\lambda_iV_i(t)=\argmax{1\leq i\leq N}V_i(t)$ and the policy reduces to a {\color{blue}\emph{cyclic} exhaustive service} policy.}
\end{rem}

In the literature, $V_j(t)$ is also called Time-Since-Last-Service (TSLS). The TSLS-based algorithm proposed in \cite{li-etal17emulating-round-robin} schedules queues in a round robin fashion and requires knowledge of the empty-non empty status of every queue in the network in every time slot. However our theoretical results and the QZMAC algorithm, both of which \emph{preceded}  \cite{li-etal17emulating-round-robin}, require this information of queues only at the last time they were served. Our algorithms, therefore, actually are amenable to decentralized implementation.\\
\indent In Prop.~\ref{prop:classTOStatMarkovPolicies} in the Appendix, we show that the class of stabilizing, stationary Markov policies is \emph{not} a singleton, and the policies therein have different mean delays. Thus, $\pi^*$ is not a trivial solution to our MDP.
\subsubsection{The Average-Cost Criterion}\label{secTheAverageCostCriterion}
We use the technique described in \cite{sennott1989average-cost-mdps} to show that $\pi^*$ is
optimal for the long term time-averaged cost criterion as well. We consider a sequence of discount
factors\footnote{Given a sequence $\{x_n\}_{n\geq1},\{x_n\}\uparrow1$ means $x_n$ approaches $1$ from below, in an increasing manner.} $\{\alpha_n\}\uparrow1$ and the corresponding optimal policies and show that under certain conditions, the limit point of this sequence solves the Average-Cost problem, i.e., \eqref{eqnTACost}. For details, refer \cite[Sec.~4.5]{mohan-etal21low-delay-iot}. 
Furthermore, in \cite[Thm.~3]{mohan-etal21low-delay-iot}, we show another desirable property of $\pi^*$, namely, that it is \emph{\color{blue}throughput optimal.}
\vspace{-4mm}
\subsection{The Distributed SLQ Policy}\label{sec:theDistributedSLQPolicy}

From Prop.~\ref{propNonIdlingAndExhOpt} we concluded that we can limit our scope to greedy and exhaustive policies. This left us with the problem of determining a switching rule when the scheduled queue is exhausted. From the MDP in Sec.~\ref{secDelayOptimalityInSymmetricSystems}, for the case of equal arrival rates,  we have just learnt that at such a switching point, say, $t$, assuming that all nodes know the number of slots back that each node was last served (i.e., $V_k(t), k \in [N]$), it is optimal to switch to a queue $\argmax{k \in [N]} V_k(t)$.  For equal arrival rates, this is the same as the SLQ or the LEQ policy.

The following issues remain in the design of a distributed policy. \textbf{\color{blue}(1)} The design of a distributed mechanism by which all nodes can determine that an incumbent node  has an empty queue at a scheduling instant $t$ (Sec.~\ref{secMechanismsForDecentralizedScheduling}). 
\textbf{\color{blue}(2)} Greedy and exhaustive service locks out all the other queues for a long time. This leads to short-term unfairness. At the end of Sec.~\ref{sec:handling-CCA-errors-subroutine}, we introduce a technique to address this concern.

\vspace{-3.5mm}
\section{Mechanisms for Decentralized Scheduling}\label{secMechanismsForDecentralizedScheduling} 
In this section, we describe mechanisms by which transmission and lack of transmission in a slot, by the incumbent node can be inferred by all other nodes in the system, \emph{\color{blue}without any explicit exchange} of information. The lack of transmission by the incumbent node is the trigger for switching to another incumbent.
\\
\textbf{Transmission sensing:} We assume that all nodes transmit at the same fixed power, and the maximum internode distance is such that every other node can sense the power from a transmitting node. Suppose a node has been scheduled to transmit in a slot. Then, whether or not the node actually transmits can be determined by the other nodes by averaging the received power over a small interval, akin to the \emph{\color{blue}CCA mechanism} (Clear Channel Assessment) 
\cite{kinney01cca-method-wpan-working-group}. For reliable assessment, the slot will need to be of a certain length, and the distance between the nodes will need to be limited. Let such an activity sensing slot be called a \emph{\color{blue}minislot} (Fig.~\ref{figMinislotStructure}). See Sec.~\ref{secFrameStructureTimeSlotsAndMinislots} for further details.
\begin{figure}[tb]
\centering
\includegraphics[scale = 0.23]{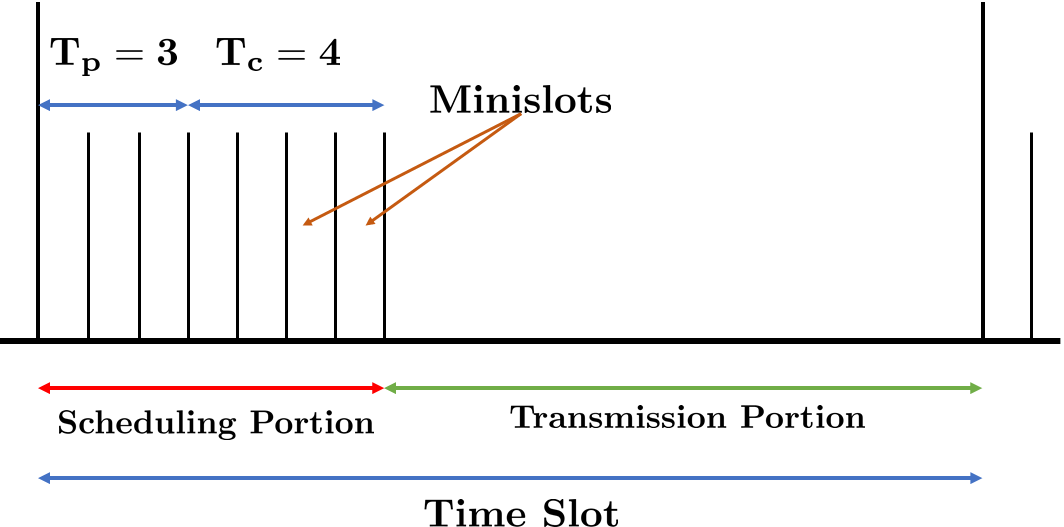}
\caption{Illustrating the slot and minislot structures. Since $T_p=3$, \emph{three} poll-and-test procedures can be performed by the system (incumbent plus two other queues). If none of these results in identifying a nonempty queue, the system goes into contention over the next $T_c=4$ minislots.}
\label{figMinislotStructure}
\vspace{-5mm}
\end{figure}
\begin{figure}[tb]
 \centering
 \includegraphics[scale = 0.23]{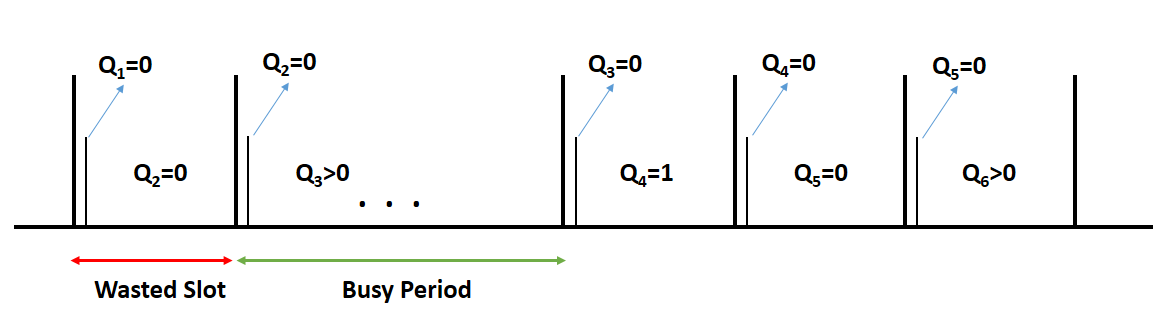}
 \caption{Sample path illustrating how $\pi^*$ is implemented in a system with $T_p=1$ and $T_c=0.$ Of particular note is the initially \emph{wasted} slot. Since the poll-and-test did not identify any nonempty queue and since the system cannot go into contention ($T_c=0$), the slot is wasted.
\vspace{-7mm}
 } 
 \label{figSys2_omega}
 \end{figure} 
\\
\textbf{Further optimization.} \textbf{\color{blue}(1)} When an incumbent queue becomes empty (at, say $t$), an application of the policy results in switching to another queue (say, $k$). There is now a positive probability that $Q_k(t) = 0$, which will result in the slot $t$ being wasted. The above approach of determining the emptying out of the incumbent's queue can be used immediately in the same slot, if there is a provision of a \emph{second} sensing slot. In this manner, if there are $T_p$ sensing slots, up to $T_p$ successive, wasteful switchings can be avoided. Each of these sensing actions can be viewed as a \enquote{poll-and-test} operation.
\\
\textbf{\color{blue}(2)} In a light traffic setting, all the poll-and-test procedures might fail, with high probability. Having a large value of $T_p$ reduces this probability but increases the overhead. So, the next mechanism is \emph{contention,} for which we have $T_c (\geq 0)$ minislots following the poll-and-test slots. The nodes that have not been polled can contend over these $T_c$ minislots, hopefully leading to the identification of one nonempty queue. 
\\
\textbf{\color{blue}(3)} Fig.~\ref{figSys2_omega} illustrates this mechanism for a simple system with $T_p=1,T_c=0$. In the first slot in the figure, the absence of power indicates, to all nodes in the network, that Node~1 is now empty and Node~2 is allowed to transmit. Since Node~2 is also empty, and since no more poll-and-test procedures are possible (since $T_p=1$), the slot is \emph{wasted.} If $T_p=2$, and $T_c\geq1$, sensing that Node~2 is empty, the other nonempty nodes in the network will contend for the slot.
\vspace{-2mm}
\section{Protocol Design}\label{secProtocolDesign}
In this section, use the results in Sec.~\ref{secRefineModel} to motivate the practical design of our protocol, \textbf{QZMAC}. 
We will clearly describe how the optimal polling schemes in Section \ref{secOptPollCont} are implemented in a distributed manner which will establish the self-organizing nature of our protocols.
\vspace{-4mm}
\subsection{The EZMAC Protocol}
\label{secEz}
We begin by explaining the limitations of ZMAC. 
As mentioned before, ZMAC first sets up a TDMA schedule allotting a slot to every queue in a
cyclically-repeating frame, so $N$ nodes in the system means a frame with $N$ slots. Slot $j,~1\leq j\leq N,$ in every frame is assigned to Node $j$ which is called the \emph{Primary User} (PU) in this slot. The other $(N-1)$ nodes, i.e., nodes $1, 2,\cdots,j-1,j+1,\cdots,N,$ are called \emph{Secondary Users} (SUs) in this slot. These SUs are allowed to contend for transmission rights in every TDMA slot with an empty PU. This is 
accomplished as follows. At the \emph{beginning} of each time slot, each queue checks if it has a packet. If it does, and this is its
TDMA slot (which means it is the PU), it proceeds to transmit the packet. If the 
current slot is \textit{NOT} its TDMA slot, the queue first checks if the PU is transmitting (for a period 
of 1 minislot). If it hears nothing, it backs off over a duration chosen uniformly randomly over 
$\{T_p+1, \cdots, T_p+T_c\}$, and if the channel is clear, starts transmitting. Collisions are assumed to be
detected instantaneously and that slot is assumed wasted. Note that a mechanism is needed to set up the TDMA schedule before transmissions can begin (e.g., DRAND \cite{rhee06drand}).
The pseudocode for ZMAC is in our Tech Report \cite[Sec.~6.1]{mohan-etal21low-delay-iot}.
\begin{figure}[b]
\centering
\includegraphics[scale = 0.3]{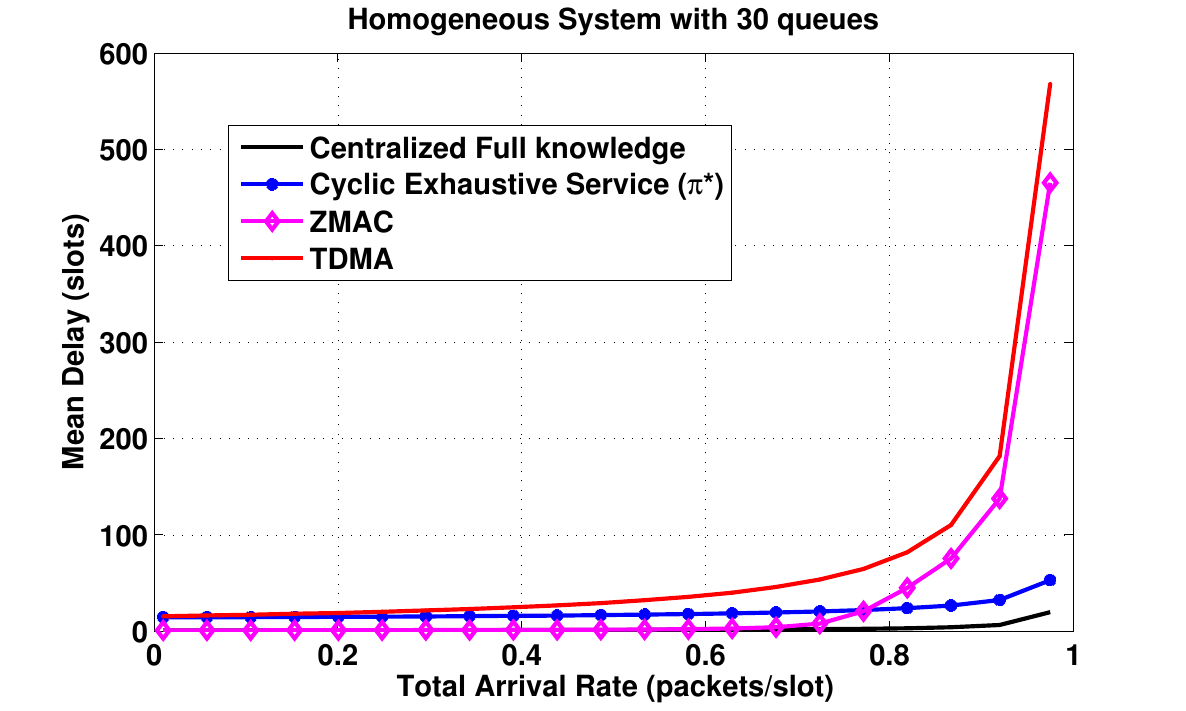}
\caption{Comparing $\pi^*$ and ZMAC in a \emph{symmetric} system with 30 queues. The total arrival rate to the system = $N\lambda\in[0,1)$.}
\label{figCrossoverSLQandZMAC}
\end{figure}

As opposed to ZMAC, $\pi^*$ shows larger mean delay at lower arrival rates because the scheduled queue is, with high probability, going to be empty. It is for this reason that the mean delay curves of ZMAC and $\pi^*$ crossover near $\lambda = 0.75$ in Fig.~\ref{figCrossoverSLQandZMAC}. In Sec.~\ref{secQz}, we solve this problem as well, and in doing so, arrive at the design of our protocol, {\bf QZMAC}. 


  
Our first protocol, EZMAC, differs from ZMAC in the contention resolution (CR) portion. Here, once the winner of a contention is determined, it is allowed to transmit in all slots without a PU until it empties. It is assumed that at the end of the winner's transmission, the packet contains an end of transmission message (a bit in the header, perhaps) that can be decoded by the other users. We relegate details to \cite[Sec.~6]{mohan-etal21low-delay-iot}, but simulations results clearly show the improvement this change brings, to both mean delay and the delay CDF. 
\vspace{-5mm}
\subsection{The QZMAC protocol}\label{secQz}
In Sec.~\ref{secOptSlq} we proved that cyclic exhaustive service is delay optimal in systems with equal arrival rates. 
Recall that Thm.~\ref{thmCycExhDelOpt} therein, assumed $T_p=1.$ However, we use this result to propose a scheduling protocol, QZMAC, for general $T_p\geq1,$ and show through simulations that violating the $T_p=1$ assumption does not hurt the performance of QZMAC. Solving the MDP for general $T_p$ turns out to be a hard problem, and as simulation results in Figures \ref{figDelAll4Qs10} and \ref{figDelAll4Qs30} show, \emph{cannot} result in any dramatic improvement in delay.
We now describe \emph{QZMAC} for $T_p=3$ minislots. Every queue maintains its own copy of a vector $\mathbf{V}(t)$ that is used to render the scheduling process fully distributed. The protocol proceeds as in Protocol.~\ref{protocol:QZMAC}.
 
  \begin{algorithm}[tbh]
    \small 
    \floatname{algorithm}{Protocol}
       \allowdisplaybreaks
       	\caption{{\bf QZMAC} (with $T_p=3$ minislots),} \label{protocol:QZMAC}
       	\begin{algorithmic}[1]
       	\State \textbf{Input:} $N,T_c${\color{blue}\Comment{This runs at every Queue~$j,~j\in[N]$ in the network}}
       	
       	\State \textbf{Init:} $t\leftarrow 0,\tau\leftarrow1,PU\leftarrow1,SU\leftarrow2$  {\color{blue}\Comment{$\tau$ keeps track of the minislot number}}
       	\State \textbf{Init:} $\forall~k\in[N], V_k(0)\leftarrow k$ 
       	\While{$t\geq0$}
       	
       	$\tau\leftarrow1$ {\color{blue}\Comment{Keeps track of minislot number.}}
       	 \If{$PU==j$ \&\& $Q_{j}(t)>0$}{\color{blue}\Comment{If you are the PU}}
       	    \State $PU$ transmits packet.
       	    \For{$k=1:N$} 
       	        \State $V_k(t)\leftarrow \left(V_k(t)+1\right)\mathbb{I}_{\{k\neq PU\}}$
       	    \EndFor
       	    \State GOTO \ref{QZMAC:incrementTime}
       	\Else 
       	    \State $\tau\leftarrow\tau+1$ {\color{blue}\Comment{The 2\textsuperscript{nd} of $T_p$ minislots}}
       	    \State $i^*\leftarrow\argmax{i\in[N]}V_i(t)$ \label{qzmacProtocol:ChooseLEQ}
   	        \State $PU\leftarrow i^*$
       	    \If{$i^*==j$ \&\& $Q_{i^*}(t)>0$} 
       	        \State $PU$ transmits packet.
       	        \For{$k=1:N$}
       	        \State $V_k(t)\leftarrow \left(V_k(t)+1\right)\mathbb{I}_{\{k\neq PU\}}$
       	        \EndFor
       	        \State GOTO \ref{QZMAC:incrementTime}
       	        \ElsIf{$SU==j$ \&\& $Q_{SU}(t)>0$} 
       	            \State $\tau\leftarrow\tau+1$ {\color{blue}\Comment{The last $T_p$ minislot}}
       	            \State $SU$ transmits packet.
       	            \State GOTO \ref{QZMAC:incrementTime}
       	        \Else {\color{blue}\Comment{SU empty $\Rightarrow$ network enters contention}}
       	            \label{qzmacProtocol:incont}
       	            \State {\bf Over} $\tau\in\{4,5,\cdots,T_c+3\}$ {\bf do}
       	                \If{$Q_j(t)>0$} {\color{blue}\Comment{Queue~$j$ knows its own backlog}}
       	                    \State $U_j\sim Unif\left(\lbrace4,5,\cdots,T_c+3\rbrace\right)$
       	                    \State Wait for $U_j-4$ minislots. 
       	                    \State If \texttt{CCA == SUCCESS,} Tx pkt.
       	                    \State If no collisions detected, $SU\leftarrow j$
       	                \EndIf
       	  \EndIf
       	  \label{qzmacProtocol:outcont}
       	 \EndIf 
       	 \State\label{QZMAC:incrementTime} $t\leftarrow t+1$ {\color{blue}\Comment{Next time slot}}
       	\EndWhile
       	\end{algorithmic} 
 \end{algorithm}
There are three points to note here. \textbf{\color{blue}(1)} QZMAC has clearly been obtained by \emph{separately} optimizing polling and contention protocols. Jointly optimizing both turns out to be intractable and even separate optimization ultimately shows excellent delay performance. 
\textbf{\color{blue}(2)} Note that depending on the value of $T_p$, QZMAC can show a range of behavior. 
When $T_p=1$, the system never enters contention since one minislot is spent in ascertaining that the incumbent is empty and thereafter, $i^*$ (see Step~\ref{qzmacProtocol:ChooseLEQ}) is allowed to transmit. Determining that $i^*$ is empty requires another minislot which is not available since $T_p=1$, and the system can enter contention only when $i^*$ and the SU are empty. 
So, in order for the system to enter contention, QZMAC needs $T_p=3$ (one each for the PU, $i^*$ and the SU). Similarly, ZMAC requires $T_p=1$. 
\textbf{\color{blue}(3)} As mentioned in Sec.~\ref{secPlcyStr}, the use of the $\mathbf{V}(t)$ vector automatically induces a cyclic schedule. 
This is a much simpler technique than DRAND, used in ZMAC \cite{rhee06drand}, which involves several rounds of communication among the nodes to converge to a TDMA schedule, even for fully connected interference graphs.
\vspace{-2mm}
\section{Extending QZMAC}\label{secQzmacDesign}
In this section, we show how QZMAC can be modified to handle a variety of different applications. 
\vspace{-4mm}
\subsection{Handling Unequal Arrival Rates}
Prop.~\ref{propSLEQunequal} helps generalize QZMAC very easily to accommodate unequal arrival rates. The only modification necessary here is to Step~\ref{qzmacProtocol:ChooseLEQ} where system now needs to choose $\argmax{1\leq i\leq N}{\color{blue}\lambda_i} V_i(t)$. 
Clearly, finding $\argmax{1\leq i\leq N}\lambda_iV_i(t)$ in every scheduling slot requires knowledge of the entire arrival rate vector $\boldsymbol{\lambda}$. Arrival rate estimation might not be feasible in several sensor networks of today that are expected to begin performing as soon as they are installed. For more information about this \enquote{peel and stick} paradigm refer to \cite{dujovne-etal14ip-enabled-industrial-iot}. 

We resolve this problem by first observing that under stability, the time average of the total number of packets transmitted by any queue tends to the arrival rate. Let $D_i(t)$ be the random variable that denotes whether or not a departure occurred from queue $i$ at the end of slot $t$. Then, under stability,
$\lim_{T\rightarrow\infty}(1/T)\sum_{t=0}^TD_i(t)=\lambda_i,~a.s.$
So, let $\hat{\lambda}_i(T):=(1/T)\sum_{t=0}^TD_i(t).$
The policy that chooses $\argmax{1\leq i\leq N}\hat{\lambda}_i(t)V_i(t)$ in every scheduling slot also shows the same mean delay performance as the LEQ policy, as evidenced by simulation results presented in Sec.~\ref{secDistributedLEQ}. This estimate can be maintained independently at each queue and the resulting decisions are still consistent.

\vspace{-4mm}
\subsection{Handling CCA errors}\label{sec:handling-CCA-errors-subroutine}
We have, hitherto, assumed that when a network queue tests for channel activity (or lack thereof), the test always succeeds. In real wireless sensor networks this operation, called a \emph{Clear Channel Assessment} (CCA), involves a hypothesis test based on \emph{noisy} samples of channel activity and hence, is susceptible to error. This makes a 100\% success rate a strong assumption and we relax it in this section. Notice that crucially, CCA errors affect the fidelity of the $\mathbf{V}(t)$ vector across nodes and we now have a matrix $V_{N\times N}(t)=[\mathbf{V}_1(t),\mathbf{V}_2(t),\cdots,\mathbf{V}_N(t)]$, where $\mathbf{V}_i(t)$ is the local copy at Queue~$i.$

Extensive experimentation (reported in Sec.~\ref{sec:Experiments}) reveals that \emph{absence} of activity on the channel can be detected without error. This means that whenever a queue that is supposed to transmit (incumbent, $i^*=\argmax{i\in[N]}V_i(t)$ or the SU) is empty, all CCAs across the network in that minislot report a \enquote{clear channel.} In detection theoretic parlance, the probability of a \emph{false alarm} is zero, i.e., $p_{FA} = 0$. 
On the other hand, our observation is that the probability of CCA declaring an \enquote{active} channel as \enquote{clear} is not zero, i.e., $p_{miss} > 0$. However, our experiments show that $p_{miss}\approx3\times10^{-6}$, i.e., CCA miss is a \emph{rare} event. 
\\
\indent We begin by listing the different types of misalignment such errors can produce across the columns of the aforementioned $V(t)$ matrix in \cite[Sec.~11.14]{mohan-etal21low-delay-iot}, of which only the category termed M2 therein, necessitates modifications to QZMAC. We explain therein, how QZMAC automatically resolves the other types of misalignment. A misalignment of type M2 occurs when the $i^*$ Node assumes that the incumbent is empty and, being \emph{nonempty,} begins transmitting. Naturally, the transmissions from the two nodes collide \emph{persistently,} and further provisions are now required within QZMAC to extricate the network from this state. 
 Due to paucity of space, we refer the reader to \cite[Sec.~7.2]{mohan-etal21low-delay-iot} for the \texttt{\color{blue}RESET} subroutine designed to do precisely this. 
\textbf{Short term unfairness.} Given that QZMAC is based on exhaustive service, nodes might end up being starved for service in the short term. We analyze the short term fairness performance of QZMAC in \cite[Sec.~4.6]{mohan-etal21low-delay-iot}. To alleviate this issue, we modify QZMAC using a $K$-limited polling scheduler which we prove is \textit{\color{blue}throughput optimal} (Prop.~6).
\vspace{-3mm}
\section{Simulation and Experimental Results}\label{secNumRes}
We now report the results of our extensive implementation and simulation studies of the various algorithms proposed in the earlier sections. 
We also deal with practical issues, such as arrival rate estimation and channel utilization. We begin with the performance of the LEQ scheduler analyzed in Prop.~\ref{propSLEQunequal}. 
Protocols such as ZMAC \cite{rhee08zmac} and subsequent developments such as \cite{doerr2005multimac}, \cite{ahn06funneling} and \cite{sitanayah2010er-mac} have not been designed to address the issue of unequal arrival rates. These protocols are not capable of stabilizing queues in this general setting, let alone provide low delay. Furthermore, the rigid cyclic TDMA structure imposed by the slot assignment protocol prevents extension to unequal arrival rates. 

To the best of our knowledge, this is the first work that deals with hybrid MAC scheduling for systems with unequal arrival rates. Consequently, we do not have any hybrid algorithms against which to compare the LEQ policy. The ZMAC protocol's cyclic TDMA schedule assigns a single slot to each queue in one TDMA frame and hence is not Throughput Optimal. Furthermore, computing even an approximately delay optimal TDMA schedule requires complete knowledge of the arrival rate vector and even then requires prohibitively high message passing between nodes. One such procedure is given in \cite{hofri-rosberg87packet-delay-golden-ratio}. It can be seen that the lower bound on delay even within the class of cyclic TDMA policies cannot always be achieved by this method; one example is when the arrival rates are irrational. Moreover, the general problem of finding the periodic TDMA schedule of minimum length is NP-complete \cite{ahmad-etal2008efficient-algorithm-broadcast-schedule-tdma}.
We, therefore, provide comparisons only with the lowest delay that can be achieved in this system, i.e., the one with a centralized scheduler that schedules a nonempty queue in every slot. 
\vspace{-4mm}
\subsection{Distributed Implementation of the LEQ policy}\label{secDistributedLEQ}
In Fig.~\ref{figQzmac10queues} and Fig.~\ref{figQzmac30queues}, we show the delay performance of QZMAC which takes decisions based on  exact knowledge of arrival rates and the distributed version that only uses estimates of arrival rates. The green and red curves, corresponding to scheduling with estimated and exact arrival rates respectively, overlap significantly, showing that arrival rate estimation does not degrade performance. Moreover, in small systems (Fig.~\ref{figQzmac10queues}), the delay is almost the lowest that can be achieved, while with larger systems, the difference with the centralized scheduler's delay becomes nonnegligible only near saturation. 

\begin{figure*}[ht]
\begin{subfigure}{.5\textwidth}
\centering
\includegraphics[scale = 0.3]{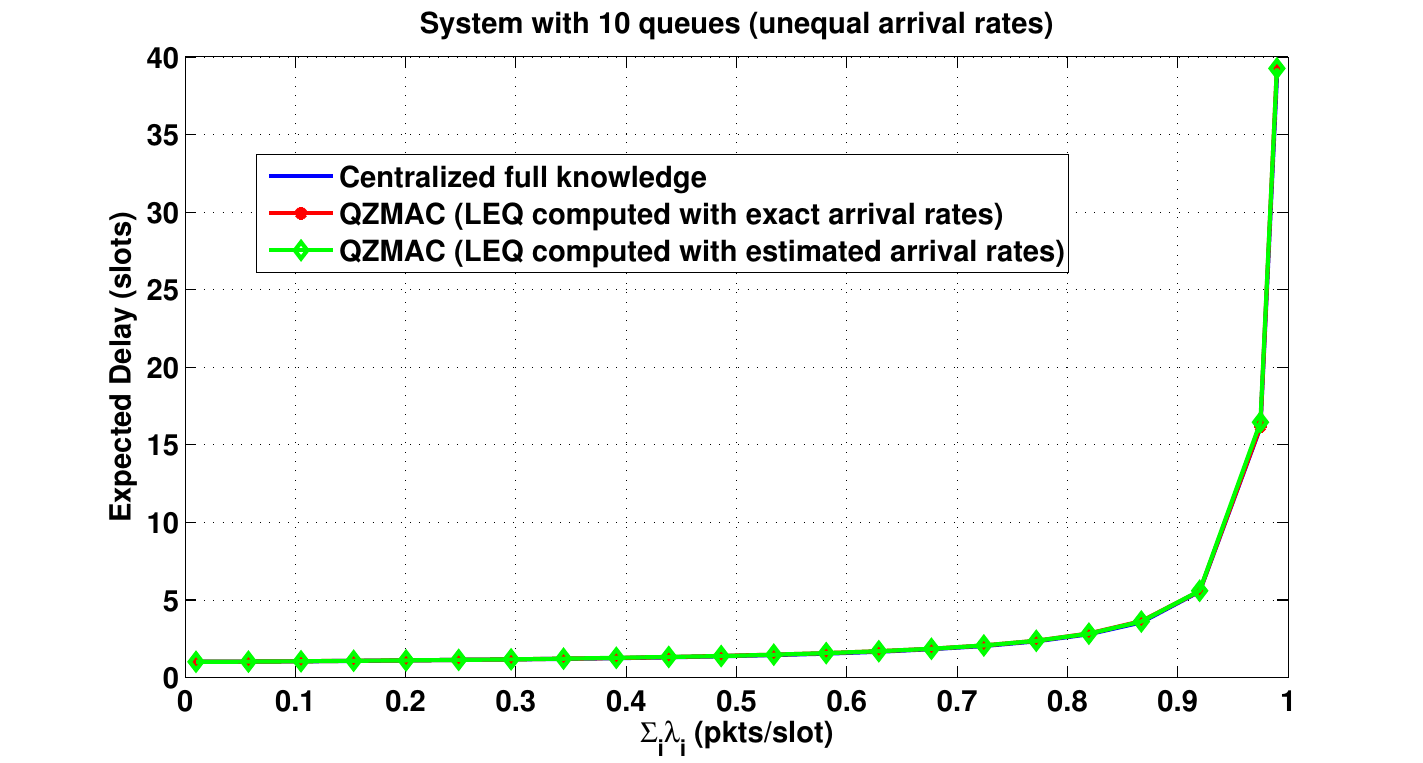}
\caption{System with 10 queues.}
\label{figQzmac10queues}
\end{subfigure}
\begin{subfigure}{.5\textwidth}
\centering
\includegraphics[scale = 0.3]{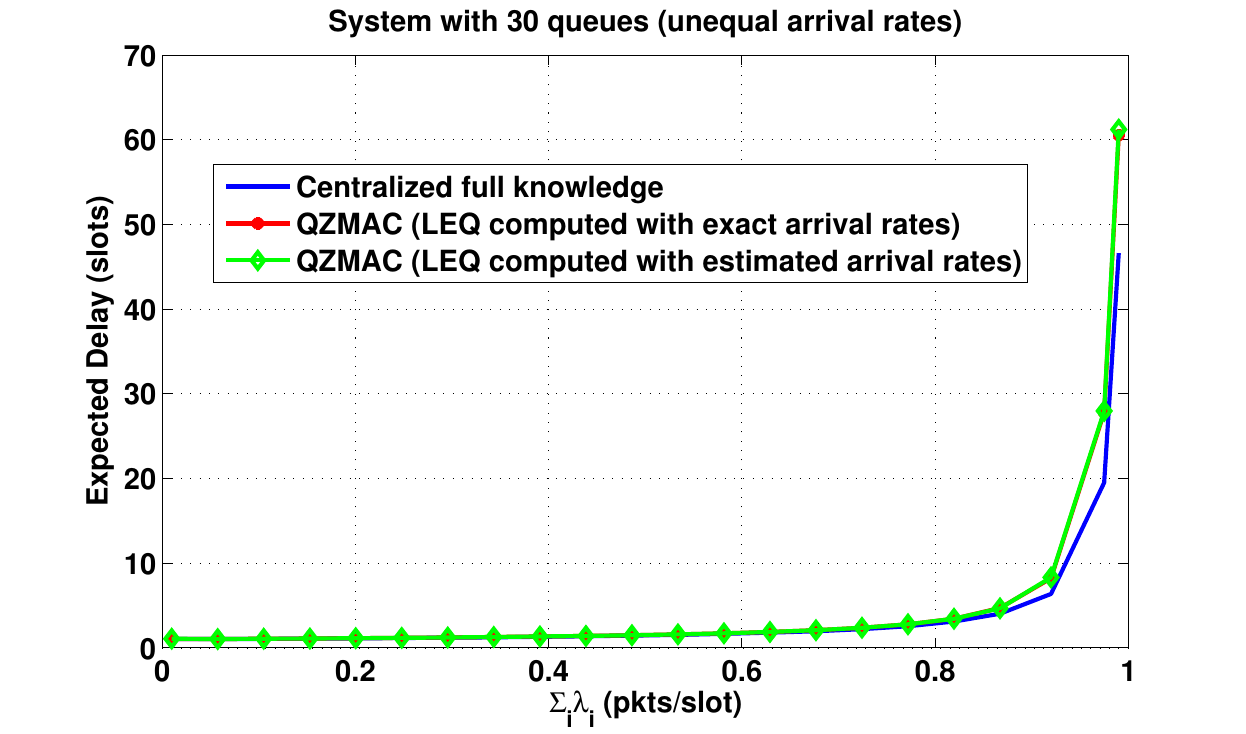}
\caption{System with 30 queues.}
\label{figQzmac30queues}
\end{subfigure}
\caption{Performance of QZMAC with exact and estimated arrival rates ($T_p=3$ and $T_c=7$).}
\vspace{-4mm}
\end{figure*}

\vspace{-4mm}
\subsection{Mean Delay Performance of  QZMAC} \label{secSymmetricSystemsSimulations}
With Bernoulli arrivals to each queue, the centralized full knowledge scheduler converts the system into a $Geo^{[x]}/D/1$ queue. In every slot, the arrivals to this queue are distributed according to a Binomial$(N,\lambda)$ distribution, where $\lambda\in[0,1/N)$ is the common arrival rate to all queues. This queue can be analyzed for expected delay, and we get
$W(\lambda)=\frac{2-(N+1)\lambda}{2(1-N\lambda)}$. 
 
We first consider a system with equal arrival rates and compare the performance of QZMAC and EZMAC with that
of ZMAC and the $Geo^{[X]}/D/1$ queue ($W(\lambda)$). We see from 
figures \ref{figDelAll4Qs10} and \ref{figDelAll4Qs30}, that the delay achieved by QZMAC is very close
to optimal as is that achieved by EZMAC (although slightly worse than QZMAC). In Fig.~\ref{figDelAll4Qs10}, the $Geo^{[x]}/D/1$ queue
delay cannot be seen explicitly since QZMAC performs almost exactly like it. This is quite 
encouraging since QZMAC only needs knowledge of the time since a queue was served last, while the optimal
algorithm needs full queue-length knowledge at all times. With $30$ queues, in fact, QZMAC 
shows a decrease in delay over ZMAC of more than $60\%$ and EZMAC of more than $40\%$.
 \begin{figure*}[htb]
 \begin{subfigure}{.5\textwidth}
 \centering
 \includegraphics[scale = 0.35]{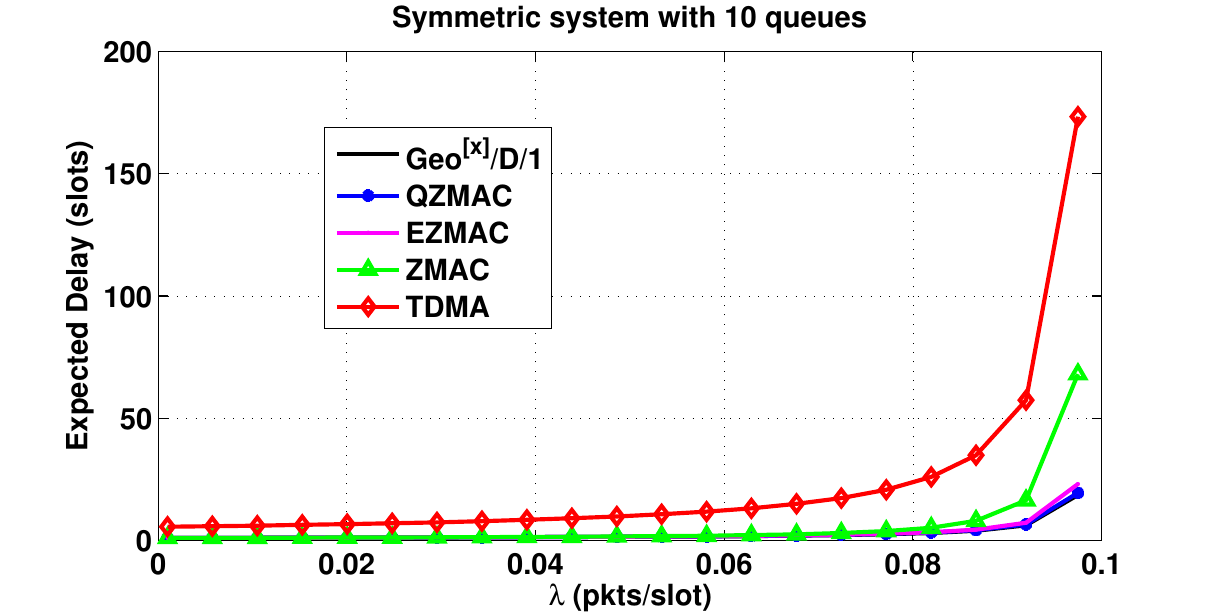}
\caption{System with 10 queues.}
 \label{figDelAll4Qs10} 
 \end{subfigure} 
\begin{subfigure}{.5\textwidth}
 \centering
 \includegraphics[scale = 0.35]{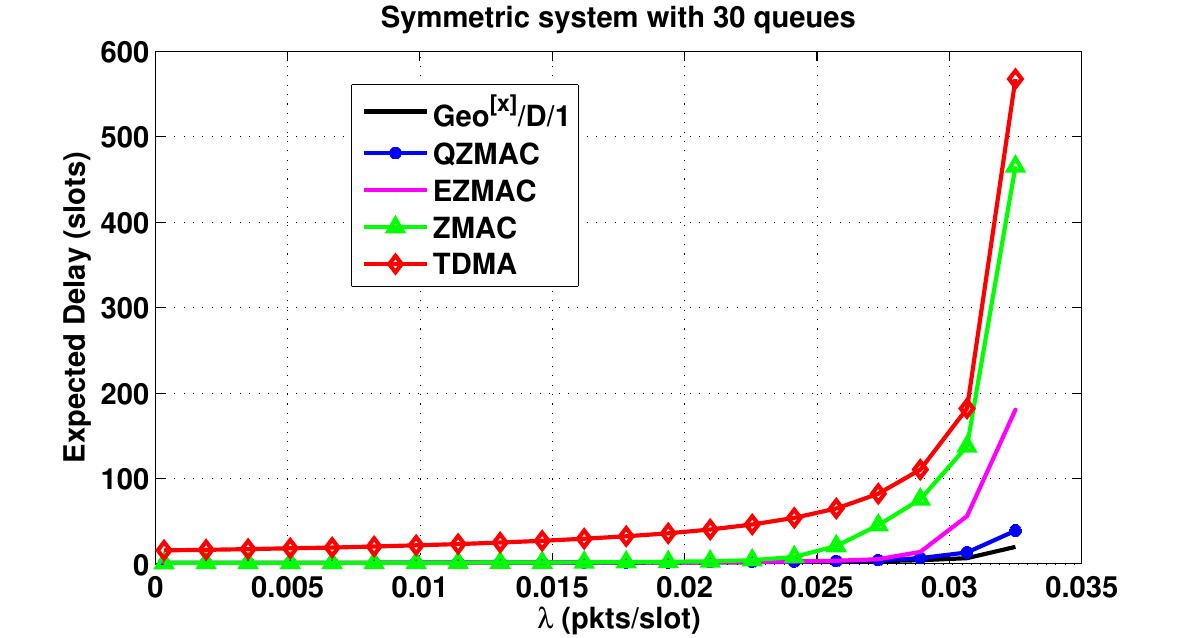}
 \caption{System with 30 queues.}
 \label{figDelAll4Qs30}
 \end{subfigure}
\caption{Expected delay with QZMAC, ZMAC and EZMAC. The systems have 10 queues \ref{figDelAll4Qs10} and 30 queues \ref{figDelAll4Qs30}. $T_{c}=9$ for ZMAC and $8$ for EZMAC. For QZMAC, $T_{c}=7$.}
\vspace{-7mm}
\end{figure*}

Note that since QZMAC requires $T_p=3$, in order to keep the scheduling portion's length $T_c+T_p$ constant
so as to maintain the scheduling portion to transmission portion 
(see Fig.~\ref{figMinislotStructure}) fraction uniform across protocols, we have reduced the 
contention window size by three minislots while simulating QZMAC. Also, as the first figure shows, 
for low and moderate 
arrival rates, in small WSNs \emph{one can use EZMAC and still achieve good delay performance.}\\
\indent We computed the empirical cumulative distribution functions (CDFs) of system backlog with all three protocols (and basic TDMA as an upper bound). 
As Fig.~\ref{figQueueLengthCDFs} clearly shows, QZMAC indeed provides the \emph{stochastically} smallest sum queue lengths. 
With QZMAC, the CDF hits $1$ at $12$ packets, while its closest competitor, EZMAC's CDF has a support that extends until $69$ packets and that of ZMAC extends to $220$ packets. 
\begin{figure}[tb]
\centering
\includegraphics[scale = 0.3]{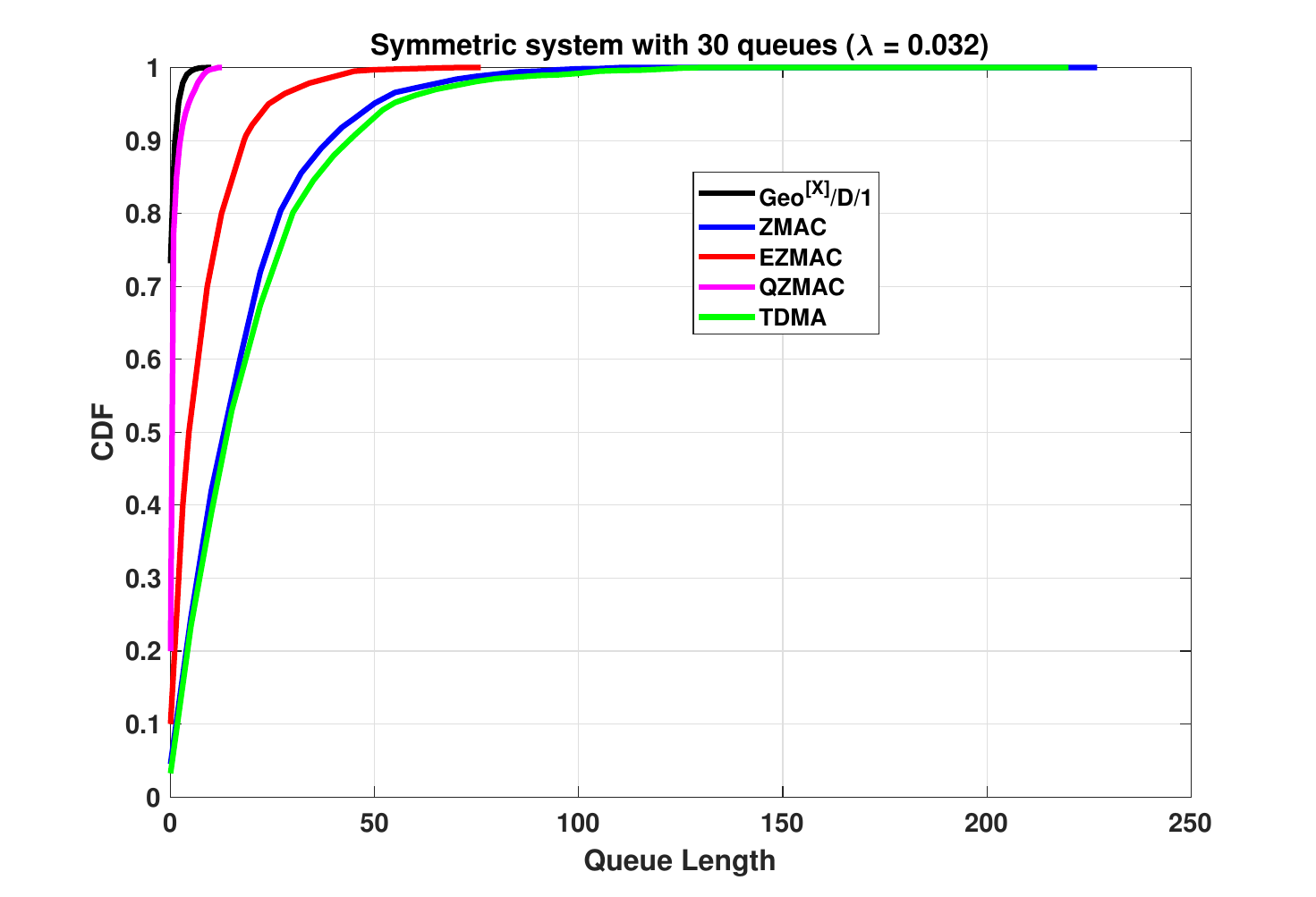}
\caption{The CDFs of total system backlog with QZMAC, EZMAC, ZMAC, TDMA and the $Geo^{[x]}/D/1$ queue. }
\label{figQueueLengthCDFs}
\vspace{-7mm}
\end{figure}
\vspace{-3mm}
\subsection{Tuning QZMAC}
One can tune QZMAC by changing $T_p$ and $T_c$ keeping $T_p+T_c=c$, a constant. As the number of slots 
reserved for polling ($T_p$) increases,
Fig.~\ref{figDelQzTune} shows that QZMAC with small $T_p$ initially performs better, but at high 
loads, is overtaken by QZMAC with small $T_c.$
At arrival rates close to $0$ (i.e., $\lambda\approx0$), it is an SU that transmits during most slots. 
This is because, with high probability, the PU's are empty and the protocol enters the contention portion.
Clearly, higher the value of $T_c$, greater the probability of resolving this contention. But at high
loads ($\lambda\approx\frac{1}{N}$) the PU's are nonempty with high probability and hence, the polling
portion is more likely to yield a nonempty queue. With a higher value of $T_p$, a nonempty queue can be 
found and the system does not have to go into contention.
This trend is expected and opens up avenues for further research about protocols that automatically choose 
the \emph{right} $T_c$ and $T_p$.
\\
\begin{figure}[tb]
\centering
\includegraphics[scale = 0.3]{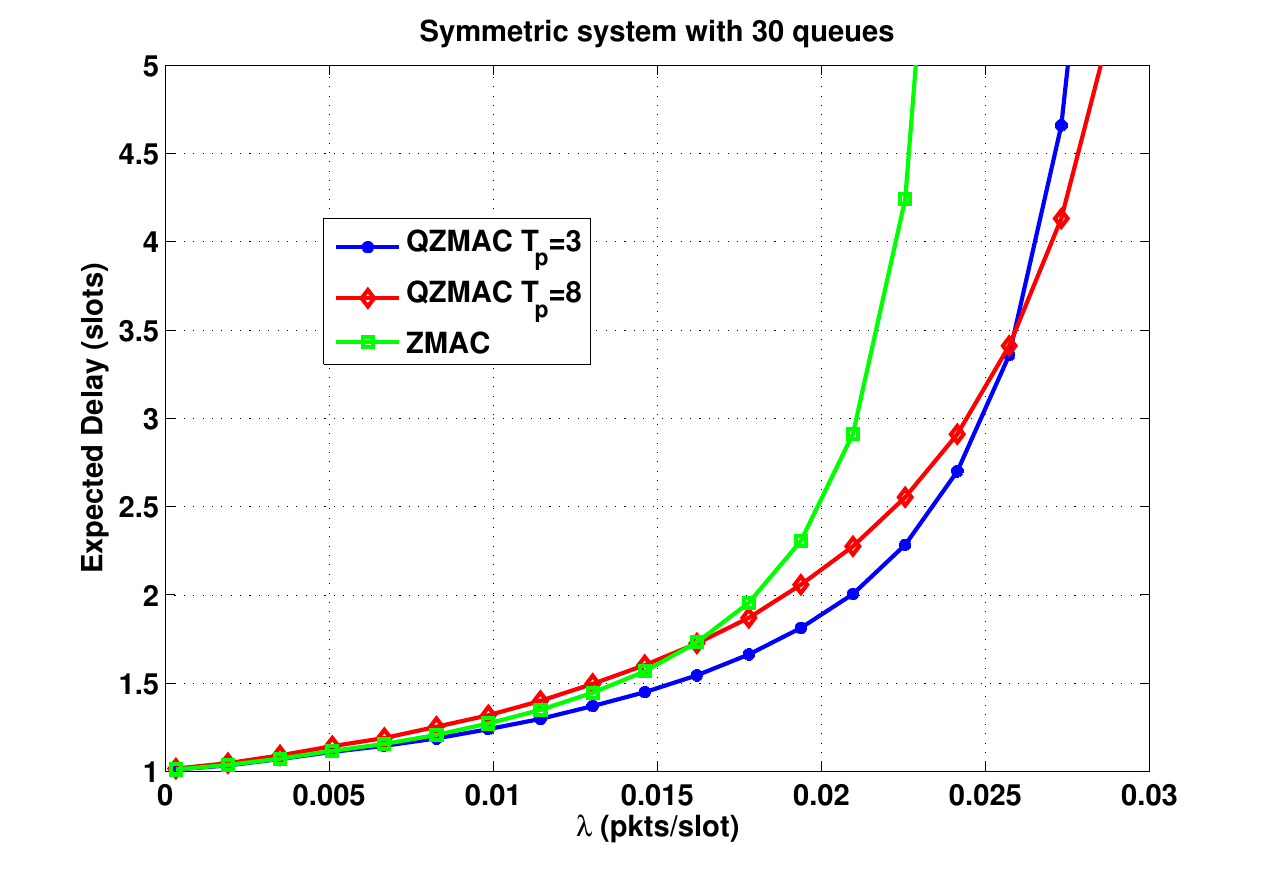}
\caption{Figure showing how QZMAC can be tuned. Keeping $T_p+T_c$ fixed at $10$ minislots, increasing
$T_p$ causes delay to increase at light loads and decrease at heavy loads. }
\label{figDelQzTune}
\vspace{-7mm}
\end{figure}
\vspace{-7mm}
\subsection{Channel Utilization}\label{appendix:exptsChannelUAndDelay}
Recall from our earlier discussion (see Prop.~\ref{propNonIdlingAndExhOpt}) that the optimal scheduling policy under the information structure assumed in this paper lies in the class of \emph{non-idling, exhaustive} policies. Consequently, the more efficiently a MAC algorithm finds nonempty queues in the network (assuming there are packets in the network), the better its mean delay performance is likely to be. In the literature, this probability of a nonempty queue being served when one exists in the network is termed \emph{\color{blue}Channel Utilization,} ($\zeta$) and is a performance metric commonly used to compare the efficiency of MAC algorithms \cite{rhee08zmac,warrier2005stochastic}. Formally, given a scheduling policy $\pi,$ 
\vspace{-5mm}

{
\small
\begin{eqnarray*}
 \zeta^{\pi}(\mathbf{Q}(0)) &:=& \lim_{t\rightarrow\infty}\mathbb{P}^{\pi}_{\mathbf{Q}(0)} \left(\sum_{j=1}^N \mathbb{I}_{\{Q_j(t)>0\}}D_j(t)>0\bigg|\sum_{j=1}^N Q_j(t)>0\right)\nonumber\\
 &\stackrel{(*)}{=}& \lim_{t\rightarrow\infty} \frac{\sum_{s=0}^{t-1}\sum_{j=1}^N \mathbb{I}_{\{Q_j(s)>0\}}\mathbb{I}_{\{D_j(s)>0\}}}{\sum_{s=0}^{t-1}\mathbb{I}_{\{\sum_{j=1}^N Q_j(s)>0\}}},
\end{eqnarray*}
}

where the $(*)$ is under the assumption that the queue length process is ergodic.
%
%
We used the setup described in Sec.~\ref{sec:Experiments}, to compare the channel utilization of QZMAC with that of ZMAC. 
The network comprised 7 collocated nodes transmitting to a base station. Packet arrivals to the nodes followed i.i.d Bernoulli processes with rates $\boldsymbol{\lambda}=[0.17,0.20,0.04,0.17,0.17,0.02,0.07]$. This rate vector is clearly within the network capacity region, because $\sum_{i=1}^7\lambda_i=0.84<1$. The results of the experiment are shown in Table.~\ref{table:channel-utilization-comparison}. The experiment was repeated for different values of control overhead, i.e., the portion of the time slot wasted in scheduling a queue ($T_p+T_c$). The value of $T_p$ was kept constant, $T_p=1$ for ZMAC and $T_p=3$ for QZMAC, and the number of contention minislots was varied. 

One obvious trend is that channel utilization increases as control overhead increases, since it becomes easier to resolve contention whenever it occurs with more $T_c$ minislots. However, the other point to note is that regardless of contention overhead, QZMAC outperforms ZMAC. This is due to the fact that both the polling and the contention mechanisms of the former are designed better than the latter to find nonempty queues with greater probability.

\begin{table}[t]
  \centering
  \begin{tabular}{c|c|c|c}
  Algorithm & $T_p+T_c=7$ & $T_p+T_c=8$ & $T_p+T_c=9$\\
  \specialrule{2pt}{2pt}{0pt}
    ZMAC ($T_p=1$) & 0.88968 & 0.90379 & 0.91356 \\
    \hline
    QZMAC ($T_p=3$) & 0.96312 & 0.9706 & 0.97486 \\
    \specialrule{2pt}{2pt}{0pt}
  \end{tabular}
  \caption{Comparing the channel utilization of QZMAC with that of ZMAC with varying control overhead ($T_p+T_c$).
  Clearly, $\zeta^{QZMAC}>\zeta^{ZMAC}$, showing that the former wastes the channel less often, thereby clearing the network of packets more efficiently.}
  \label{table:channel-utilization-comparison}
  \vspace{-0.500cm}
\end{table}

\section{Conclusion and Future work}
In this paper, we first derived optimal polling policies for systems with limited information structures and proved the delay-optimality of exhaustive service and cyclic polling in a symmetric version of our scheduling problem. Leveraging these results, we proposed two distributed protocols QZMAC and EZMAC that perform much better than those available in the literature, both with respect to mean delay and system backlog distributions. 
We then implemented QZMAC on a test bed comprising telosB motes and demonstrated the operation of several salient aspects of the protocol. 
\indent Extensions to this work can include the inclusion of Markovian arrivals, channel fading, extension to noncollocated nodes, and sleep-wake cycling sensor nodes.
\vspace{-2mm}
\section{Appendix}
\subsection{Experiments on an Implementation of QZMAC}\label{sec:Experiments}
We implement the QZMAC algorithm as an additional module in the MAC layer of the 6TiSCH communication stack under the Contiki operating system \cite{contiki-ref}. With the aid of in-built Contiki-based APIs, we implemented different wireless transceiver functionalities and carried out our experiments over Channel~15 of the 2.4 GHz ISM band. We deployed CC2420 based telosB motes placed equidistant from the receiver node on a circular table Fig.~\ref{fig:qzmac-testbed}. Here, the node placed in the center acts as a \enquote{Border Router} (BR), more commonly known as a \enquote{Sink.} The BR is always connected to a PC (host) through a USB cable; the BR collects the data packets from the sensor nodes and sends them to the host, which can further route them to the data processing computer over the Internet. Nodes labeled $1$ to $4$ in Fig.~\ref{fig:qzmac-testbed} are the sensor nodes upon which QZMAC runs.
\begin{figure}[tb]
    \centering 
    \includegraphics[scale = 0.03]{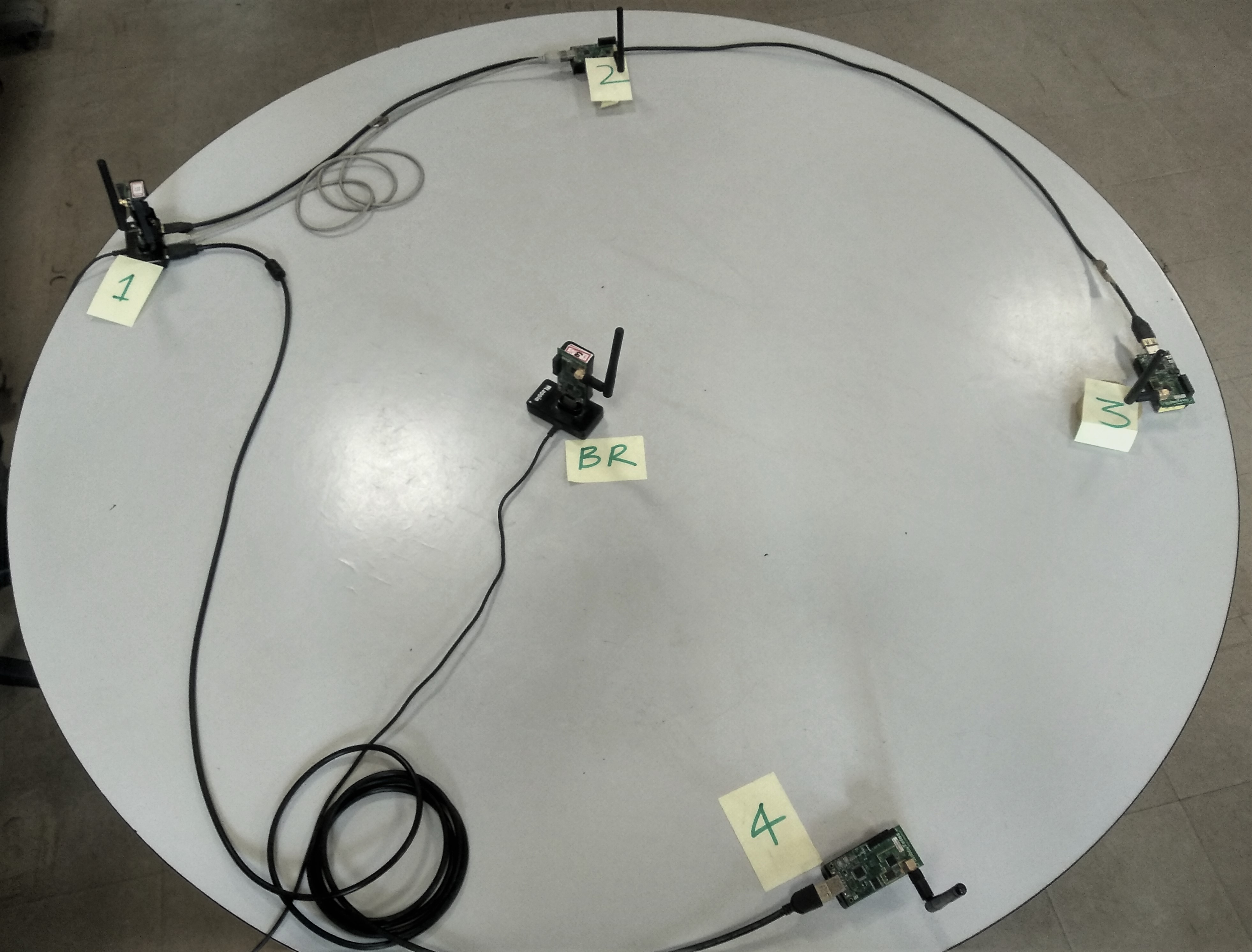}
    \caption{The center node is the Border Router (sink) and the other nodes, numbered 1 to 4, are the sensor nodes. The sensors are programmed to run QZMAC and transmit packets to the sink.}
    \label{fig:qzmac-testbed}
    \vspace{-4mm}
\end{figure}
Our experiments aim to study the following aspects:
  \textbf{\color{blue}(a)} The network infers and maintain the status of $\mathbf{V}$ vector based on empty and nonempty status of sensor nodes,
  \textbf{\color{blue}(b)} Verify the protocol working in contention mode as described between step~\ref{qzmacProtocol:incont} and step~\ref{qzmacProtocol:outcont} of the QZMAC,
  \textbf{\color{blue}(c)} CCA status inference across the slots, and
  \textbf{\color{blue}(d)} Synchronization within the network. 
\vspace{-4mm}
\subsection{Frame Structure: Time Slots and Mini-Slots}\label{secFrameStructureTimeSlotsAndMinislots}
Time Slotted Channel Hopping (TSCH) is a MAC layer specified in the IEEE 802.15.4-2015 \cite{ieee-tsch-std} standard, with a design inherited from WirelessHART and ISA 100.11a standards. In our implementation, we used the slots defined as part of TSCH within Contiki (sans the channel hopping utility). The duration of a time slot is $15ms$ which is sufficiently long for the transmitter to transmit the longest possible packet and for the receiver to return an acknowledgment.\\
\indent The slots are further divided into polling minislots and contention minislots as shown in Fig.~\ref{figMinislotStructure}. The minislots are of two standard Clear Channel Assessment (CCA) duration, where $1$ CCA duration is $128$ microseconds. 
In our experiments, we used 
$T_p=3$, and $T_c=9$.
\vspace{-4mm}
\subsection{Time Synchronization}
Maintaining network-wide time synchronization at the level of the minislots is a nontrivial problem. For our experiments, we used the Adaptive Synchronization Technique \cite{chang2015adaptive}. The Border Router periodically broadcasts Enhanced Beacons (EBs) containing a field indicating the current slot number also known as Absolute Slot Number (ASN). The other nodes store the ASN value and increment it every slot to keep the time slot number aligned. To align the time slot boundaries, the first mini-slot begins after a guard time offset of {\tt TsTxOffset} from the leading edge of every slot. Every node timestamps the instant it starts receiving the EB, and then aligns its internal timers so that its slot starts exactly {\tt TsTxOffset} before the reception of EB. In our implementation we used {\tt TsTxOffset}=$1.8$ms. Through extensive experimentation, we verified that clock drifts were not affecting the protocol's working across the minislots.
The nodes form a $1$-hop fully connected network. We used the \enquote{Routing Protocol for Low-Power and Lossy Links} (RPL) to form the routes within the network \cite{winter-etal12RPL} and verified the working of our firmware on the COOJA simulator \cite{cooja-ref}, before compiling it on to real target motes. We now describe our methodology in detail.

\vspace{-4mm}
\subsection{CCA Errors : Inference and Handling}\label{sec:ccaErrorsInferenceAndHandling}
To study the effect of CCA errors on our protocol, we set up two testbeds (similar to the one in Fig.~\ref{fig:qzmac-testbed}) with network diameter $6$ meters and $8$ meters respectively, each running QZMAC for a period of 12 hours (i.e., $3\times10^6$ time slots). The sensor nodes were programmed to generate packet at constant rate ensuring the queue at sensor nodes are always nonempty and transmit as per the protocol. One of the sensor nodes was connected to a terminal and designated to monitor the CCA status during the experiment. Over the course of our experiment, we observed $1$ CCA \enquote{Miss} error on the testbed with $6$ m network diameter and, on the network with $8$ m diameter, we recorded $3$ CCA \enquote{Miss} errors, both calculated over a period of 12 hours. No False Alarms were observed. Further, in any given time-slot {\color{blue}at most one} CCA error occurred network-wide. 
\vspace{-5mm}

\bibliographystyle{IEEEtran}
\bibliography{IEEEabrv,techreport}

%

\vspace{-1.50cm}
\begin{IEEEbiographynophoto}{Avinash Mohan (S.M. '16, M '17)} 
obtained his MTech from the Indian Institute of Technology (IIT) Madras, and PhD from and the Indian Institute of Science (IISc) Bangalore, in 2010 and 2018 respectively. He was a postdoctoral fellow with the Reinforcement Learning Research Labs ($(RL)^2$) at the Technion, Israel Institute of Technology, Haifa, Israel and is currently with The Boston University, Massachusetts, USA. His research interests include reinforcement learning, stochastic control, analysis of deregulated energy markets and resource allocation in wireless communication networks.
\end{IEEEbiographynophoto}
\vspace{-1.00cm}
\begin{IEEEbiographynophoto}{Arpan Chattopadhyay}
 obtained his B.E. in Electronics and Telecommunication Engineering from Jadavpur University, India in 2008, and M.E. and Ph.D in Telecommunication Engineering from Indian Institute of Science, Bangalore, India in  2010 and 2015, respectively. 
He is currently working as an Assistant Professor with the Electrical Engineering Department, IIT Delhi. Previously, he held postdoctoral positions with the Electrical Engineering Department, University of Southern California, and INRIA/ENS Paris. His research interests include wireless communication and networks, cyber-physical systems, networked estimation and control, and reinforcement learning.

\end{IEEEbiographynophoto}
\vspace{-1.00cm}
\begin{IEEEbiographynophoto}{Shivam Vinayak Vatsa (B Tech. '16)} 
obtained Bachelor of Technology from NIIT University, India in Computer Science. He is currently a Project Associate II at Center for Networked Intelligence, Indian Institute of Science, Bangalore, India. In past, he worked as Software Engineer in Common Algorithm Development Group at ABB India. His research interests include Internet of Things, Cyber Physical Systems and  Data analysis of wireless communication networks. 
\end{IEEEbiographynophoto}
\vspace{-1.00cm}
\begin{IEEEbiographynophoto}{Anurag Kumar}
 (B.Tech., Indian Institute of Technology (IIT)
Kanpur; PhD, Cornell University, both in Electrical Engineering) was
with Bell Labs, Holmdel, N.J., for over 6 years.  Since then he has
been on the faculty of the ECE Department at the Indian Institute of
Science (IISc), Bangalore; he was the Director of the Institute during
2014-2020, and now holds an emeritus position.  His area of research
has been communication networking, and he has recently focused
primarily on wireless networking. He is a Fellow of the IEEE, the
Indian National Science Academy (INSA), the Indian National Academy of
Engineering (INAE), the Indian Academy of Sciences (IASc), and The
World Academy of Sciences (TWAS).  He was an associate editor of IEEE
Transactions on Networking, and of IEEE Communications Surveys and
Tutorials.
\end{IEEEbiographynophoto}

\newpage
\section{Supplementary Material}
\subsection{Glossary of Notation and Acronyms}
\begin{enumerate}
    \item $A_j(t)$: number of packets arriving to Queue~$j$ in time slot $t.$
    \item $B(V_j)$: If $A$ is a generic Bernoulli($\lambda$) random variable, and $C$ is distributed Binomial($V_j,\lambda$), then $B(V_j)$ is a random variable whose distribution is the same as that of $(C-1)^++A$.
    \item: Carrier Sense Multiple Access
    \item CCA: Clear Channel Assessment.
    \item $D_j(t)$: the number of departures from Queue~$j$ in time slot $t.$
    \item DRAND: Distributed RAND (distributed implementation of the RAND protocol)
    \item $Geo^{[x]}/D/1$ queue: has arrivals that are sums of Bernoulli random variables (here, $N$), and has a single server (hence, the \enquote{$1$}), whose service times are \emph{deterministic} (hence, the \enquote{$D$}).
    \item $H^\pi_t$: the history of policy $\pi$ up to time $t$ as defined in \eqref{eqn:DefnHistoryOfPolicyPi}.
    \item $\mathbf{Q}(t)$: backlog of Queue~$j$ at the beginning of time slot $t.$
    \item $Q^{\pi}_j(t)$: backlog of Queue~$j$ at the beginning of time slot $t$ under scheduling policy $\pi.$
    \item $K_{thr}$: the threshold that triggers the RESET subroutine. It is denoted by {\tt THRSLD} in the subroutine description.
    \item $\boldsymbol{\lambda}$: the arrival rate vector.
    \item $\boldsymbol{\Lambda}$: the capacity region of the queueing network.
    \item $\boldsymbol{\Lambda}_{LEQ}$: the set of arrival rates that the LEQ policy can stabilize.
    \item $[N]$: the set of integers $\{1,2,\cdots,N\}$.
    \item {\tt NDST:} is short for \enquote{Node State.} Attains values {\tt COLL} (meaning \enquote{in collision}) or {\tt NOCOLL} (meaning \enquote{not in collision}).
    \item PU: Primary User.
    \item $\Pi$: the space of all admissible policies.
    \item $\Pi_e$: the subset of all exhaustive service policies. 
    \item $\Pi_g$: the subset of all non-idling service policies.    
    \item {\tt RSTBCN}: Reset Beacon.
    \item SU: Secondary User. In ZMAC this refers to any queue which is not the current PU. In EZMAC and QZMAC, this refers to the queue that won the latest contention.
    \item $T_c$: number of minislots reserved for contention.
    \item TDMA: Time Division Multiple Access
    \item $T_p$: number of minislots reserved for polling.
    \item $V^\pi_j(t)$: is the number of slots prior to slot $t$ in which Queue~$j$ was allowed to transmit under a generic scheduling policy $\pi$.
    
\end{enumerate}

\subsection{Terminology Related to MAC Protocol Analysis}
We provide brief explanation for some standard terminology related to the analysis of MAC protocols. For further  details, the reader can refer to \enquote{Multiple access protocols: performance and analysis,} by R. Rom and M. Sidi, Springer-Verlag New York, 1990.
\begin{itemize}
    \item Backlog: at time $t+,$ i.e., the beginning of slot $t,$ the number of packets in Queue~$j$ is called its \enquote{backlog,} and is denoted by $Q_j(t).$
    \item Stochastic dominance: given two random variables $X$ and $Y,$ we say that $X$ \enquote{stochastically dominates} $Y$ if, $\forall q\in(-\infty,\infty),$ $P(X\leq q)\leq P(Y\leq q).$
        
    \item Stochastic ordering: two random variables $X$ and $Y$ are said to be \enquote{stochastically ordered}, if either $X$ stochastically dominates $Y,$ or vice versa.
    \item Switchover delay: informally, this is the amount of time the server takes to stop service at one queue and begin service at a different queue.
    \item Overhearing: when wireless devices communicate using a shared medium, such as in a wireless network, due to the nature of wireless communication, it is possible for nearby devices to receive signals that are intended for other devices. This unintentional reception of signals is known as overhearing. 
    \\
    Wireless transmissions are typically encrypted to prevent decoding through overhearing. In our paper, however, we take advantage of overhearing to simply detect the presence of an ongoing transmission (the eavesdropping device does not attempt to decode anything).
\end{itemize}


\subsection{Proof of Prop.~\ref{propNonIdlingAndExhOpt}}\label{AppendixProofOfPropNonIdlingAndExhOpt}
This proof proceeds along the lines of the proof in \cite[Prop.~4.2]{liu-nain92optimal-polling}. It is easy to show that we can restrict 
attention to non-idling policies and so, we begin our proof with some policy $\pi\in\Pi_g$, where, as mentioned in Sec.~\ref{secTheCentralizedScheduler}, $\Pi_g$ is the set of all non-idling policies. Consider a sample path $\omega$ $\left(\right.$a realization of the input sequences
$\left.\left\lbrace A_{i}(n)\right\rbrace_{n=1}^\infty,i\in I\right)$ and on that path, let
\begin{equation}
 n^*:=\inf\left\lbrace n\geq1: Q^\pi_{\pi^q_{n-1}}(n)>0,~Q^\pi_{\pi^q_{n}}(n)=0,~\text{and}~\pi^a_n=2\right\rbrace.
 \label{defineNStar}
\end{equation}
This is the first instant when a time slot is wasted by $\pi$ on account of switching to an empty queue
though the previous one was non-empty. 
Further define 
\begin{equation}
m^*=\inf\left\lbrace n>n^*: Q^\pi_{\pi^q_{n-1}}(n)>0,\pi^q_{n-1}=\pi^q_{n^*},~\text{and}~\pi^a_n=1 \right\rbrace.  
\end{equation}
This is the first time since $n^*$ that the server continues to serve the (same) incumbent finding it to be 
non-empty. 
\begin{figure*}[tb]
\centering
\includegraphics[height=8.0cm, width=15cm]{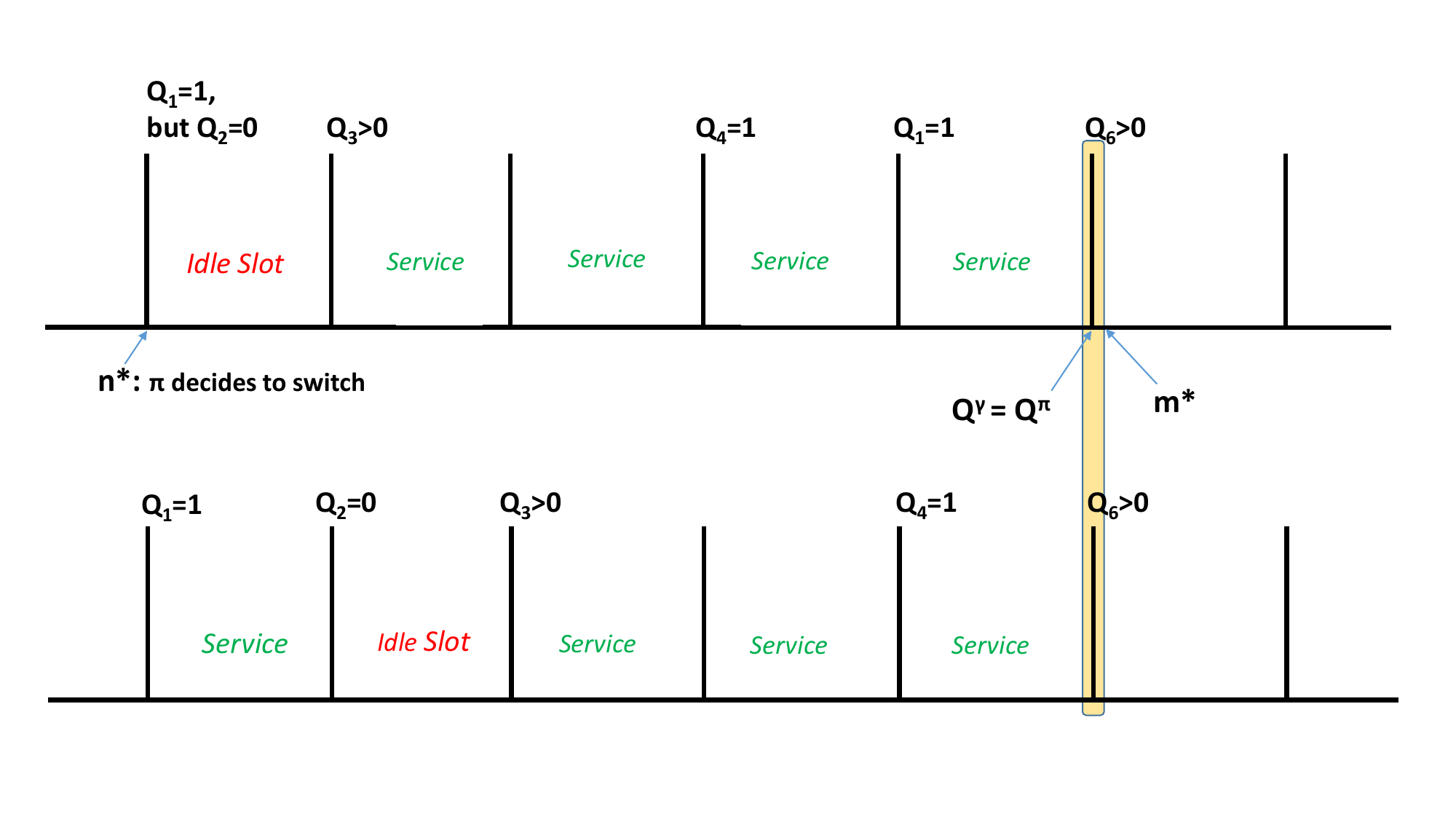}
\caption{A sample path illustrating how $\gamma$ is derived from $\pi$.}
\label{figPiGamSmplPath}
\end{figure*}
Construct another policy $\gamma$ as follows (refer to Fig.\ref{figPiGamSmplPath} for an illustration).
\begin{itemize}
 \item On $[1,n^*)$ and $[m^*+1,\infty)$, $(\gamma^a_k,\gamma^q_k)=(\pi^a_k,\pi^q_k)$. 
 \item $(\gamma^a_{n^*},\gamma^q_{n^*})=(1,\pi^q_{(n^*-1)})$, and 
 \item for $n^*+1\leq k\leq m^*$, $(\gamma^a_k,\gamma^q_k)=(\pi^a_{k-1},\pi^q_{k-1})$. 
\end{itemize}
So, $\gamma$ follows $\pi$ over $[1,n^*)$ and $[m^*+1,\infty)$ and at $n^*$, $\gamma$ deviates
from $\pi$ and serves the incumbent (possible, since that queue is non-empty). Between $n^*+1$ and $m^*$, 
$\gamma$ follows $\pi$ albeit with one slot delay, i.e., it does whatever $\pi$ did one slot ago.
Clearly,
\begin{eqnarray}
 Q^\gamma(t)=Q^\pi(t),~\forall~t\in[1,n^*)\cup[m^*+1,\infty),~\text{and}\nonumber\\
 Q^\gamma(t)\leq Q^\pi(t),~\forall~t\in[n^*+1,\leq m^*],\nonumber\hspace{1cm}
\end{eqnarray}
whenever $m^*<\infty$. This relation is true for every sample path $\omega$. 
$m^*$ is set equal to $\infty$ whenever either $n^*=\infty$, \emph{or}, $\gamma$ is able to
serve a packet in some slot $l$ ($\geq n^*+1$) while $\pi$ could not serve any packet
in slot $l-1$. 

To explain the latter case further, observe that after $n^*$, $\gamma$ does whatever $\pi$ did
one slot ago. If $\pi$ switched to some empty queue in $l-1$, $\gamma$ switches to the same 
queue in $l$. But if $\pi^q_{l-1}$ received a packet in $l-1$, $\gamma$ will
be able to serve that packet increasing $Q^\pi(t)-Q^\xi(t)$ to 2 (one at $n^*$ and the second 
at $l$). Since $\gamma$ follows $\pi$ with exactly one step delay, $\pi$ will never be able to 
make up the difference! Hence, $m^*=\infty$.

To complete the proof, one needs to consider epochs $l^*$ such as:
\begin{eqnarray}
l^*&:=&\inf\{n\geq1: Q^\pi_{\pi^q_{n-1}}(n)>0,Q^\pi_{\pi^q_{n}}(n)>0,\nonumber\\
&&\text{and}~\pi^a_n=2\}, 
\end{eqnarray}
wherein the queue which $\pi$ decides to serve, unlike the case with $n^*$, is \emph{non-empty} and hence,
slot $l^*$ is not wasted. Even in this case defining $\gamma$ as before produces a policy that does no worse 
than $\pi$ in terms of total system backlog.  

In the same manner, find such an $n^*$ (or $l^*$ as the case may be) for $\gamma$ and 
refine this policy using the procedure described. Iterating this procedure, one ends up with a policy
say $\xi$, for which such an instant never exists, i.e., which never switches to a different queue when the incumbent is nonempty. This is by definition exhaustive. We have, thus, defined both queueing processes on some common probability space $\left(\Omega,\mathcal{F},\mathbb{P}\right)$ and for
every $\omega\in\Omega$, we have shown that 
\begin{equation}
Q^\xi(t,\omega)\leq Q^\pi(t,\omega),~\forall~t\geq1,\nonumber 
\end{equation}
which means that $\forall x\in\mathbb{R}$,
\begin{equation}
 \mathbb{P}\{\omega:Q^\pi(t,\omega)>x\}\geq\mathbb{P}\{\omega:Q^\xi(t,\omega)>x\}.
\end{equation}
This coupling argument shows that 
\begin{equation}
Q^\xi(t)\stackrel{st}{\leq} Q^\pi(t),~\forall~t\geq1.
\end{equation}

\IEEEQED

\subsection{Longest Expected Queue (LEQ) and Stochastically Largest Queue (SLQ) Policies}
\label{secLEQunequal}
\label{AppendixProofOfPropSLEQunequal}
Let the arrival rate vector be denoted by $\boldsymbol{\lambda}$ ($\in\mathbb{R}^N_+$). 
For notational convenience, in the sequel, we denote the number of slots since Queue~$i$ was last served, formerly denoted $V_{i}^\pi(t)$, by $V_i(t)$ (omitting the policy superscript). Since we consider exhaustive service, $\lambda_iV_i(t)$ is the expected backlog of queue $i$ at $t$, and $\argmax{1\leq i\leq N}\lambda_iV_i(t)$ is the longest expected backlog. We call the policy that, at every scheduling instant, schedules this queue the Longest Expected Queue (LEQ) policy.
\begin{tcolorbox}
\prop
The LEQ policy stabilizes all arrival rate vectors in the set
\begin{equation}
\Lambda_{LEQ}:=\bigg\{\boldsymbol{\lambda}\in\mathbb{R}^N_+\bigg|\sum_{i=1}^N\lambda_i<1~\text{and}~\min_{1\leq i\leq N}\lambda_i>0\bigg\}.
\end{equation}
\end{tcolorbox}

\begin{rem}
\begin{enumerate}
    \item \emph{Comparing the definition of the region $\Lambda_{LEQ}$ with that of the region $\Lambda$ in Eqn.~\ref{eqnBasicCapacityRegion} shows quite clearly, that the LEQ policy is throughput optimal.}
    \item \emph{With $\lambda_i = \lambda, \forall i$, we have and the policy is easily seen to be a \emph{cyclic} exhaustive policy, i.e., queues are served in a fixed cycle, with each queue being served to exhaustion.
    $\argmax{1\leq i\leq N}\lambda_iV_i(t)=\argmax{1\leq i\leq N}V_i(t)$ and the policy reduces to a {\color{blue}\emph{cyclic} exhaustive service} policy.}
\end{enumerate}

\textbf{LEQ vs SLQ. }Note that when $V_l\neq V_k$ and $\lambda_l\neq\lambda_k,$ it might not even be possible to stochastically order random variables distributed according to Binomial($V_l,\lambda_l$) and Binomial($V_k,\lambda_k$). An example is when $V_l>V_k$, but $\lambda_l<\lambda_k.$ Hence, with unequal arrival rates, the LEQ policy doesn't necessarily choose the stochastically longest queue, but only the queue with the largest \emph{mean} backlog. This is the difference between the LEQ and SLQ policies.
In the sequel, we denote the set of all SLQ policies by $\Pi_s$.
\end{rem}
Recall that under symmetry, the LEQ policy essentially schedules $\argmax{1\leq i\leq N}V_i(t)$ at every scheduling instant and exhaustively serves this queue. 
Why exhaustive service? As an example, 
Consider the scenario with 4 queues in the system and suppose that at the beginning of Queue $1$\textquotesingle s service, $\mathbf{V} = [0, v_2, v_3, v_4]$, with $v_2 > v_3 > v_4$. At the end of Queue $1$\textquotesingle s service, the vector changes to $\mathbf{V}' = [0, v_2', v_3', v_4']$, but the ordering is still preserved, i.e., $v'_2 > v'_3 > v'_4 > v'_1=0$. Hence, Queue 2 is chosen for service next. But at the end of queue 2\textquotesingle s service, once again, we see that $v''_3 > v''_4 > v''_1$, as it was at the beginning of Queue 2\textquotesingle s service, and Queue 3 is chosen, after which queue 4 is chosen followed by queue 1 and this process repeats. We see that under symmetry, LEQ induces a cyclic service system.  This policy is discussed in detail in Sec.~\ref{secDelayOptimalityInSymmetricSystems}. 
where we will, in fact, show it to be \emph{mean delay optimal} in symmetric systems. This also helps bolster our confidence in the LEQ policy itself. We now explore in greater detail the behavior of the LEQ policy in symmetric systems. 

To prove this we invoke Theorem~(3.1) in \cite{foss-last96stability-exhaustive-state-dependent} that proves the stability of certain exhaustive service policies with state dependent routing. We first require some more notation.

We denote by $S(t)$ the queue in service during slot $t$, and set $S(t)=0$ if the server is idling at some queue (which from Prop.~\ref{propNonIdlingAndExhOpt}, has to be empty) during some slot. We denote by $W_n$, the time taken by the server to begin the $n^{th}$ service. By this we mean the time taken by the server to find a non empty queue after serving some queue. So, if the incumbent is non empty, $W_n=0$ since the server simply begins serving the next packet in the queue served in the previous slot without switching away from it. Otherwise the server begins to search for a non empty queue to begin a busy period. Specifically, if $Q_{S(0)}(0)=0$, then $W_1$ is the time taken to find a non empty queue. During this period, the server may have visited several empty queues. 

Assume the underlying probability space is denoted by $\left(\Omega,\mathcal{F},P\right)$
and let 
\begin{equation}
\mathcal{F}_n:=\sigma\left(\bigg\{ \left[ Q_1(t),\cdots,~Q_N(t),~S(t)\right],~t\leq n\bigg\}\right),~n\geq0,\nonumber
\end{equation}
be the filtration describing the history of the system. Let $A_i[t]$ be the number of arrivals to queue $i$ until and including time $t$, and $A[t]=\sum_{i=1}^NA_i[t]$, the total number of arrivals to the system over the same duration.

The proof of Theorem~(3.1) in \cite{foss-last96stability-exhaustive-state-dependent} relies on two assumptions\footnote{Eqn.~(2.3) and Eqn.~(2.4) in \cite{foss-last96stability-exhaustive-state-dependent}.} that we show are satisfied in our case. Firstly, that there exists $w>0$ such that
\begin{equation}
\mathbb{E}\left[W_1|\mathcal{F}_0\right]<w,~P-a.s.~\text{on}~\left\lbrace Q_{S(0)}=0\right\rbrace,
\label{eqnSLQunequalAssumption1}
\end{equation}
and secondly that there exists some $p>0$ such that
\begin{equation}
P\left(A(W_1)=0|\mathcal{F}_0\right)>p,~P-a.s.~\text{on}~\left\lbrace Q_{S(0)}=0,\sum_{i=1}^NX_i(0)>0\right\rbrace.
\label{eqnSLQunequalAssumption2}
\end{equation}
In other words, \eqref{eqnSLQunequalAssumption1} means that if the system starts out empty, the time taken to find a non empty queue under the policy being considered, should have a finite mean. \eqref{eqnSLQunequalAssumption2}, on the other hand, refers to the fact that when the system starts off non empty and the server is at an empty queue at time $0$, the probability of $0$ arrivals in $W_1$ is positive. 
 
First note that on $\{Q_{S(0)}=0\}$, $W_1\geq1$, and let $\lambda_{min}=\min_{1\leq i\leq N}\lambda_i$ be the smallest arrival rate which, by assumption is strictly positive. The first assumption is satisfied, since 
\begin{eqnarray}
\mathbb{E}\left[W_1|\mathcal{F}_0\right]&=&\sum_{k=1}^\infty P\{W_1\geq k|\mathcal{F}_0\}\nonumber\\
&\stackrel{(*a)}{\leq}&\sum_{k=1}^\infty (1-\lambda_{S(k)})^k\nonumber\\
&\leq&\sum_{k=1}^\infty (1-\lambda_{min})^k\nonumber\\
&=&\frac{1}{\lambda_{min}}-1<\infty.\nonumber
\end{eqnarray}
Inequality $(*a)$ is true, since, for the walking time to be at least $k$, the $k^{th}$ polled queue, i.e., $S(k)$, must be empty. Since on $\{Q_{S(0)}=0\}$ the system also starts off empty, this probability is $(1-\lambda_{S(k)})^k$. To prove \eqref{eqnSLQunequalAssumption2}, first define  $p_k=P\{W_1=k|\mathcal{F}_0\}$. 
\begin{eqnarray*}
P\{A(W_1)|\mathcal{F}_0\}&=&\sum_{k=1}^\infty P\{A(k)=0|\mathcal{F}_0, W_1=k\} P\{W_1=k|\mathcal{F}_0\}\\
&=&\sum_{k=1}^\infty P\{A(k)=0|\mathcal{F}_0, W_1=k\} p_k.\\
&=&\sum_{k=1}^\infty \left(\sum_{i=1}^N\left(1-\lambda_i\right)^k\right) p_k.
\end{eqnarray*}
But on $\{X_{S(0)}=0\}$, $p_1>(1-\lambda_{min})~a.s.$ So,
\begin{eqnarray*}
P\{A(W_1)|\mathcal{F}_0\}&=&p_1\sum_{i=1}^N\left(1-\lambda_i\right)\\
&&+\sum_{k=2}^\infty \left(\sum_{i=1}^N\left(1-\lambda_i\right)^k\right) p_k\\
&\geq&\sum_{i=1}^N(1-\lambda_i)>0.
\end{eqnarray*}
With both \eqref{eqnSLQunequalAssumption1} and \eqref{eqnSLQunequalAssumption2} satisfied, we note that Theorem~(3.1) in \cite{foss-last96stability-exhaustive-state-dependent} is proved under a much more general model, where the behavior of the system is influenced by another process $U(t),t\geq0$, that takes values in some measurable space $\left(\mathbf{U},\mathcal{U}\right)$. Setting $\mathbf{U}=\mathbb{N}^N$ and $U(t)=V(t)$, the proof is complete.

\IEEEQED

\subsection{LEQ with Equal Arrival Rates}\label{appendix:leqInSymmetricSystems}

With equal arrival rates, $\argmax{1\leq i\leq N}\lambda_iV_i(t)=\argmax{1\leq i\leq N}V_i(t)$ and the policy reduces to a {\color{blue}\emph{cyclic} exhaustive service} policy, as can be seen from the following example.
Consider the scenario with 4 queues in the system and suppose that at the beginning of Queue $1$\textquotesingle s service, $\mathbf{V} = [0, v_2, v_3, v_4]$, with $v_2 > v_3 > v_4$. At the end of Queue $1$\textquotesingle s service, the vector changes to $\mathbf{V}' = [0, v_2', v_3', v_4']$, but the ordering is still preserved, i.e., $v'_2 > v'_3 > v'_4 > v'_1=0$. Hence, Queue 2 is chosen for service next. But at the end of queue 2\textquotesingle s service, once again, we see that $v''_3 > v''_4 > v''_1$, as it was at the beginning of Queue 2\textquotesingle s service, and Queue 3 is chosen, after which queue 4 is chosen followed by queue 1 and this process repeats. We see that under symmetry, LEQ induces a cyclic service system.  This policy is discussed in detail in Sec.~\ref{secDelayOptimalityInSymmetricSystems}. 
where we will, in fact, show it to be \emph{mean delay optimal} in symmetric systems. This also helps bolster our confidence in the LEQ policy itself. We now explore in greater detail the behavior of the LEQ policy in symmetric systems. 

 \subsection{Stability of Cyclic Exhaustive Service}\label{secStablePiStar}
 We know that cyclic exhaustive service shows optimal mean delay performance in systems with equal arrival rates. In this section, we make an important observation about its stability in systems with \emph{general} arrival rates, i.e., not necessarily equal.\\
 \indent The cyclic exhaustive service policy is \emph{\color{blue}throughput optimal}. The proof of the stability of the cyclic exhaustive service policy proceeds along the lines of the analysis in \cite{altman92stability-mcity}. Note, however, that the stability of this policy is proved 
 for a system with general (possibly unequal) arrivals rates $\lambda_1,\cdots,\lambda_N$, and includes, as a special case, the stability of the symmetric system. Formally, 

 \begin{thm}\label{thmCycExhStable} 
 Consider the system capacity region defined in Eqn.~\eqref{eqnBasicCapacityRegion}.
 If $\boldsymbol{\lambda}\in\Lambda$, the system is stable under cyclic exhaustive service.
 \end{thm}
 \begin{IEEEproof} The proof involves showing that the drift of a Lyapunov function computed at the end of $N$ busy periods, i.e., between two successive returns to any queue $i$, is negative. Refer to Sec.~\ref{ApendixProofOfThmCycExhStable} for details.
 \end{IEEEproof}
 \begin{rem}
 Although we motivated cyclic exhaustive service as the LEQ policy specialized to symmetric systems, Thm.~\ref{thmCycExhStable} shows that the former can actually stabilize unequal arrival rates as well. So we now have \emph{two} throughput optimal scheduling policies (LEQ and cyclic exhaustive service). But since the LEQ policy requires estimates of the arrival rate vector and cyclic exhaustive service does not, can we do away with LEQ altogether? As Fig.~\ref{figComparisonOfDelayOfLEQandSLQ} shows, with unequal arrival rates, LEQ can deliver substantially better delay performance and hence, is preferable to cyclic exhaustive service. In Sec.~\ref{secNumRes}, we show how LEQ can be implemented in a \emph{distributed} manner using arrival rate estimation, but where such estimation cannot be performed or is undesirable, cyclic exhaustive service can still be used to stabilize the system albeit at the cost of increased delay. 
 \end{rem}
  \begin{figure}[tb]
\centering
\includegraphics[height=5.5cm, width=8cm]{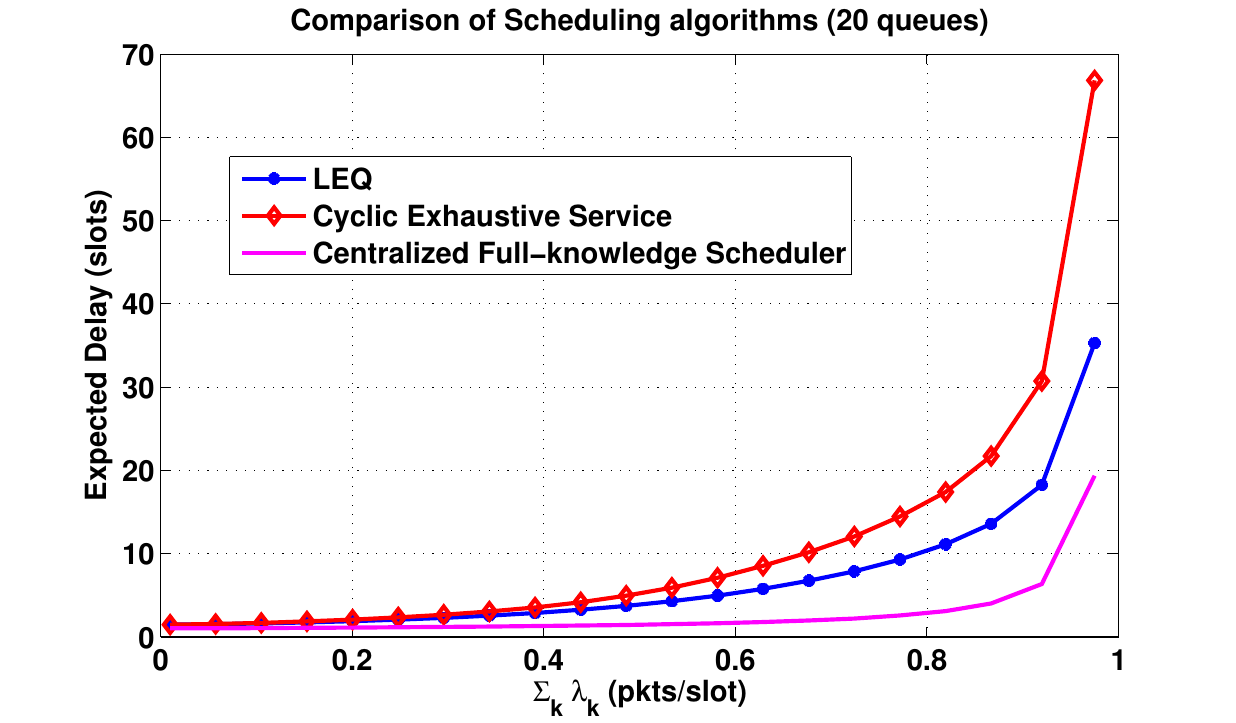}
\caption{Comparison of the performance of the LEQ and Cyclic Exhaustive Service policies. Also shown is the lowest possible delay achievable in this system (black curve). This is discussed in detail in Sec. VI. We assume a fixed ratio of packet rates that determines the arrival to each queue. Specifically, $r_i\lambda$ is the arrival rate to Queue $i$ with $\sum_{i=1}^N r_i = 1$. In the plot, $\lambda \in [0,1)$.}
\label{figComparisonOfDelayOfLEQandSLQ}
\end{figure}
\subsection{Proof of Thm.~\ref{thmCycExhStable}}\label{ApendixProofOfThmCycExhStable}
In this section, we return to the system discussed Sec.\ref{secSysMod}, i.e., one that is not necessarily symmetric\footnote{Proving stability for this system obviously proves stability in the symmetric system.}. The analysis proceeds along the lines of the proof of the theorem in \cite{altman92stability-mcity}. We will be looking at the system's state only at the epochs at which the server \emph{begins} its visit to a queue. 

With a slight abuse of notation, we define this system's state is by
$\mathbf{s}(n)=\left[\mathbf{Q}(n),I(n)\right]^T$, 
where, 
$\mathbf{Q}(n)=\left[Q_1(n),~Q_2(n),\dots,Q_N(n)\right]^T,n\geq 1$ is the vector of queue lengths
at the $n^{th}$ switching instant (at which the server visits the $n^{th}$ queue), $I(n)$ is the identity
of the $n^{th}$ queue polled. 
Clearly, $\mathbf{s}(n)\in\mathbb{N}^N\times I$, and 
this process is embedded at the instants at which the server \emph{arrives} at a queue. 

Further, the number of arrivals to queue $k$ in $t$ time-slots is denoted by $A_k(t)$ and 
the length of the $n^{th}$ busy period starting with $l$ packets by $G_n(l)$. Also, 
 \[
   I(n+1)=
\begin{cases}
    I(n)+1~~\text{if}~1\leq n\leq N-1\\
    1\hspace{1.4cm}\text{if}~n=N.
\end{cases}
\]
Consequently, when 
\begin{enumerate}


\item $j\neq I(n)$
 \[
   Q_j(n+1)=
\begin{cases}
    Q_j(n)+A_j\left(G_n\left(Q_{I(n)}(n)\right)\right),\\
    ~~~\text{if}~Q_{I(n)}(n)\geq1\\
    Q_j(n)+A_j(1),\text{if}~Q_{I(n)}(n)=0\\
\end{cases}
\]

 
\item $j=I(n)$
 \begin{equation}  
    Q_j(n+1)=A_j(1).\nonumber
 \end{equation}

\end{enumerate}

The mean duration of a busy period, $\mathbb{E}B(I(n))$ of queue $I(n)$ beginning with \emph{one}
packet is found by observing that
\begin{small}
\begin{eqnarray*}
\hspace{-2.0cm}
 \mathbb{E}B(I(n))&=&\left[1\times\lambda_{I(n)}+\{1+\mathbb{E}G_n(l)\}\{1-\lambda_{I(n)}\}\right].\nonumber\\
\Rightarrow\mathbb{E}B(I(n))&=&\frac{1}{1-\lambda_{I(n)}}.
\end{eqnarray*}
\end{small}
Hence, since $G_n(l)$ begins with $l$ packets instead of one, 
\begin{eqnarray}
\hspace{-1.0cm}
 \mathbb{E}G_n(l)&=&l\times\left[1\times\lambda_{I(n)}+\{1+\mathbb{E}G_n(l)\}\{1-\lambda_{I(n)}\}\right]\nonumber\\
&=&\frac{l}{1-\lambda_{I(n)}}\hspace{5cm}\nonumber\\
&\Rightarrow &\mathbb{E}A_j\left(G_n(l)\right)=\frac{\lambda_j}{1-\lambda_{I(n)}}l~\forall j\neq I(n)\nonumber\\
&\Rightarrow &\mathbb{E}\sum_{j\neq I(n)}A_j\left(G_n(l)\right)=\frac{\rho-\lambda_{I(n)}}{1-\lambda_{I(n)}}l.
\label{expBusy}
\end{eqnarray}

We will now show that $\rho=\sum_{i=1}^N\lambda_i<1$ is a sufficient condition for stability. 
Clearly, $\{\mathbf{Q}_n, n\geq 1\}$, is an irreducible DTMC. However,
it is not time homogeneous since, as (\ref{expBusy}) shows, the transition
probabilities depend on the queue being served, i.e., on $I(n)$. Hence, as defined in 
\cite{altman92stability-mcity}, the system is said to be \emph{stable} if the $N$ irreducible, 
\emph{homogeneous} DTMCs $\{\left[\mathbf{Q}(nN+k),k\right], n\geq 1, 1\leq k\leq N, i\in\{1,2\}\}$
are all ergodic.  

Now, from the definition of $Q_j(n+1)$, we see that
\begin{small}
\begin{eqnarray*}
\hspace{-0.15cm}
 \mathbb{E}\left[\sum_{j=1}^NQ_j(n+1)|\mathbf{Q}(n)\right]&\leq &\sum_{j=1}^NQ_{j}(n)+\rho\\
&&+\frac{\rho-\lambda_{I(n)}}{1-\lambda_{I(n)}}Q_{I(n)}(n)-Q_{I(n)}(n)\nonumber\\
&=&\sum_{j=1}^NQ_{j}(n)+\rho-(1-\rho)h_{I(n)}Q_{I(n)}(n),\\
&&~\text{where $h_{I(n)}=\frac{1}{1-\lambda_{I(n)}}$},
\end{eqnarray*}
\end{small}
which means that
\begin{small}
\begin{eqnarray*}
\hspace{-2.0cm}
\mathbb{E}\left[\sum_{j=1}^NQ_{j}(n+2)|\mathbf{Q}(n)\right]&=&\mathbb{E}\left[\mathbb{E}\left[\sum_{j=1}^NQ_{j}(n+2)\mid\mathbf{Q}(n+1)\right]\mid\mathbf{Q}(n)\right]\nonumber\\
&\leq &\mathbb{E}\bigg[\sum_{j=1}^NQ_{j}(n+1)+\rho\\
&&-(1-\rho)h_{I(n+1)}Q_{I(n+1)}(n+1)\mid \mathbf{Q}(n)\bigg] \nonumber \\
&\leq &\mathbb{E}\bigg[\sum_{j=1}^NQ_{j}(n+1)\mid \mathbf{Q}(n)\bigg]+\rho\\
&&-(1-\rho)h_{I(n+1)}Q_{I(n+1)}(n+1)\nonumber\\
&=&\sum_{j=1}^NQ_j(n)+2\rho-(1-\rho)\bigg[h_{I(n)}Q_{I(n)}(n)\\
&&+h_{I(n+1)}Q_{I(n+1)}(n)\bigg].
\end{eqnarray*}
\end{small}
Proceeding similarly, we get
\begin{small}
\begin{eqnarray*}
 \mathbb{E}\left[\sum_{j=1}^NQ_j(n+2N)\mid \mathbf{Q}(n)\right]&\leq&\sum_{j=1}^NQ_j(n)+2N\rho\nonumber\\
 &&-(1-\rho)\sum_{m=0}^{2N-1}h_{I(n+m)}Q_{I(n+m)}(n),\nonumber
 \end{eqnarray*}
 
 \begin{eqnarray} 
 \Rightarrow\mathbb{E}\left[\sum_{j=1}^NQ_j(n+2N)-\sum_{j=1}^NQ_j(n)\mid \mathbf{Q}(n)\right]\leq 2N\rho\nonumber\\
-(1-\rho)\sum_{m=0}^{2N-1}h_{I(n+m)}Q_{I(n+m)}(n).\hspace{-0.4cm}
 \label{NslotDrift}
\end{eqnarray}

\end{small}
Clearly, the RHS of (\ref{NslotDrift}) is at most $2N\rho$, and if for even one $1\leq k\leq 2N$ and
$\epsilon>0$, $Q_{I(n+k)}(n)>\frac{N\rho\epsilon}{(1-\rho)h(I(n+k))}$, then 
$\mathbb{E}\left[\sum_{j=1}^NQ_j(n+2N)-\sum_{j=1}^NQ_j(n)\mid \mathbf{Q}(n)\right]<-\epsilon$. Hence, by the 
Foster-Lyapunov criterion (\cite{fayolle-etal95constructive-theory-markov-chains}, Thm.~2.2.3), we see that 
$\{\left[\mathbf{Q}(nN+k),k,i\right], n\geq 1\}$ is positive recurrent for every 
$1\leq k\leq N$, and $i\in\{1,2\}$. Being irreducible DTMCs as well, it is also ergodic.

\IEEEQED
\subsection{Formulating the MDP}\label{AppendixSecFormulatingTheMDP}

\begin{enumerate}
 \item State Space: The state of the system at time $n\geq1$ is the vector 
 \begin{equation}
  s(t)=\left[Q_{u_{n-1}}(n),\mathbf{V}(n),\mathbf{r}(n),u_{n-1}\right].
 \end{equation} 
So $s(n)\in\mathbb{N}\times\mathbb{N}^N\times\mathbb{N}^{N-1}\times I$. %

\item Action Space: The action space $\mathbb{A}=I$ for all states. 

\item Initial distribution on the State Space: We assume that the system begins empty and so, the initial
distribution $D$ is simply $\mathbb{I}\{s(0)\}$, where $s(0)=\left[q,\mathbf{V}(0),\mathbf{r}(0),i_0\right],$
is given and hence, known.

\item History: Since the state is observable, the initial history is nothing but the initial state,
which is supplied as the initial information. 
\begin{equation}
 h_0=s(0).
\end{equation}
At time $t$, the history provides the complete picture of successive observations and control actions
chosen so far; for $n\geq1$:
\begin{equation}
 h_n=\left(s(0), a(0),\dots,s(n-1),a(n-1),s(n)\right),
\end{equation}
where the action taken at time $n$, $a(n)\in\mathbb{I},\forall n\geq0.$

\item Policies: In general, a policy is a sequence of conditional distributions $\xi=\{\xi_n\}$ such that for each $n$, $\xi_n$ is a distribution on $\mathbb{A}$ given $h_n$. As mentioned before, we restrict attention to policies that serve any queue $j$ with $r_j(t)>0$ first and then return to serve the incumbent. As a consequence of Prop.~\ref{propNonIdlingAndExhOpt}, we also restrict to the case where if $q>0$ then we serve $i$ again.

\item Transition Law: We denote the current action by $u_n$ and distinguish two cases:
\begin{itemize}
 \item $r_i(n)\geq1$ for some $i\in I-{u_{n-1}}$. Then, \small
\begin{eqnarray}
\hspace{-1.00cm}
q_{u_n}(s(n),s(n+1))&=&\mathbb{I}_{\left\lbrace r_i(n+1)=r_i(n)-1\right\rbrace}\times\nonumber\\ 
&&\Pi_{j\neq i}\mathbb{I}_{\left\lbrace r_j(n+1)=r_j(n)\right\rbrace}\mathbb{I}_{\left\lbrace u_{t}=i\right\rbrace}\times\nonumber\\
&&\left(\Pi_{k=1}^N\mathbb{I}_{u_n}\{V_k(n),V_k(n+1)\}\right)\times\nonumber\\
&&\left(\Pi_{k=1}^Nq'_{u_n}(Q_k(n),Q_k(n+1))\right),\hspace{0.9cm}
\end{eqnarray} 
\normalsize where \small 
\[
\hspace{-2.0cm}   {I}_{u_n}\{V_k(n),V_k(n+1)\}:=
\begin{cases}
    1~\text{if}~V_k(n+1)=V_k(n)+1~\text{and}~k\neq u_n,\\
    1~\text{if}~V_k(n+1)=V_k(n)=0~\text{and}~k=u_n,~\text{and }\\
    0~\text{otherwise},
\end{cases}
\]
and 
\[
\hspace{-1.0cm}   q'_{u_n}(Q_k(n),Q_k(n+1)):=
\begin{cases}
    \lambda~~~~~~\text{if}~Q_k(n+1)=Q_k(n)+1\\ 
    ~~~~~~~\text{OR, if}~Q_k(n+1)=(Q_k(n)-1)^++1\\~~~~~~~~\text{and }k=i,\\
    1-\lambda~\text{if}~Q_k(n+1)=Q_k(n)+0\\ 
    ~~~~~~~\text{OR, if}~Q_k(n+1)=(Q_k(n)-1)^++0\\~~~~~~~~\text{and }k=i,\\
    0~~~~~~\text{otherwise},
\end{cases}
\]
\normalsize
\item $\mathbf{r}(n)=\mathbf{0}$. Then, \small
\begin{eqnarray}
q_{u_n}(s(n),s(n+1))&=&\Pi_{j\in I}\mathbb{I}_{\left\lbrace r_j(n+1)=0\right\rbrace}\times\nonumber\\ &&\left(\Pi_{k=1}^N\mathbb{I}_{u_n}\{V_k(n),V_k(n+1)\}\right)\times\nonumber\\
&&\left(\Pi_{k=1}^Nq'_{u_n}(Q_k(n),Q_k(n+1))\right),\hspace{0.9cm}
\end{eqnarray} 
\normalsize where \small
\[
\hspace{-2.0cm}   {I}_{u_n}\{V_k(n),V_k(n+1)\}:=
\begin{cases}
    1~\text{if}~V_k(n+1)=V_k(n)+1~\text{and}~k\neq u_n,\\
    1~\text{if}~V_k(n+1)=V_k(n)=0~\text{and}~k=u_n,~\text{and}\\
    0~\text{otherwise},
\end{cases}
\]
\normalsize and \small
\[
\hspace{-1.0cm}   q'_{u_n}(Q_k(n),Q_k(n+1)):=
\begin{cases}
    \lambda~~~~~~\text{if}~Q_k(n+1)=Q_k(n)+1~\text{and}~k\neq u_n,\\
    ~~~~~~~\text{OR, if}~Q_k(n+1)=(Q_k(n)-1)^++1\\~~~~~~~~\text{and }k=u_n,\\
    1-\lambda~\text{if}~Q_k(n+1)=Q_k(n)+0~\text{and}~k\neq u_n,\\
    ~~~~~~~\text{OR, if}~Q_k(n+1)=(Q_k(n)-1)^++0\\~~~~~~~~\text{and }k=u_n,\\    
    0~~~~~~\text{otherwise},
\end{cases}
\]
\end{itemize} 
\normalsize 
\item Single Stage Cost: This gives the expected cost over the current step when the current state and 
action are known. 
\begin{eqnarray}
\hspace{-1.00cm}
 c(s(n),u_n)&=&\mathbb{E}\left[\sum_{i=1}^N Q_i(n)|s_n\right]\nonumber\\
&=&Q_{u_{n-1}}(n)+\lambda\sum_{j\neq u_{n-1}}V_j(n)+\sum_{i\in I}R_i(n),\hspace{0.9cm}
\end{eqnarray}
which is the expected sum of the current queue lengths conditioned on the current state.
\end{enumerate}
Recall that the optimal cost (Eqn.~\eqref{eqnDiscountedCostOfCOMDP}) is given by
\begin{equation}
 J^*(s(0))=min_{\xi}~\mathbb{E}^\xi_{s(0)}\sum_{n=0}^\infty\alpha^n\mathbb{E}\left[\sum_{i=1}^NQ_i(n)|s_n\right].
\end{equation}
It is to be noted that $J^*(s(0))$ on the R.H.S of \eqref{eqnDiscountedCostOfCOMDP} exists for every initial state $s(0)$. Since arrivals are IID Bernoulli variables, the single stage cost can increase at most linearly with time, while the discount factor $\alpha^n$ decreases exponentially and dominates. 

The Bellman Optimality equations \cite{bertsekas95mdp-control} associated with this MDP formulation are as follows (recall from Eqn.~\eqref{eqnBellmanOpt} that we denote $J^*(q; \mathbf{V}; \mathbf{r}; i)$ by $J_i^*(q; \mathbf{V}; \mathbf{r})$). When $\mathbf{r}=\mathbf{0}$, 
\begin{eqnarray*}
J^*_i(q>0; \mathbf{V}; \mathbf{r}=\mathbf{0}) &=& q+\lambda \sum_{k \neq i} V_k + \alpha \mathbb{E} J^*_i\left(q-1+A; \right.\\
&& \left. V_i=0, \mathbf{V}_{-i}+\mathbf{1}; \mathbf{r}=\mathbf{0}\right), \\
J^*_i(q=0; \mathbf{V}; \mathbf{r}=\mathbf{0}) &=& \lambda \sum_{k \neq i} V_k +\alpha \min_{j \neq i} \mathbb{E} J^*_j\left(B(V_j);\right.\\ 
&& \left. V_j=0, \mathbf{V}_{-j}+\mathbf{1}; \mathbf{r}=\mathbf{0} \right),
\end{eqnarray*}
where $A$ is a generic Bernoulli($\lambda$) random variable, 
$\mathbf{V}_{-i}=[V_1,\dots,V_{i-1},V_{i+1},\dots,V_N]$ and $\mathbf{1}\in\mathbb{R}^{(N-1)}$ is the
vector with 1's at all coordinates. Finally, if random variable $C$ is distributed
Binomial($V_j,\lambda$), $B(V_j)$ is a random variable whose distribution is the same as that of
$(C-1)^++A$.

If $r_j\geq1$ for some $j\in I-\{i\}$, the server first serves the $\sum_{j\neq i}r_j$ packets and then returns to queue $i$. Let $\tilde{R}=\sum_{j\neq i}r_j$, $L_{-i}=\sum_{j\neq i}V_j$, and $\tilde{L}_{-i}=\sum_{k=1}^{\tilde{R}}\alpha^{(\tilde{R}-k)}(L_{-i}+Nk)$. Clubbing all these $\tilde{R}$ slots, we get:
\begin{eqnarray}
J^*_i(q>0; \mathbf{V}; \mathbf{r}) &=& q\sum_{k=0}^{\tilde{R}-1}\alpha^k+\lambda\tilde{L}_{-i}+ \alpha^{\tilde{R}}\mathbb{E} J^*_i\left(q-1+\tilde{A};\right.\nonumber\\
&& \left. V_i=0, \mathbf{V}_{-i}+\tilde{R}\mathbf{1}; \mathbf{r}=\mathbf{0}\right) \nonumber\\
J^*_i(q=0;\mathbf{V};\mathbf{r}) &=& \lambda\tilde{L}_{-i}+\alpha^{\tilde{R}}\mathbb{E} J^*_j\left(\tilde{A}; V_j=0, \right.\nonumber\\
&& \left.\mathbf{V}_{-j}+\tilde{R}\mathbf{1};\mathbf{r}=\mathbf{0}\right),\hspace{0.75cm}
\label{eqnBellmanWhenRPositive}
\end{eqnarray}
where $\tilde{A}$ is a Binomial$\left(\tilde{R},\lambda\right)$ random variable and represents the number of arrivals over $\tilde{R}$ slots.
\subsection{Proof of Thm.~\ref{thmCycExhDelOpt}}\label{AppendixProofOfThmCycExhDelOpt}
Before we prove that cyclic exhaustive service solves the discounted cost MDP described in Sec.~\ref{AppendixSecFormulatingTheMDP}, we require an intermediate result on monotonicity of the optimal cost function. 
\subsubsection{Monotonicity of $J^*(\cdot)$}\label{secMtone}
Consider, once again, the Bellman equations. 
\begin{eqnarray}
J^*_i(q>0; \mathbf{v}; \mathbf{r}=\mathbf{0})&=& q+\lambda \sum_{k \neq i} v_k + \alpha \mathbb{E} J^*_i\left( q-1+A; \right.\nonumber\\
&& \left. v_i=0, \mathbf{v}_{-i}+\mathbf{1}; \mathbf{r}=\mathbf{0} \right),\nonumber\\
J^*_i(q=0; \mathbf{v}; \mathbf{r}=\mathbf{0}) &=& \lambda \sum_{k \neq i} v_k + \alpha \min_{j \neq i} \mathbb{E} J^*_j \left(B(v_j); \right.\nonumber\\
&& \left. v_j=0, \mathbf{v}_{-j}+\mathbf{1}; \mathbf{r}=\mathbf{0} \right),
\label{jStarR0}
\end{eqnarray}
where, we recall that if random variable $C$ is distributed Binomial($v_j,\lambda$), $B(v_j)$ is a random variable whose distribution is the same as that of $(C-1)^++A$ and $A$ is a Bernoulli random variable with $\mathbb{E}A=\lambda$. We shall now prove that the optimal cost function is monotonically increasing in the first coordinate 
when $\mathbf{r}=\mathbf{0}$ (this is really the only scenario where decisions might have to be taken).

\lem When $\mathbf{r}=\mathbf{0}$, the optimal $\alpha-$discounted cost function, $J^*$ satisfies
\begin{equation}
J^*(q_1,\mathbf{v},\mathbf{0})\geq J^*(q_2,\mathbf{v},\mathbf{0}),~\forall q_1\geq q_2,~\text{and}~\mathbf{v}\in\mathbb{N}^N.
\end{equation}
\label{lemJstarInc}

\begin{IEEEproof} Let $\mathbb{M}$ denote the set of all functions $g:\mathbb{S}\mapsto\mathbb{R}_+$. 
Considering the R.H.S of equations(\ref{jStarR0}), we define the Dynamic Programming Operator $T:\mathbb{M}\mapsto\mathbb{M}$
as in \cite{bertsekas95mdp-control} as follows. For any $g\in\mathbb{M}$,
\begin{eqnarray*}
(Tg)_i(q>0,\mathbf{v},\mathbf{r}) &=& q+\lambda \sum_{k \neq i} v_k + \alpha \mathbb{E} g_i\left(q-1+A;\right. \\
&& \left. v_i=0, \mathbf{v}_{-i}+\mathbf{1}; \mathbf{r}=\mathbf{0}\right),\nonumber\\
(Tg)_i(q=0,\mathbf{v},\mathbf{r}) &=& \lambda \sum_{k \neq i} v_k + \alpha \min_{j \neq i} \mathbb{E} g_j\left( B(v_j);\right.\\
&& \left. v_j=0, \mathbf{v}_{-j}+\mathbf{1}; \mathbf{r}=\mathbf{0} \right),
\end{eqnarray*}
Consider any $g\in\mathbb{M}$ that is increasing in $q$, i.e., 
$g(q+k,\mathbf{v},\mathbf{r})\geq g(q,\mathbf{v},\mathbf{r}),~\forall k\geq1$.
When
\begin{itemize}
\item $q\geq1$: 
\begin{eqnarray}
\hspace{-1.50cm} 
(Tg)_i(q+k,\mathbf{v},\mathbf{r}) &=& (q+k)+\lambda \sum_{k \neq i} v_k + \alpha \mathbb{E} g_i\bigg(q+k-1+A;\nonumber\\
&& v_i=0, \mathbf{v}_{-i}+\mathbf{1}; \mathbf{r}=\mathbf{0}\bigg),\nonumber\\
&\geq& q+\lambda \sum_{k \neq i} v_k + \alpha \mathbb{E} g_i\bigg(q-1+A; \nonumber\\ 
&& v_i=0, \mathbf{v}_{-i}+\mathbf{1}; \mathbf{r}=\mathbf{0}\bigg),\nonumber\\
&=& (Tg)_i(q,\mathbf{v},\mathbf{r}).\nonumber
 \end{eqnarray}
 
This shows that the operator $T$ preserves monotonicity for $q\geq1$. Proposition 2.1 of \cite{bertsekas95mdp-control}
shows that value iteration, starting from  any $g\in\mathbb{M}$, converges to the 
fixed point of the operator $T$, namely $J^*$, which also resides in
$\mathbb{M}$. Hence, we see that $J^*$ is monotonically increasing in $q$ when $q\geq1$. 

 \item $q=0$:\small
 \begin{eqnarray}
 \hspace{-1.50cm}
 J^*_i(0; \mathbf{v}; \mathbf{r}=\mathbf{0}) &=& \lambda\sum_{k \neq i} v_k+\alpha\min_{j\neq i}\mathbb{E} J^*_j \bigg( B(v_j);v_j=0, \mathbf{v}_{-j}+\mathbf{1}; \mathbf{r}=\mathbf{0} \bigg),\nonumber\\ 
&\stackrel{\dagger}{=}& \lambda\sum_{k \neq i} v_k+\alpha\min_{j\in I}\mathbb{E} J^*_j \bigg( B(v_j);v_j=0, \mathbf{v}_{-j}+\mathbf{1}; \mathbf{r}=\mathbf{0} \bigg),\nonumber\\ 
&=& \lambda\sum_{k \neq i} v_k\nonumber\\ 
&& +\alpha\min\left\lbrace\min_{j\neq i}\mathbb{E} J^*_j \bigg( B(v_j);v_j=0, \mathbf{v}_{-j}+\mathbf{1}; \mathbf{r}=\mathbf{0} \bigg),\right.\nonumber\\
&& \left. \mathbb{E}J^*_i\bigg(A; v_i=0,\mathbf{v}_{-i}+\mathbf{1}; \mathbf{r}=\mathbf{0}\bigg)\right\rbrace,\nonumber\\
&\leq& 1+\lambda\sum_{k \neq i} v_k+\alpha\mathbb{E}J^*_i\bigg(A; v_i=0,\mathbf{v}_{-i}+\mathbf{1}; \mathbf{r}=\mathbf{0}\bigg),\nonumber\\
&=& J^*_i(1; \mathbf{v}; \mathbf{r}=\mathbf{0}).\nonumber
\end{eqnarray}
\normalsize
In equality $\dagger$ minimization over $I\setminus\{i\}$ is the same as minimization over all $I$, since attempting to schedule queue $i$ would result in slot wastage and idling, which, as Prop.~\ref{propNonIdlingAndExhOpt} shows is suboptimal and hence, will not result in a smaller cost.
\end{itemize}
\end{IEEEproof}
\subsubsection{Proof of Thm.~\ref{thmCycExhDelOpt}}
In this section, we use $\mathbf{r}_{-n}$ instead of $\mathbf{r}$ to explicitly show which queues' $r_j$'s are known to the server. So, $\mathbf{r}_{-n}$ includes $r_j~\forall j\neq n$. Let D be a $Bin(v_n-v_m-1,\lambda)$ random variable. Wherever it becomes necessary to clarify the measure with respect to (w.r.t) which the expectation is being computed, we shall use $\mathbb{E}_\mu$ to explicitly denote expectation w.r.t the probability mass function (p.m.f) $\mu$. Let $\nu_m$ denote the p.m.f of $B(v_m)$, $\nu_n$ that of $B(v_m)$ and $\mu$ that of $D$, i.e., $Bin(v_n-v_m-1,\lambda)$. Also, define $\chi\stackrel{\Delta}{=}\nu_m\ast \mu$ (here, $\ast$ denotes convolution) and $\beta\stackrel{\Delta}{=}\nu_m\times \mu$ (the product measure).
Now, with $v_n-v_m>0$, we can write:

{
\small
\begin{eqnarray}
 && \mathbb{E}_{\nu_n} \bigg[ J^*_n\bigg( B(v_n);v_n=0, \mathbf{v}_{-n}+\mathbf{1}; \mathbf{r}_{-n}=\mathbf{0} \bigg)\bigg] \nonumber\\
 && ~~~~~ - \mathbb{E}_{\nu_m} \bigg[ J^*_m \bigg( B(v_m);\mathbf{v}_{-m}+\mathbf{1}; \mathbf{r}_{-m}=\mathbf{0} \bigg) \bigg] \nonumber\\
&&\stackrel{a}{\leq}\mathbb{E}_{\chi} \bigg[ J^*_n \bigg( B(v_m)+D;\mathbf{v}_{-n}+\mathbf{1}; \mathbf{r}_{-n}=\mathbf{0} \bigg)\bigg] \nonumber\\
&& ~~~~~ - \mathbb{E}_{\nu_m} \bigg[ J^*_m \bigg( B(v_m);\mathbf{v}_{-m}+\mathbf{1}; \mathbf{r}_{-m}=0 \bigg) \bigg]\\
&&= \mathbb{E}_{\beta} \bigg[ J^*_n \bigg( B(v_m)+D;\mathbf{v}_{-n}+\mathbf{1}; \mathbf{r}_{-n}=\mathbf{0} \bigg) \nonumber\\
&& ~~~~~ - J^*_m \bigg( B(v_m);\mathbf{v}_{-m}+\mathbf{1}; \mathbf{r}_{-m}=0 \bigg) \bigg]\\
&&\stackrel{\dag}{=} \sum_{r=1}^{v_n-v_m-1} P(r)~\mathbb{E} \bigg[ J^*_n \bigg( B(v_m)+r;\mathbf{v}_{-n}+\mathbf{1}; \mathbf{r}_{-n}=\mathbf{0} \bigg)\nonumber\\
&& ~~~~~ - J^*_m \bigg( B(v_m);\mathbf{v}_{-m}+\mathbf{1}; \mathbf{r}_{-m}=\mathbf{0} \bigg) \bigg| D=r \bigg]  \nonumber\\
&& ~~~~~ + P(0)~\mathbb{E} \bigg[ J^*_n \bigg( B(v_m);\mathbf{v}_{-n}+\mathbf{1}; \mathbf{r}_{-n}=\mathbf{0} \bigg) \nonumber\\
&& ~~~~~ - J^*_m \bigg( B(v_m);\mathbf{v}_{-m}+\mathbf{1}; \mathbf{r}_{-m}=\mathbf{0} \bigg) \bigg| D=0 \bigg]  \nonumber\\
&&\stackrel{b}{\leq} \sum_{r=1}^{v_n-v_m-1} P(r)~\mathbb{E} \bigg[ J^*_n \bigg( B(v_m)+r;\mathbf{v}_{-n}+\mathbf{1}; \mathbf{r}_{-n}=\mathbf{0} \bigg)\nonumber\\
&& ~~~~~ -J^*_m\bigg( B(v_m);\mathbf{v}_{-m}+\mathbf{1}; \mathbf{r}_{-m}=\mathbf{0} \bigg) \bigg| D=r \bigg] 
\label{eqn:outline1}
\end{eqnarray}
\normalsize
Before we explain inequality $a$ we will require the following result.
\lem\label{lemStochasticOrdering}
\begin{equation}
B(v_n)\stackrel{st}{\leq}B(v_m)\ast D,
\end{equation}
}

where, as before, $st$ refers to stochastic ordering and $\ast$ denotes convolution.

\pf Refer Sec.~\ref{AppendixProofPfStochasticOrdering}. \IEEEQED \\

Since nondecreasing functions of stochastically ordered random variables are themselves stochastically ordered, inequality $a$ follows from Lem.~\ref{lemStochasticOrdering} and the monotonicity of $J^*$ in the first coordinate as proved in Lemma.\ref{lemJstarInc} in Sec.~\ref{secMtone}. In equality $\dag$, $P(r)= Pr\{D=r\}$, as mentioned before. 

To explain inequality $b$ we use the following observation about the difference between the state in which our controlled Markov chain (CMC) actually is and what the decision maker knows about that state. 
Note that 
the lemma that follows is true for CMCs \emph{in general}. Let $\mathcal{X}$ (a countable set) be the state space on which the CMC actually evolves. The MDP is then defined on the set of all probability mass functions, or p.m.fs on $\mathcal{X}$ denoted by $\mathcal{P}(\mathcal{X}):= \{\mathbf{p}:\forall x\in \mathcal{X},~0\leq p(x)\leq1, ~\text{and} \sum_{x\in\mathcal{X}} p(x)=1\}$. Let the action space be denoted by $\mathcal{A}$.  

\begin{lem}\label{lemCondition2Know} 
Let $\mu\in\mathcal{P}(\mathcal{X})$ be a probability mass function, or p.m.f, on $\mathcal{X}$ (this p.m.f \emph{can} be degenerate, i.e., $1$ at some $x\in\mathcal{X}$ and $0$ elsewhere) and $\pi^*$ any stationary optimal policy (assuming one exists). Denote the action space when the state is $\mu$ by $A_\mu\subset\mathcal{A}$ and the state of the CMC at time $0$ by $X_0\in\mathcal{X}$. Finally, let $c:\mathcal{P}(\mathcal{X})\times\mathcal{A}\mapsto \mathbb{R}_+$ denote the single stage cost and $\mu_o=\mu$. Then\footnote{The discount factor is denoted by $\alpha$.},
\begin{eqnarray}
\mathbb{E}^{\pi^*}_{\mu}\bigg[\sum_{k=0}^\infty \alpha^k c(\mu_k,A_k)\bigg| X_0=x\bigg]\geq 
\mathbb{E}^{\pi^*}_{\delta_x}\bigg[\sum_{k=0}^\infty  \alpha^k c(\mu_k,A_k)\bigg],
\label{eqnCondition2Know}
\end{eqnarray}
where $\mu_k\in\mathcal{P}(\mathcal{X}),~A_k\in A_{\mu_k},\text{ for all } k\geq0$.
\end{lem}
\begin{IEEEproof}
Intuitively speaking, the L.H.S of \eqref{eqnCondition2Know} is an expectation over the subset of sample paths over which $X_0=x$ and when decision-making entity (in our case, the centralized scheduler) begins with knowledge of some p.m.f $\mu$ over the state. Clearly, this does not mean that the entity \emph{knows} that $X_0=x$. But the expectation on the R.H.S indicates that the scheduler begins knowing that $X_0=x$, hence the $\delta_x$ in the subscript. Obviously, the action taken by the decision-making entity when it has complete knowledge of the state (represented by $\delta_x$) cannot be worse than that taken when there is ambiguity about the state (represented by some p.m.f $\mu$ which might not be $\delta_x$). For details, refer Sec.~\ref{AppendixProofOfCondition2Know} in the Appendix.
\end{IEEEproof}

In short, the lemma says that while transitioning from conditioning to knowledge, the optimal cost cannot increase. 
Getting back to proving Eqn.~\eqref{eqn:outline1}, Lem.~\ref{lemCondition2Know} immediately gives us inequality $c$ below, because $D=0$ essentially means that queue $n$'s distribution is $Bin(v_m,\lambda)$. Inequality $c$ is explained in more detail after \eqref{eqnUsingConditionToKnowForDEqualTo0}.
\begin{small}
\begin{eqnarray}
&&\mathbb{E} \bigg[J^*_m \bigg( B(v_m);\mathbf{v}_{-m}+\mathbf{1}; \mathbf{r}_{-m}=\mathbf{0} \bigg) \bigg| D=0 \bigg]\nonumber\\
&\stackrel{c}{\geq}& \mathbb{E} \bigg[ J^*_m \bigg( B(v_m);\mathbf{v}_{-\{m,n\}}+\mathbf{1}, v_n=v_m+1; \mathbf{r}_{-m}=\mathbf{0} \bigg)\bigg]\nonumber\\
&\stackrel{d}{=}& \mathbb{E} \bigg[ J^*_n \bigg( B(v_m);\mathbf{v}_{-n}+\mathbf{1}; \mathbf{r}_{-n}=\mathbf{0} \bigg)\bigg]\nonumber\\
&\stackrel{e}{=}& \mathbb{E} \bigg[ J^*_n \bigg( B(v_m);\mathbf{v}_{-n}+\mathbf{1}; \mathbf{r}_{-n}=\mathbf{0}\bigg)\bigg| D=0\bigg].
\label{eqnUsingConditionToKnowForDEqualTo0}
\end{eqnarray}
\end{small}
\normalsize
When the scheduler \emph{knows} that both queue $m$ and queue $n$ have the same number of packets, viz, $B(v_m)$, serving either queue gives the same expected cost, since both queues have statistically similar arrival processes. So, equality $d$ follows from the fact that two queues of length Bin$(v_m,\lambda)$, with equal arrival rates are stochastically indistinguishable. The R.H.S of equality $e$ is the expected cost of serving queue $n$, when the server knows that queue $m$ is distributed $Bin(v_m,\lambda)$ and the backlog of queue $m$ is \emph{actually} a $Bin(v_m,\lambda)$ random variable, but that is exactly the R.H.S of equality $d$! Contrast this with the L.H.S of inequality $c$, where the server only knows that queue $n$ is distributed $Bin(v_n,\lambda)$, while conditioning on $D=0$ means that queue $n$ is actually distributed $Bin(v_m,\lambda)$. So,
\small
\begin{eqnarray}
\hspace{-1.00cm}
&& \mathbb{E} \bigg[ J_n^* \bigg( B(v_m)+r;\mathbf{v}_{-n}+1; \mathbf{r}_{-n}=0 \bigg)\nonumber\\
&& ~~~~~ - J_m^* \bigg( B(v_m);\mathbf{v}_{-m}+1; \mathbf{r}_{-m}=0 \bigg) \bigg| D=0 \bigg]\leq0.
\label{eqnDEquals0}
\end{eqnarray}
\normalsize
Now, on to the R.H.S of inequality $b$. Invoking Lem.~\ref{lemCondition2Know} and noting that $r_n=r$ means the scheduler \emph{knows} that queue $n$ has at least $r$ packets, we get
\small
\begin{eqnarray}
&& \mathbb{E} \bigg[ J^*_m\bigg( B(v_m);\mathbf{v}_{-m}+\mathbf{1}; \mathbf{r}_{-m}=\mathbf{0} \bigg) \bigg| D=r \bigg] \nonumber\\
&\geq& \mathbb{E} \bigg[ J^*_m\bigg( B(v_m);\mathbf{v}_{-m}+\mathbf{1}; \mathbf{r}_{\{-m,n\}}=\mathbf{0}, r_n=r \bigg) \bigg] \nonumber \\
&\stackrel{f}{=}&\mathbb{E} \bigg[ J^*_n\bigg( B(v_m)+r; \mathbf{v}_{-n}+\mathbf{1}; \mathbf{r}_{-n}=\mathbf{0}\bigg)\bigg] \nonumber\\
&\stackrel{g}{=}&\mathbb{E} \bigg[ J^*_n\bigg( B(v_m)+r; \mathbf{v}_{-m}+\mathbf{1}; \mathbf{r}_{n}=\mathbf{0}\bigg)\bigg|D=r \bigg]. 
\end{eqnarray}
\normalsize
Equality $f$ is explained as follows. In the space of policies we are considering, a server always first serves all queues with positive $r$ and then returns to the incumbent. So, when it is known that $r_n=r$, the server will first serve the $r$ packets in queue $n$ and then return to queue $m$. But that will result in the same mean cost as serving queue $n$ exhaustively, when it is known that queue $n$'s backlog is distributed $\delta_r\ast Bin(v_m,\lambda)$. Equality $g$ uses the same logic as Equality $e$ in \eqref{eqnUsingConditionToKnowForDEqualTo0}.
%
\small
\begin{align}
\mathbb{E} \bigg[ J_n^* \bigg( B(v_m)+r;\mathbf{v}_{-n}+1; \mathbf{r}_{-n}=0 \bigg) - J_m^* \bigg( B(v_m);\mathbf{v}_{-m}+1; \mathbf{r}_{-m}=0 \bigg) \bigg| D=r \bigg]\leq0, ~\forall 1\leq r\leq v_n-v_m.
\label{eqnDGreaterThan0}
\end{align}
\normalsize
From \eqref{eqnDEquals0} and \eqref{eqnDGreaterThan0}, the policy structure now follows. 
\IEEEQED
\subsection{Proof of Lem.~\ref{lemStochasticOrdering}}\label{AppendixProofPfStochasticOrdering}
Given $k\geq1$ if $Z$ is a random variable with p.m.f Bin($k,\lambda$) we define $\tilde{B}(k)$ as the p.m.f of $(Z-1)^+$. Let $X_k~(k\geq1),$ and $Y_n~(n\geq1)$ be random variables with distributions $\tilde{B}(k)$, and Bin($n,\lambda$), respectively. In what follows, we will first show that 
\begin{equation}
X_k+Y_n\stackrel{st}{\geq}X_{k+1}+Y_{n-1},~\forall k\geq1.
\label{eqnStochasticOrderPf1}
\end{equation}
Since the L.H.S of Eqn.~\ref{eqnStochasticOrderPf1}, $X_k+Y_n\stackrel{d}{=}X_k+Y_1+Y_{n-1}$ (where $\stackrel{d}{=}$ denotes equality of distribution), it is sufficient to prove that 
\begin{equation}
X_k+Y_1\stackrel{st}{\geq}X_{k+1}.
\label{eqnStochasticOrderPfEquivalent}
\end{equation}
We show this by proving that 
\begin{equation}
P\{X_k+Y_1\leq l\}\leq P\{X_{k+1}\leq l-1\},~\forall l\in\{0,1,\dots,k\}.
\end{equation}
We consider several cases
\begin{itemize}
\item $l=0$
\begin{align}
P\{X_{k+1}\leq 0\}=(1-\lambda)^{k+1}+(k+1)(1-\lambda)^k\lambda
=(1-\lambda)^k\left(1+k\lambda\right),\nonumber
\end{align}
while
\begin{eqnarray}
P\{X_k+Y_1\leq 0\} &=& \left((1-\lambda)^{k}+k(1-\lambda)^{k-1}\lambda\right)(1-\lambda)\nonumber\\
&=& (1-\lambda)^k\left(1+(k-1)\lambda\right)\nonumber\\ 
&<& P\{X_{k+1}\leq 0\}.\nonumber
\end{eqnarray}
\item $1\leq l\leq k-1$
\begin{align}
P\{X_{k+1}=l\}={k+1\choose l+1}  \lambda^{l+1}(1-\lambda)^{k-l},\nonumber
\end{align}
while
\begin{eqnarray}
P\{X_k+Y_1=l\} &=& {k\choose l+1}\lambda^{(l+1)}(1-\lambda)^{(k-l-1)}(1-\lambda)\nonumber\\
&& +{k\choose l}\lambda^l(1-\lambda)^{(k-l)}\lambda\nonumber\\
&=& \left({k\choose l+1}+{k\choose l}\right)\lambda^{l+1}(1-\lambda)^{(k-l)}\nonumber\\
&=& P\{X_{k+1}=l\}.\nonumber
\end{eqnarray}
Hence, in this case, $P\{X_{k+1}\leq l\}\geq P\{X_{k}+Y_1\leq l\}$.
\item $l=k$
\begin{eqnarray}
P\{X_{k+1}=l\}&=&\lambda^{k+1},~\text{and}\hspace{3.0cm}\nonumber\\
P\{X_k+Y_1=l\}&=&P\{X_k=k\}P\{Y_1=1\}\nonumber\\ 
&=& \lambda^{k+1}\nonumber\\
&=& P\{X_{k+1}=l\}.\hspace{2.6cm}\nonumber\\
\end{eqnarray}
Once again, $P\{X_{k+1}\leq k\}\geq P\{X_{k}+Y_1\leq k\}$.
\end{itemize}
This proves Eqn.~\ref{eqnStochasticOrderPfEquivalent} and hence, Eqn.~\ref{eqnStochasticOrderPf1}. By the same token, it is at once obvious that $X_k+Y_n\stackrel{st}{\geq}X_{k+1}+Y_{n-1}  \stackrel{st}{\geq}X_{k+2}+Y_{n-2} \dots \stackrel{st}{\geq}X_{k+n}$, which means that,
\begin{eqnarray}
\tilde{B}(k+n) &\stackrel{st}{\leq}& \tilde{B}(k)\ast Bin(n)\nonumber\\
\Rightarrow \tilde{B}(k+n)\ast Bin(1,\lambda) &\stackrel{st}{\leq}& \tilde{B}(k)\ast Bin(n)\ast Bin(1,\lambda)\nonumber\\
\Rightarrow B(k+n) &\stackrel{st}{\leq}& B(k)\ast\text{B}(n).
\label{eqnStochasticOrderingProofPMF}
\end{eqnarray}
Lem.~\ref{lemStochasticOrdering} follows by setting $v_m=k$ and $v_n=k+n$, in \eqref{eqnStochasticOrderingProofPMF}.
\IEEEQED
\subsection{Proof of Lem.~\ref{lemCondition2Know}}\label{AppendixProofOfCondition2Know}
For every $a\in\mathcal{A}$ denote by $p_{a}(z,\nu)$ the stochastic transition kernel\footnote{Refer Chapter 4 of \cite{hernandez89adaptive-control-markov} for details.} from $\mathcal{P}(\mathcal{X})\times \mathcal{A}\mapsto \mathcal{P}(\mathcal{X})$. Such a kernel exists by Lem.~3.2 and Eqn.~3.7 in Chapter 4 of \cite{hernandez89adaptive-control-markov}. Now, observe that the R.H.S of \ref{eqnCondition2Know}
\small
\begin{eqnarray}
\hspace{-1.00cm}
\mathbb{E}^{\pi^*}_{\delta_x}\bigg[\sum_{k=0}^\infty \alpha^k c(\mu_k,A_k)\bigg] &=& c(\delta_x,\pi^*(\delta_x))\nonumber\\
&& +\alpha\int_{\mathcal{P}(\mathcal{X})}\left(\bigg[\mathbb{E}^{\pi^*}_{\nu} \sum_{k=1}^\infty \alpha^{k-1} c(\mu_k,A_k)\bigg]\right.\nonumber\\ 
&& ~~~~~\left.\times p_{\pi^*(\delta_x)}(\delta_x,d\nu)\right)\nonumber\\
&\stackrel{\#1}{\leq}& c(\delta_x,a)\nonumber\\
&& + \alpha\int_{\mathcal{P}(\mathcal{X})}\bigg[\mathbb{E}^{\pi^*}_{\nu} \sum_{k=1}^\infty \alpha^{k-1}c(\mu_k,A_k)\bigg]p_{a}(\delta_x,d\nu),\nonumber
\end{eqnarray}
\normalsize
for all $a\in A_\mu.$ Inequality $\#1$ follows from the fact that $\pi^*$ is an optimal policy. Hence, the L.H.S of \ref{eqnCondition2Know},
\small
\begin{eqnarray}
\hspace{-1.00cm}
\mathbb{E}^{\pi^*}_{\mu}\bigg[\sum_{k=0}^\infty \alpha^k c(x_k,A_k) \mid  X_0 = x\bigg] &\stackrel{\#2}{=}& c(\delta_x,\pi^*(\mu))\nonumber\\
&&+\alpha\int_{\mathcal{P}(\chi)} \left(\bigg[\mathbb{E}^{\pi^*}_{\nu} \sum_{k=1}^\infty \alpha^{k-1} c(\mu_k,A_k)\bigg]\right.\nonumber\\
&& ~~~~~ \left.\times p_{\pi^*(\mu)}(\delta_x,d\nu)\right)\nonumber\\
&\geq& \mathbb{E}^{\pi^*}_{\delta_x}\bigg[\sum_{k=0}^\infty \alpha^k c(\mu_k,A_k)\bigg].\nonumber
\end{eqnarray}
\normalsize
Note that in the first term of the R.H.S of equality $\#2$, because of the conditioning on $\{X_0=x\}$ the CMC is in state $\delta_x$ which becomes the first argument of $c(\cdot,\cdot)$ and the controller or decision maker is only told $\mu$, which becomes the second argument of $c(\cdot,\cdot)$.
\IEEEQED
\subsection{Solving the Time-Average Cost MDP}\label{AppendixSolvingTimeAverageCostMDP}
We use the technique described in \cite{sennott1989average-cost-mdps} to show that the SLQ policy is
optimal for the long term time-averaged cost criterion as well. We consider a sequence of discount
factors $\{\alpha_n\}\uparrow1$ and the corresponding discounted optimal policies. We begin with the
following result:

\lem\label{lemDiscSubseq}(\cite{sennott1989average-cost-mdps}) Let $\{\alpha_n\}$ be a sequence of discount
factors increasing to 1, and let $\pi^*_n$
be the associated sequence of discounted optimal stationary policies. There exists a
subsequence $\{\alpha_{n_k}\}$ and a stationary policy $\pi^*$ that is a limit point of $\pi^*_{n_k}$.

The proof of Lemma \ref{lemDiscSubseq} can be found in \cite{sennott1989average-cost-mdps}. In our
case, the assertion is easily seen to be true, since the SLQ policy is optimal for every 
$\alpha\in(0,1)$. To invoke the theorem in \cite{sennott1989average-cost-mdps}, we need to first show 
that the conditions assumed therein are true. Towards this end let us first limit our MDP's state space to 
$S'=\mathbb{N}\times\mathbb{N}^N\times\mathbb{N}^{N-1}\times I-S_0$, where, 
$S_0=\{s\in\mathbb{N}\times\mathbb{N}^N\times\mathbb{N}^{N-1}\times I: v_i=v_j=0~\text{for some}~j\neq i~\text{or}~r_j\geq1~\text{for some}~j\in I\}$.
The set $S_0$ contains states where multiple coordinates of $\mathbf{V}(t)$ are zero and states with 
$r_j\geq1$. The former class of states is not allowed, since there can be at most \emph{one} transmitter
in any time slot. The latter are all transient states.
$S_0$ obviously contains the set of transient states induced by the policies under consideration, and needs
to be removed so that the resulting controlled DTMC, defined on $S'$, can be irreducible.

Now, the theorem in \cite{sennott1989average-cost-mdps} makes assumptions, which \emph{in our case} can be 
restated as follows. For
any $\alpha\in(0,1)$, and state $(q,\mathbf{V},\mathbf{r},i)\in S'$,
\begin{itemize}
 \item A: $J^*_i((q,\mathbf{V},\mathbf{r})<\infty$.
 \item B: Let\footnote{We define $\hat{\mathbf{e}}_k$ as the vector with 0's at all coordinates
 except the $k^{th}$, which is equal to $1$.} $\mathbf{V}_0=\mathbf{1}-\hat{\mathbf{e}}_1$ and define 
 $h_{i,\alpha}(q,\mathbf{V},\mathbf{r}):=J^*_i((q,\mathbf{V},\mathbf{r})-J^*_1((0,\mathbf{V}_0,\mathbf{0})$.
 There exists $N>0$, s.t.
 \begin{equation}
 h_{i,\alpha}(q,\mathbf{V},\mathbf{r})\geq-N~\forall(q,\mathbf{V},\mathbf{r},i)\in S'.  
 \end{equation}
\item C: For any given pair of states $(s,t)\in S'\times S'$, and action $a\in I$, let $P_{st}(a)$ 
denote the transition probability from state $s$ to state $t$ under action $a$ in the MDP (for further details, read Sec.\ref{AppendixSecFormulatingTheMDP}). There exists $M_i(q,\mathbf{V},\mathbf{r})\geq0$, s.t.
 \begin{equation}
 h_{i,\alpha}(q,\mathbf{V},\mathbf{r})\leq M_i(q,\mathbf{V},\mathbf{r})~\forall(q,\mathbf{V},\mathbf{r},i)\in S',~\text{and}~\alpha\in(0,1),
 \end{equation}
 and, denoting by $M_t$ the R.H.S in the above equation for state $t\in S'$,
 \begin{equation}
  \sum_{t\in S'}M_tP_{st}(a)<\infty,~\forall s\in S'~\text{and}~a\in I.
 \end{equation}

\end{itemize}

\thm\label{thmDis2Avg} If the assumptions A, B and C hold, $\pi^*$ is optimal for the Averge Cost problem as well.

Theorem \ref{thmDis2Avg} is proved in \cite{sennott1989average-cost-mdps}. To see why A is true, 
define $s(0)=\left[q,\mathbf{V},\mathbf{r},i\right]$ and consider
\small
\begin{eqnarray}
\hspace{-0.20cm}
J^*_i(s(0)) &=& \min_{\pi}\mathbb{E}_{s(0)}^\pi\sum_{n=0}^\infty\alpha^n\left[Q_{u_{n-1}}(n)+\lambda\sum_{j\neq u_{n-1}}V_j(n)+\sum_{i\in I}R_i(n)\right]\nonumber\\
&\stackrel{e}{\leq}& \min_{\pi}\mathbb{E}_{s(0)}^\pi\sum_{n=0}^\infty\alpha^n \left[ \left(Q_{u_{n-1}}(0)+n\right)+\lambda\left(n+\sum_{j\neq u_{n-1}}V_j(0)\right)\right.\nonumber\\
&& ~~~~~~~~~~~~~ \left.+\left(n+\sum_{i\in I}R_i(0)\right)\right]\nonumber\\
&\stackrel{f}{\leq}& \min_{\pi}\mathbb{E}_{s(0)}^\pi\sum_{n=0}^\infty\alpha^n \left[ \left(Q(0)+n\right)+\lambda\left(n+\sum_{j\neq u_{n-1}}V_j(0)\right)\right.\nonumber\\
&& ~~~~~~~~~~~~~ \left.+\left(n+\sum_{i\in I}R_i(0)\right)\right]\nonumber\\
&\stackrel{\star}{<}& \infty,\nonumber
\end{eqnarray}
\normalsize
where inequality $e$ follows from the fact that over $n$ slots, the backlog of queue $u_{n-1}$ cannot
have increased to more than $n$ packets over its initial backlog, i.e., $Q_{u_{n-1}}(0)+n$. In 
inequality $f$, $Q(0)=\max_{k\in I}Q_k(0)$. Finally, inequality $\star$ uses the fact that
\begin{equation}
 \sum_{n=0}^\infty n\alpha^n<\infty,~\forall\alpha\in(0,1).
\end{equation}

In Lemma.\ref{lemJstarInc}, we have already proven that $J^*$ is increasing in its first coordinate.
Now, keeping the incumbent the same, suppose one of the coordinates of $\mathbf{V}$ is increased,
\small
\begin{eqnarray}
\hspace{-0.50cm}
 J^*_i(0;\mathbf{V};\mathbf{0})&=&\lambda\sum_{k\neq i}V_k+\alpha\min_{k\neq i}\mathbb{E}J^*_j\left(B(V_k);V_k=0,\mathbf{V}_{-k}+\mathbf{1};\mathbf{0}\right)\hspace{0.0cm}\nonumber\\
&\leq& \lambda (V_l+1)+\lambda\sum_{k\neq i,l}V_k+\alpha\min\left\lbrace\min_{k\neq i,l}\mathbb{E}J^*_j\left(B(V_k);\right.\right.\nonumber\\
&& \left. V_k=0,\mathbf{V}_{-k}+\mathbf{1}+\hat{\mathbf{e}}_l;\mathbf{0}\right),\nonumber\\
&& \left. J^*_l\left(B(V_l+1);V_l=0,\mathbf{V}_{-k}+\mathbf{1};\mathbf{0}\right)\right\rbrace\nonumber\\
&=& J^*_i(0;\mathbf{V}+\hat{\mathbf{e}}_l;\mathbf{0}).\hspace{0.0cm}
\label{eqnIncV}
\end{eqnarray}
\normalsize
A similar analysis proves that $J^*$ is increasing in the coordinates of $\mathbf{r}$ also. Finally, 
starting with a different incumbent, since
\begin{equation}
J^*_l(0;\mathbf{V};\mathbf{0})=\lambda\sum_{k\neq l}V_k+\alpha\min_{k\neq l}\mathbb{E}J^*_j\left(B(V_k);V_k=0,\mathbf{V}_{-k}+\mathbf{1};\mathbf{0}\right),\\
\label{eqnVl}
\end{equation}
we see that if $V_l$ in the first equation of (\ref{eqnIncV}) is equal to $V_i$ in (\ref{eqnVl}), then
$J^*_l(0;\mathbf{V};\mathbf{0})=J^*_i(0;\mathbf{V};\mathbf{0})$. This is achieved by 
noting that all since arrivals are statistically identical and independent one can simply route the 
packets of queue $l$ to queue $i$ and vice versa in every sample path. This doesn't change the statistics
of the system in any way and proves that
\begin{eqnarray}
 J^*_i((q,\mathbf{V},\mathbf{r})-J^*_1((0,\mathbf{V}_0,\mathbf{0})=h_{i,\alpha}(q,\mathbf{V},\mathbf{r})\geq0~\forall i\in I.
\end{eqnarray}

Proving assumption C is nontrivial, and Prop.~5 (pp.629) in \cite{sennott1989average-cost-mdps} gives a
sufficient condition for C to hold. 
\lem Assume that the MDP has a stationary policy $f$ that induces an irreducible, positive recurrent 
Markov chain on the state space of the problem, viz, $S'$. Let $p_f(q,\mathbf{V},\mathbf{r},i)$ denote 
the invariant distribution of the chain and let $c_f(q,\mathbf{V},\mathbf{r},i)$ denote the 
non-negative cost in state $\mathbf{s}=\left[q;\mathbf{V};\mathbf{r};i\right]$. If 
\begin{equation}
 \sum_{\mathbf{s}}p_f(q,\mathbf{V},\mathbf{r},i)c_f(q,\mathbf{V},\mathbf{r},i)=\mathbb{E}\left[\tilde{Q}+\lambda\sum_{k\neq \tilde{i}}\tilde{V}_k+\sum_{l=1}^N\tilde{r}_l\right]<\infty
 \label{eqnStatCostFinit}
\end{equation}
then assumption C is true. Here, $\tilde{Q},\tilde{\mathbf{V}},\tilde{\mathbf{r}}~\text{and}~\tilde{i}$
are the random variables with joint distribution $p_f$.
Notice, however, from \eqref{eqnCOMDPssCost} that for all policies under consideration that are also
stabilizing,
\begin{eqnarray}
 \mathbb{E}\left[\tilde{Q}+\lambda\sum_{k\neq \tilde{i}}\tilde{V}_k+\sum_{l=1}^N\tilde{r}_l\right]
 &=& \mathbb{E}\left[\mathbb{E}\left(\sum_{k=1}^NQ_k\bigg|\tilde{Q},\tilde{\mathbf{V}}\right)\right]\nonumber\\
 &=& \mathbb{E}\sum_{k=1}^NQ_k,\hspace{2.1cm}
\end{eqnarray}
where $[Q_1,\dots,Q_N]$ is the stationary queue-length vector of the system. We have already shown that $\pi^*$ is stabilizing and know that the policy discussed in Sec.~\ref{secClassOfStableStationaryMarkovPolicies}, $\pi'$, is also stabilizing from the discussion therein. So, we set $f\equiv \pi'$, since $\pi'$ is stationary and Markov, and observe that for any $\lambda\in[0,\frac{1}{N})$ queueing delay for such policies is finite. Using Little's Theorem, we see that under $\pi'$,
$\mathbb{E}\sum_{k=1}^NQ_k=\sum_{k=1}^N\mathbb{E}Q_k<\infty$. Since $\pi'$ is stationary, this proves the claim and 
establishes the fact that $\pi^*$ is optimal for the average cost criterion as well.
\IEEEQED




\subsection{The Class of Throughput Optimal Stationary Markov Policies}\label{secClassOfStableStationaryMarkovPolicies}
\label{appendix:proofOfTOStatMarkovPolicies}
We would next like to see if the class of stabilizing 
stationary Markov policies on these state and action spaces is a singleton, and if not, whether all policies
in this class have the same cost. Obviously, both these situations would result in $\pi^*$ being a trivial
solution to the MDP problem and hence, need to be discussed in detail.
We now prove that that is, in
fact, \textit{not} the case.
\begin{prop}\label{prop:classTOStatMarkovPolicies}
For the information structure in Sec.~\ref{secRefineModel}, the class of throughput optimal, Markov scheduling policies is not a singleton.
\end{prop}

\begin{proof}
The proof proceeds by creating throughput optimal policies from $\pi^*$ by taking suboptimal actions  on finite set of states. Foster-Lyapunov theory can be used to show that the resulting policies are still throughput optimal. Hence, the set of stabilizing policies is not a singleton. A detailed proof is provided in Sec.~\ref{appendix:proofOfTOStatMarkovPolicies} in the Appendix.
\end{proof}

First, recall that $\pi^*$ itself is stabilizing as proved in Sec.~\ref{secStablePiStar}. Next, consider a policy $\pi'\in\Pi_g\cap\Pi_e$, starting with a given vector $\mathbf{V}(0)=\mathbf{V}$ ($\pi^*$ starts with the same vector). When the in-service queue empties, $\pi'$ chooses $m:=\argmax{1\leq j\leq N}V_j$ except when $\mathbf{V}\in S'\subset\mathbb{N}^N$ where $S'$ is some \textit{finite} set. On $S'$, $\pi'$ makes a suboptimal choice and deviates 
from $\pi^*$, e.g., choosing $\argmax{j\neq m}V_j$.

Clearly, this policy is stationary and Markov. We need to establish that it is also throughput optimal. Recall that $\mathbf{Q}(t)\in \mathbb{N}^N$ denotes the vector of backlogs at time $t$. Now, since $\pi^*$ itself is stabilizing,
by the Foster-Lyapunov criterion \cite[Thm~2.2.3]{fayolle-etal95constructive-theory-markov-chains},
there exists\footnote{The Foster-Lyapunov criterion is both necessary and sufficient.} a non-negative function, $y:\mathbb{N}^N\rightarrow\mathbb{R}_+$, 
$\epsilon>0$ and a finite set $A\subset\mathbb{N}^N$, such that,
\begin{equation*}
    \mathbb{E}\left[y(\mathbf{Q}(t+1))-y(\mathbf{Q}(t))| \mathbf{Q}(t)=\mathbf{q}\right]  
    \begin{cases} 
    \leq -\epsilon, \forall~\mathbf{q}\notin A, ~\text{and} \\
    <\infty, \forall~\mathbf{q}\in A.
\end{cases}
\end{equation*}
By expanding $A$ to $A'=A\cup S'$, and using the same Lyapunov function $y$, the Foster-Lyapunov criterion
shows that $\pi'$ is also stabilizing. Thus, considering the number of such subsets $S$ that can be 
chosen, we see that $\Pi_g\cap\Pi_e$ contains \emph{infinitely many} stabilizing stationary Markov policies.
\IEEEQED

Moreover, these policies do not have the same cost! Suppose we define the set $S'$ above to be the set
$\{0,1,\cdots,k\}^N$. This means $\pi'$ will make sub-optimal decisions whenever 
$0\leq V_i(t)\leq k,\forall i\in I$. Fig.\ref{figSubOptPis} shows the difference in cost between $\pi^*$ and 
three other policies $\pi_1, \pi_2$ and $\pi_3$ corresponding to $k=40,~50$ and $80$. Note that
all three policies are in $\Pi_g\cap\Pi_e$ and are stabilizing, stationary and Markovian.

\begin{figure}[tb]
\centering
\includegraphics[height=5.5cm, width=8.5cm]{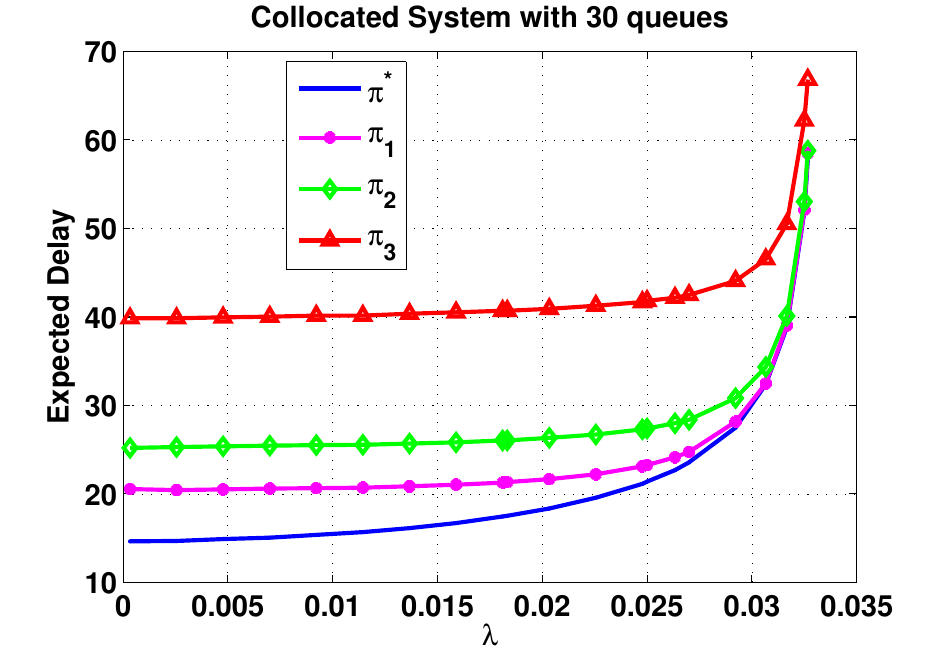}
\caption{$\pi_1, \pi_2$ and $\pi_3$ are suboptimal stabilizing stationary Markov policies, corresponding to $k=40,50\text{ and }80$ respectively. This figure clearly shows that the class of stabilizing stationary Markov policies is \emph{not} a singleton.}
\label{figSubOptPis}
\end{figure}


\subsection{Alleviating Short term Unfairness: The K-Longest Expected Queue (KLEQ) Policy}\label{secAlleviateShortTermUnfairness}
%
While exhaustive service is delay optimal, it obviously exhibits short-term unfairness, especially at high arrival rates with the incumbent being served
for long periods while other queues remain starved. While this might not be a cause of too much concern at low loads, near saturation, busy periods
tend to be very long and starving queues might not be a good idea, especially if the application requires timely delivery of packets.

One of the ways of quantifying unfairness is by defining a ``Fairness Index'', $\mathcal{J}(t)$, as in \cite{jain84fairness-resource-allocation}. In our model, suppose we define $x_i(t),1\leq i\leq N$ to be the fraction of time-slots in $\{0,1,\cdots,t\}$ during which queue $i$ was scheduled. Then, 
the fairness index at time $t$ is defined as 
\begin{equation}
 \mathcal{J}(t)=\frac{\left(\sum_{i=1}^Nx_i(t)\right)^2}{N\sum_{i=1}^Nx_i(t)^2}.
 \label{eqnJfi}
\end{equation}
It can be shown, using the Cauchy-Schwarz inequality, that $\frac{1}{N}\leq\mathcal{J}(t)\leq1,~\forall t$ and the closer $\mathcal{J}(t)$ is to $1$, the fairer the allocation/algorithm \cite{jain84fairness-resource-allocation}.
 
Fig.~\ref{figJfiTdExPoll} compares the fairness indices of TDMA and Cyclic Exhaustive Service, near saturation (for a symmetric system with $N=30$ queues, $\lambda<\frac{1}{N}\approx 0.033$). The computation is over a period of $2\times10^4$ slots (we wait for the queue length processes to attain stationarity and \emph{only then} record $\mathcal{J}(t)$ for $2\times10^4$ slots) by which time, the indices attain their maximum value of $1$.
We see that although cyclic exhaustive service is delay optimal, it does exhibit short-term unfairness.
The curious form of the curves corresponding to TDMA in figures \ref{figJfiTdExPoll} and \ref{figKLEQimprovesFairness} respectively is due to the cyclic nature of the underlying TDMA schedule. 
Consider the TDMA fairnes curve. Since every queue gets scheduled exactly once every $N$(= 30) slots, $\mathcal{J}(t)$ shows the corresponding periodicity (notice that the swings have a period of 30). So, over slots 31 through 60 for example, $\mathcal{J}(t)$ first decreases, since the first few queues get scheduled for a second time (while the others have been scheduled only once over slots 1 through 30) and then increases as this disparity reduces. 

\begin{figure*}[tbh]
\begin{subfigure}{.5\textwidth}
\centering
\includegraphics[height=4.5cm, width=7.5cm]{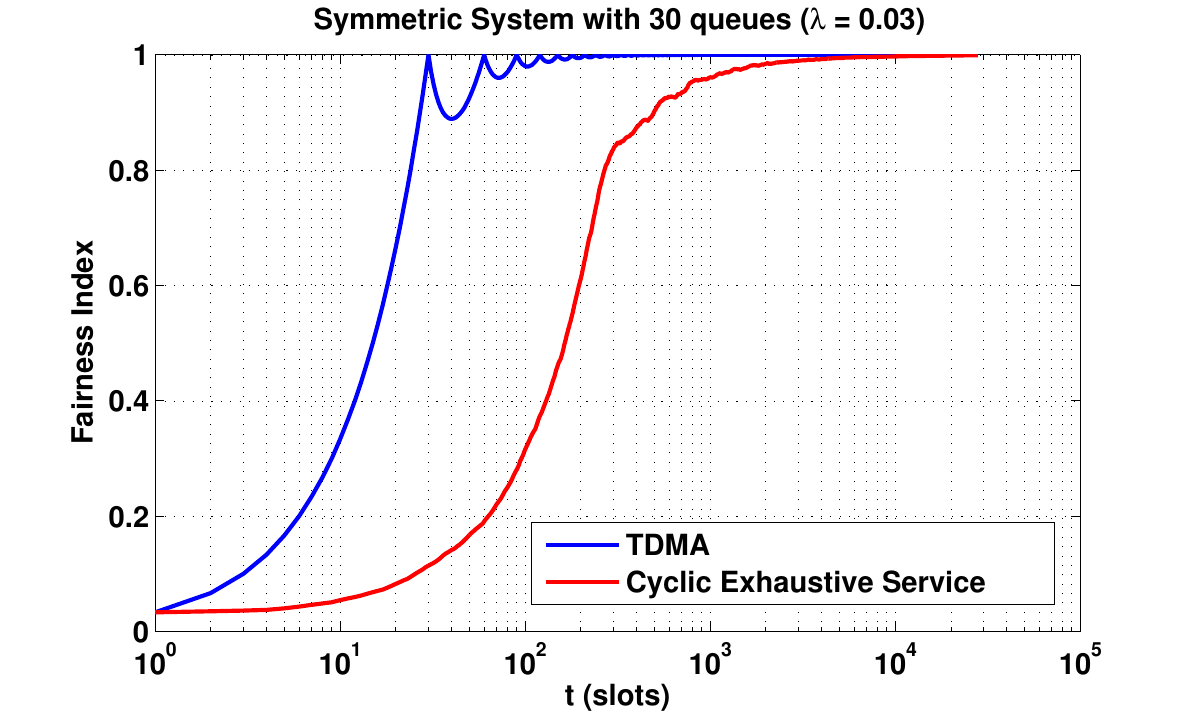}
\caption{Comparing short term unfairness with TDMA and the cyclic exhaustive service policy, i.e., $\pi^*$. TDMA clearly outperforms $\pi^*$ in this regard.}
\label{figJfiTdExPoll}
\end{subfigure}
\begin{subfigure}{.5\textwidth}
\centering
\includegraphics[height=4.5cm, width=7.5cm]{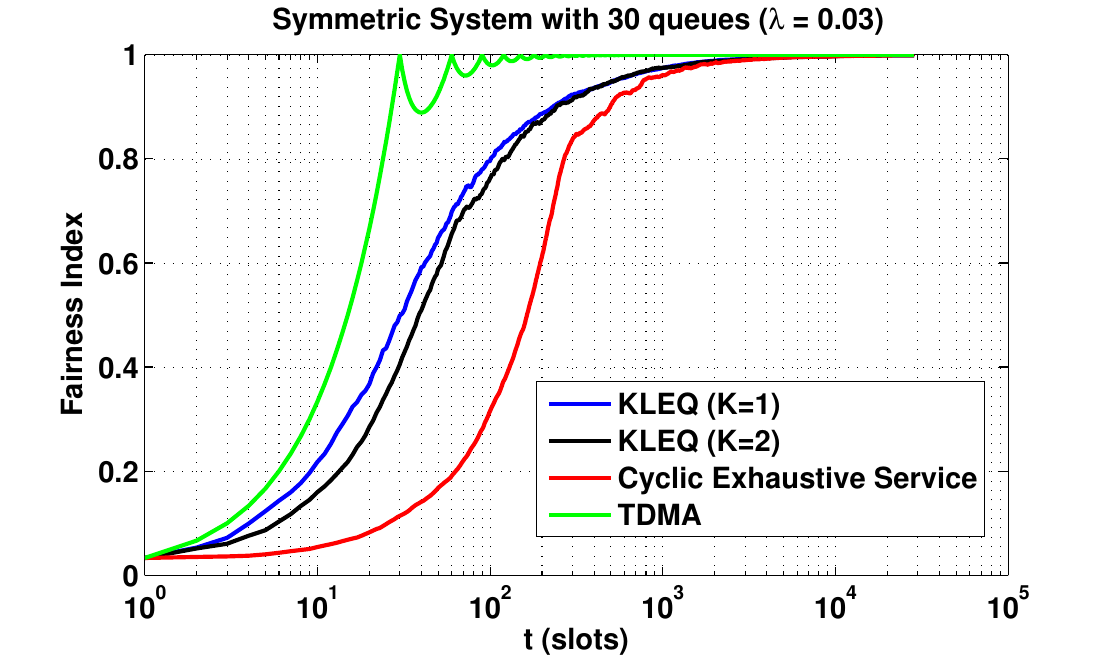}
\caption{Improvement in fairness with the KLEQ policy.}
\label{figKLEQimprovesFairness}
\end{subfigure} 
\caption{Alleviating short term unfairness through the KLEQ policy.}
\end{figure*}

Towards alleviating this problem, consider the space $\Pi_K\subset\Pi$ of policies that do not permit a queue to transmit for more than $K\geq1$ slots consecutively. 
Such policies are called $K$ limited service policies. 
It is important to note that not all policies in $\Pi_K$ are throughput optimal. One prominent example is cyclic service, where, after serving queue $i$ (for at most $K$ slots), the server visits queue $i~$mod~$N+1$; refer \cite{takagi88queuing-analysis-polling-models} for more details about the capacity regions  of this policy and some of its variants. This means that the problem of deciding to which queue the system (server) should switch once a queue is empty or served for $K$ consecutive slots, \emph{is non trivial}. Here, we propose a policy that is both fair \emph{and} throughput optimal. 
%
%
The proposed policy $\pi_K\in\Pi_K$ serves a queue for a maximum of $K\geq1$ slots, and thereafter chooses to serve $\argmax{1\leq i\leq N} Q_i(t-V_i(t))+\lambda_i V_i(t)$. 
%

\prop\label{propKlimitedStable} For all $K\geq1$ and $\boldsymbol{\lambda}\in\mathbb{R}^N_+$ such that $\sum_{i=1}^N\lambda_i<1$, $\pi_K$ is stabilizing, i.e., \emph{throughput optimal}.

\begin{IEEEproof}
 The proof involves bounding the $K$ slot conditional expected drift of a \emph{novel} Lyapunov function. 
Let $\mathbb{Z}_+$ denote the non negative integers. The proof uses a simple mean drift argument with a novel Lyapunov function. Define $L:\mathbb{Z}_+^{2N}\mapsto\mathbb{R}_+$ as
\begin{equation}
L(s(t)):=\sum_{j=1}^N\bigg[Q_j(t-V_j(t))+\lambda_jV_j(t)\bigg]^2+V_j(t)\lambda_j(1-\lambda_j).
\label{eqnKlimitedLyapunovDef}
\end{equation}
We will prove that a policy that is clearly suboptimal (refer Fig.~\ref{figLeqAndPseudoLeqDelays}), again, w.r.t. delay when compared to $\pi_K$ (meaning it idles more than $\pi_K$) is also stabilizing. This new policy, called \enquote{\color{blue}Gated K Limited service}(GKLS) policy and denoted by $\xi_K$, serves only those packets that were in the chosen queue at the beginning of service, for a maximum of $K$ slots. Also, once a queue is chosen for service, the server stays at it (idling, if necessary), for $K$ slots at the end of which $\argmax{1\leq i\leq N} \left(Q_i(t-V_i(t))+\lambda_i V_i(t)\right)$ is chosen for service for the next $K$ slots. 

\begin{figure}[tb]
\centering
\includegraphics[height=5.5cm, width=8.0cm]{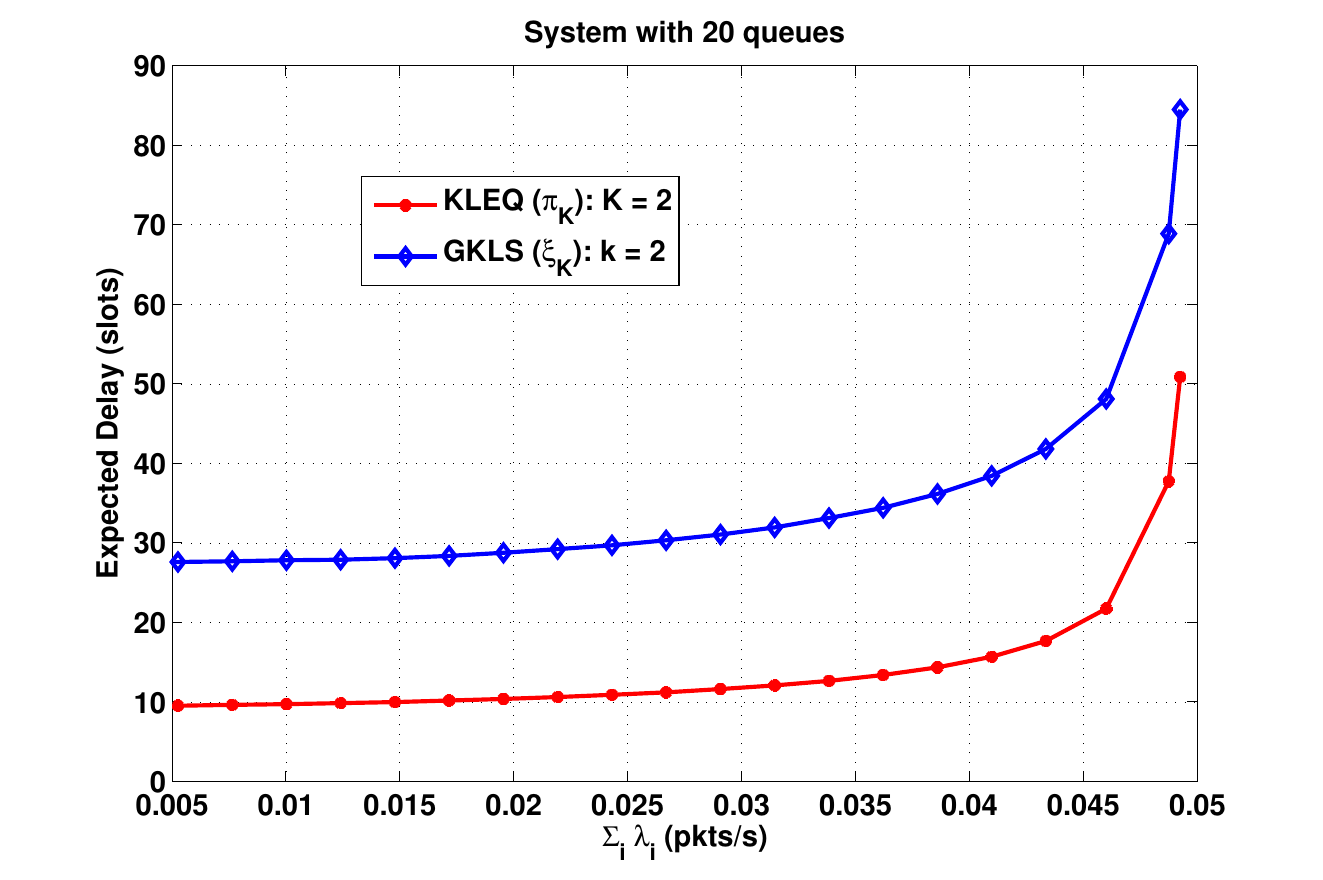}
\caption{Illustrating the suboptimality of the Gated K Limited service policy ($\xi_K$).}
\label{figLeqAndPseudoLeqDelays}
\end{figure}

For example, suppose $K=5$ and queue $3$ has been chosen for service. At the beginning of service, suppose $Q_3=2$ and over the course of the $5$ slots, it receives $4$ more packets. Then $\xi_K$ serves only the $2$ that were in Queue $3$ at the beginning of service, does not serve any of the $4$ new packets and \emph{idles} for the remaining $K-2(=3)$ slots. This is clearly suboptimal compared to $\pi_K$ which would have served $5$ packets and left Queue~$3$ with just $1$ packet at the end of $K=5$ slots.

Suppose at time $t$, $\xi_K$ switches to queue $i$. For all other queues, after K slots, $V_l(t+K)=V_l(t)+K$, and $Q_l(t+K-V_l(t+K))=Q_l(t-V_l(t))$. Therefore, at the end of $K$ slots, the new state is $s(t+K)=[Q_1(t-V_1(t)),\dots, Q_{i-1}(t-V_{i-1}(t)), Q_i(t+K), Q_{i+1}(t-V_{i+1}(t)),\dots,Q_N(t-V_N(t)),V_1(t)+K,\dots,V_{i-1}(t)+K,0,V_{i+1}(t)+K,\dots,V_N(t)+K]$.
Invoking the linearity of conditional expectation, computing the expected conditional drift can be split into contributions by the queue chosen for service at time $t$, i.e., queue $i$, and that by the other queues. For simplicity we drop the time indices and denote $V_j(t)$ by $V_j$ for the remainder of the proof. 
\begin{eqnarray}
&& \mathbb{E}\left[L(s(t+K))-L(s(t))|s(t)\right]\nonumber\\
&=& \mathbb{E}\bigg[Q^2_i(t+K)-\left(Q_i(t-V_i)+ \lambda_iV_i\right)^2- V_i\lambda_i(1-\lambda_i)\bigg|s(t)\bigg]+ \nonumber\\
&& \mathbb{E}\bigg[\sum_{j\neq i}\left(Q_j(t-V_j)+ \lambda_j\left(V_j+K\right)\right)^2+ \left(V_j+K\right)\lambda_j(1-\lambda_j)-\nonumber\\
&& \sum_{j\neq i}\left(\left(Q_j(t-V_j)+ \lambda_jV_j\right)^2+V_j\lambda_j(1-\lambda_j)\right)\bigg|s(t)\bigg]
\label{eqnKlimitedConditionalDriftContributions}
\end{eqnarray}
The contribution to $L(s(t+K))$ of all queues $j$ other than $i$ is 
\begin{eqnarray}
&& \sum_{j\neq i}\left(Q_j(t-V_j)+\lambda_j\left(V_j+K\right)\right)^2+\left(V_j+K\right)\lambda_j(1-\lambda_j)\nonumber\\
&=& \sum_{j\neq i}\left(Q_j(t-V_j)+\lambda_jV_j\right)^2+2K\left(Q_j(t-V_j)+\lambda_jV_j\right)\lambda_j\nonumber\\
&& +K^2\lambda_j^2+ V_j\lambda_j(1-\lambda_j)+K\lambda_j(1-\lambda_j).
\label{eqnKlimitedLOftPlusKeveryoneElse}
\end{eqnarray}
Using \eqref{eqnKlimitedLOftPlusKeveryoneElse}, the second term on the R.H.S of \eqref{eqnKlimitedConditionalDriftContributions} becomes
\begin{eqnarray}
&& \mathbb{E}\bigg[\sum_{j\neq i}\left(Q_j(t-V_j)+\lambda_j\left(V_j+K\right)\right)^2+\left(V_j+K\right)\lambda_j(1-\lambda_j)\nonumber\\
&& -\sum_{j\neq i}\left(Q_j(t-V_j)+\lambda_jV_j\right)^2+V_j\lambda_j(1-\lambda_j)\bigg|s(t)\bigg]\nonumber\\
&=& \mathbb{E}\bigg[2K\left(Q_j(t-V_j)+\lambda_jV_j\right)\lambda_j+K^2\lambda_j^2+K\lambda_j(1-\lambda_j)\bigg|s(t)\bigg] \hspace{0.0cm}\nonumber\\
&=& 2K\left(Q_j(t-V_j)+\lambda_jV_j\right)\lambda_j+K^2\lambda_j^2+K\lambda_j(1-\lambda_j),\hspace{0.0cm}
\label{eqnKlimitedContributionOfOthers}
\end{eqnarray}
where the last equality follows since the term inside the expectation is a function of $s(t)$. Proceeding to queue $i$, we see that
\begin{eqnarray}
\hspace{-1.00cm}
Q^2_i(t+K) &\stackrel{(*a)}{=}& \left(\left(Q_i(t-V_i)+A_i[t-V_i+1,t]-K\right)^++A_i[t+1,t+K]\right)^2,\nonumber\\
&\stackrel{(*b)}{\leq}& \left(Q_i(t-V_i)+A_i[t-V_i+1,t]\right)^2+K^2+\left(A_i[t+1,t+K]\right)^2\nonumber\\
&& -2\left(Q_i(t-V_i)+A_i[t-V_i,t]\right) (K-A_i[t+1,t+K]),
\label{eqnKLimitedQiSquare}
\end{eqnarray}
where, in equality $(*a)$, one first needs to note that queue $i$ has never been served since $t-V_i$. Also, $A_i[x,y]$ denotes the cumulative arrivals to queue $i$ over the time period $[x,y]$, and $A_i[x,x]$ is the arrival over the single slot $x$. Specifically, when the arrival process is Bernoulli, $A_i[x,y]$ has the distribution Binomial$(y-x+1,\lambda_i)$. Also, $(*b)$ follows from the fact that for any three nonnegative numbers $x,y~\text{and}~z$, $((x-y)^++z)^2\leq x^2+y^2+z^2-2x(y-z)$. 
In what follows, we use the fact that if a random variable $X\sim$Binomial($m,p$), then $\mathbb{E}X^2=mp(1-p)+m^2p^2$.
\begin{eqnarray}
\hspace{-0.40cm}
\mathbb{E}\bigg[Q^2_i(t+K)\bigg|s(t)\bigg] &\leq& \mathbb{E}\bigg[\left(Q_i(t-V_i)+A_i[t-V_i+1,t]\right)^2+ K^2\nonumber\\
&& + \left(A_i[t+1,t+K]\right)^2- \nonumber\\
&& 2\bigg(\left(Q_i(t-V_i)+A_i[t-V_i+1,t]\right)\nonumber\\
&& \times \left(K-A_i[t+1,t+K]\right)\bigg)\bigg|s(t)\bigg]\nonumber\\
&\stackrel{(*c)}{=}& Q^2_i(t-V_i)+ 2\lambda_iV_iQ_i(t-V_i)+ V_i\lambda_i(1-\lambda_i)\nonumber\\
&& +V_i^2\lambda_i^2+K^2+ K\lambda_i(1-\lambda_i)+ K^2\lambda_i^2\nonumber\\
&& -2\left(Q_i(t-V_i)+\lambda_iV_i\right)K\left(1-\lambda_i\right).\nonumber
\end{eqnarray}
In $(*c)$, one needs to note that the arrivals that enter queue $i$ \emph{after} the last time queue $i$ was served, viz., $A_i[t-V_i,t+K]$, are independent of $Q_i(t-V_i)$. With this, the contribution of queue $i$ to the expected drift (the first term on the R.H.S of \eqref{eqnKlimitedConditionalDriftContributions}) is
\small
\begin{eqnarray}
&& \mathbb{E}\bigg[Q^2_i(t+K)-\left(Q_i(t-V_i)+\lambda_iV_i\right)^2-V_i\lambda_i(1-\lambda_i)\bigg|s(t)\bigg]\nonumber\\
&& \leq K^2+ K\lambda_i(1-\lambda_i)+ K^2\lambda_i^2\nonumber\\
&& +2K\lambda_i\left(Q_i(t-V_i)+\lambda_iV_i\right)-2K\left(Q_i(t-V_i)+\lambda_iV_i\right)
\label{eqnKlimitedContributionOfQueuei}
\end{eqnarray}
\normalsize
We are now in a position to compute the expected drift. Define $\epsilon:=\left(1-\sum_j\lambda_j\right)$, strictly positive, by assumption. Using \eqref{eqnKlimitedContributionOfOthers} and \eqref{eqnKlimitedContributionOfQueuei} in \eqref{eqnKlimitedConditionalDriftContributions}, we get
\small
\begin{eqnarray}
\hspace{-0.25cm}
\mathbb{E}\left[L(s(t+K))-L(s(t))|s(t)\right]&\leq &\left(K(K-1)\sum_{j=1}^N\lambda_j^2+ K\sum_{j=1}^N\lambda_j+ K^2\right)\nonumber\\
&&-2K\left(Q_i(t-V_i)+ \lambda_iV_i\right)\left(1-\sum_{l=1}^N\lambda_l\right)\nonumber\\
&<& -\epsilon,\nonumber
\end{eqnarray} 
\normalsize
whenever
\small
\begin{equation}
\hspace{-0.25cm}
max_{i}\left(Q_i(t-V_i)+ \lambda_iV_i\right)>\frac{\left(K(K-1)\sum_{j=1}^N\lambda_j^2+ K\sum_{j=1}^N\lambda_j+ K^2\right)}{2K\epsilon}+ \frac{1}{2K},\nonumber
\end{equation}
\normalsize
Invoking the Foster Lyapunov theorem \cite{fayolle-etal95constructive-theory-markov-chains} we see that the DTMC $\{s(Kt),t\geq0\}$ is positive recurrent. To show that $s(t)$ is positive recurrent as well, we invoke the following lemma.

\lem\label{lemPRecurrenceOfOriginalDTMC} Let $\{X(t),t\geq0\}$ be an aperiodic and irreducible DTMC over a countable state space. Define $Y(t):=X(Kt),\forall t\geq0,$ for some positive integer $K$. If $Y(t)$ is positive recurrent, then so is $X(t)$.

\pf Given any $j\in\mathbb{N}$, suppose the mean recurrence time of $j$ in $Y(t)$ is denoted by $\nu^{(K)}_j$, and in $X(t)$ by $\nu_j$. Clearly, starting from any state $j$, every time $X(Kt)$ hits $j$, so does $X(t)$. This means that $\nu^{(K)}_j\geq\nu_j$ and since $\nu^{(K)}_j<\infty$, so is $\nu_j$, showing that $X(t)$ is also positive recurrent. \IEEEQED

Thus, $s(t)$ is also positive recurrent, and the proof of Prop.~\ref{propKlimitedStable} concludes. 
\end{IEEEproof}
In Sec.~\ref{secFadingAndChannelErrors} in the Appendix, we show how the above scheduler can be modified to handle Channel Errors and Fading. As can be seen from Fig.~\ref{figKLEQimprovesFairness}, the KLEQ algorithm indeed improves fairness. In fact, with $k=1$ and $k=2$, KLEQ is able to achieve a fairness of 50$\%$, about 100 and 90 slots before $\pi^*$ respectively. However, even more dramatic improvement in fairness is not possible, since stability pulls the protocol towards repeatedly scheduling the same queue over consecutive sets of $K$ slots.

\subsection{The ZMAC Protocol: Details}\label{appendix:ZMACprotocolDetails}

 
  \begin{algorithm}[ tbh ]
  
    \floatname{algorithm}{Protocol}
      \allowdisplaybreaks
      	\caption{{\bf ZMAC}} \label{protocol:ZMAC}
      	\begin{algorithmic}[1]
      	\State \textbf{Input:} $N,T_c${\color{blue}\Comment{QZMAC runs in parallel at every Queue~$j,~j\in[N]$ in the network}}
       	
      	\State \textbf{Init:} $t\leftarrow 0,\tau\leftarrow1,PU\leftarrow1$
      	{\color{blue}\Comment{$\tau$ keeps track of the minislot number}}
      	\While{$t\geq0$}
       	
      	$\tau\leftarrow1$ {\color{blue}\Comment{Keeps track of minislot number.}}
      	 \If{$PU==j$ \&\& $Q_{j}(t)>0$}{\color{blue}\Comment{If you are the PU}}
      	    \State Tx pkt.
      	    \State GOTO \ref{ZMAC:incrementPU}
      	    \Else {\color{blue}\Comment{PU empty $\Rightarrow$ network enters contention}}\label{zmac:ContentionBegins}
      	    \label{zmacProtocol:incont}
      	        \State {\bf Over} $\tau\in\{2,3,\cdots,T_c+1\}$ {\bf do}
      	        \If{$Q_j(t)>0$}
      	            \State $U_j\sim Unif\left(\lbrace2,3,\cdots,T_c+1\rbrace\right)$
      	            \State Wait for $U_j-1$ minislots. 
      	            \State If \texttt{CCA == SUCCESS}, Tx Pkt.
      	            \State GOTO \ref{ZMAC:incrementPU}
      	        \EndIf\label{zmac:ContentionEnds}
      	  \EndIf
      	  \State\label{ZMAC:incrementPU} $PU\leftarrow PU\text{ mod }N + 1$ \label{ZMAC:NextPUCyclicTDMA}
      	 \State\label{ZMAC:incrementTime} $t\leftarrow t+1$ {\color{blue}\Comment{Next time slot}}
      	\EndWhile
      	\end{algorithmic} 
 \end{algorithm}

\subsection{The EZMAC Protocol}
\label{appendix:secEZMACProtocol}
Recall the description of the ZMAC protocol given in Sec.~\ref{secProtocolDesign}. To motivate the development of EZMAC, consider ZMAC's contention mechanism \emph{in isolation}. If the contention winner can retain rights to access the channel for exactly one slot (which is what happens in ZMAC), each time the PU is empty, contention is needed to decide the next winner. We shall, henceforth, refer to the contention mechanism that allows a contention winner to transmit only in one slot as \enquote{ALOHA.} That ALOHA wastes a lot of slots in contention is clearly demonstrated in the large difference in delay in Fig.~\ref{figAlRal10}. This problem is solved as follows.


  
The proposed protocol EZMAC differs from ZMAC in the contention resolution (CR) portion. Here, once the winner of a contention is determined, it is allowed to transmit in all slots without a PU until it empties. It is assumed that at the end of the winner's transmission, the packet contains an end of transmission message (a bit in the header, perhaps) that can be decoded by the other users. Following the terminology in the literature (\cite{hofri1983analysis-raloha}, for example), we refer to the contention mechanism that allows a contention winner to transmit until it is empty, \enquote{RALOHA,} which stands for ALOHA \emph{with Reservation}. 

 \begin{figure}[tb]
 \centering
 \includegraphics[height=5.05cm, width=7.2cm]{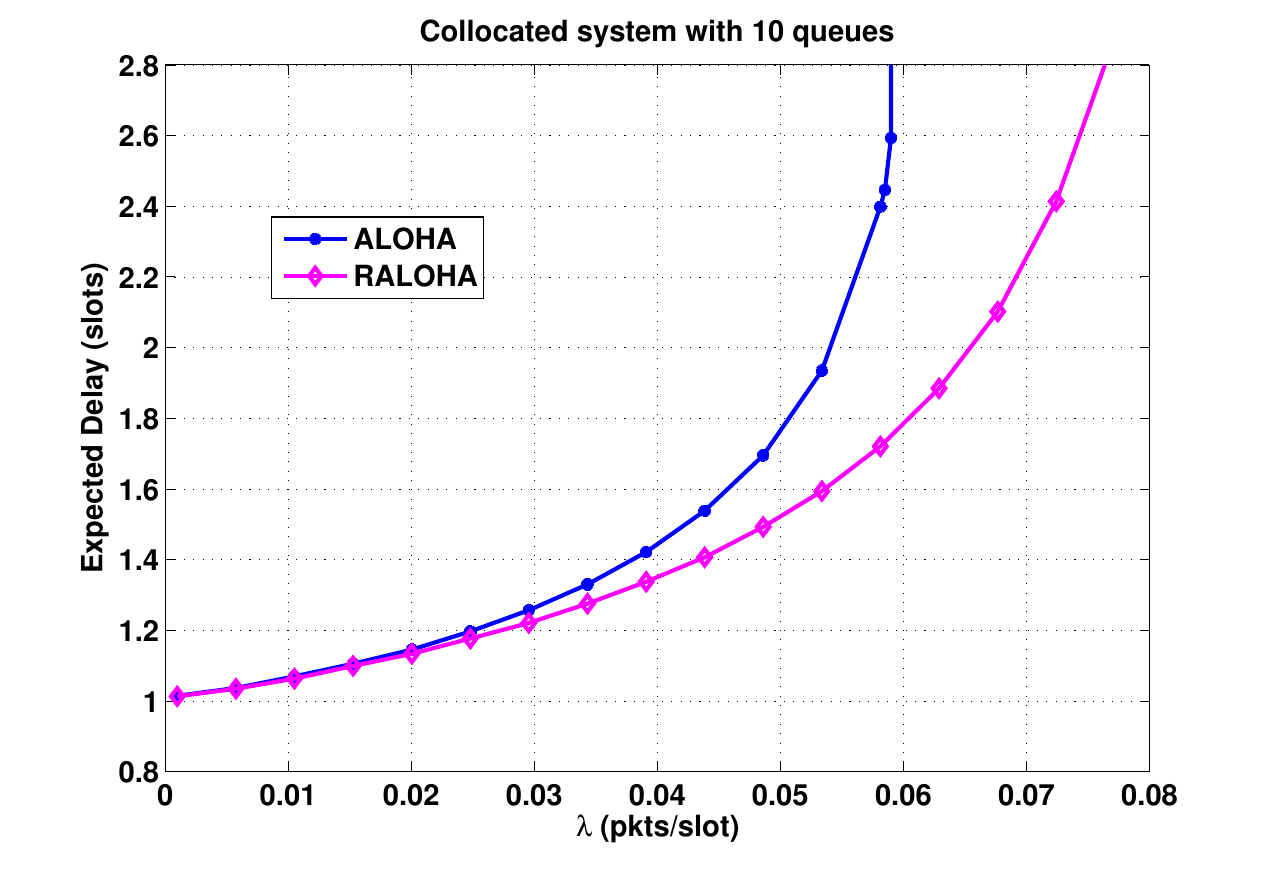}
 \caption{Mean delay with ALOHA and RALOHA in a \emph{symmetric} system with 10 queues.}
 \label{figAlRal10}
 \end{figure}
%
%

If the current slot is not the user's TDMA slot, it checks if the PU is transmitting. If 
not and if the SU is also empty, the user contends for access to the channel as described above. The
winner keeps transmitting in all slots where the PU is empty, until it is itself empty.
Thereafter, in the next slot with an empty PU, contention begins to determine the next SU.

\subsection{Channel Errors and Fading}\label{secFadingAndChannelErrors}
Until now, we have assumed that all transmission attempts succeed $w.p.1$. In this section, we consider a system similar to the one described in \cite{tassiulas93server-allocation-randomly-varying}. In addition to the system described in Sec.~\ref{secSysMod}, we also specify that every transmission from queue $i$ succeeds with probability $p_i$ independent of all other transmissions. Specifically, fading or transmission errors are independent across time slots. This necessitates a revision of Eqn.~\ref{eqn:qEvolution}, provided in Eqn.~\eqref{eqnQEvolution2WithFading} wherein the random variables $C_i(t)$, that represent the fade state of the channel between node $i$ and the receiver, are modelled as Bernoulli random variables i.i.d across time and independent across queues. 
\begin{eqnarray}
 Q_i(t+1)&=&(Q_i(t)-C_i(t)D_i(t))^++A_i(t+1),\nonumber\\
 Q_i(0)&=&A_i(0).
\label{eqnQEvolution2WithFading}
\end{eqnarray}
The coefficient of $D_i(t)$ in \eqref{eqnQEvolution2WithFading} is $-C_i(t)\in\{0,-1\}$, meaning Queue $i$'s transmissions are assumed to succeed whenever it is scheduled ($D_i(t)=1$) and its channel is \enquote{ON} ($C_i(t)=1$). 
%
We assume the system only gets to know whether a transmission succeeded or not at the end of a slot and has no knowledge of the current channel state at the scheduling instants. Since each packet at Queue~$i$ takes, on the average, $1/p_i$ slots to be transmitted (i.e., its service time is Geometric with mean $1/p_i$), the capacity region of such a system is given by 
\begin{equation}
\Lambda:=\bigg\{\boldsymbol{\lambda}\in\mathbb{R}^N_+\bigg|\sum_{i=1}^N\frac{\lambda_i}{p_i}<1\bigg\},
\end{equation}
\begin{table*}[t]
  \centering
  \begin{tabular}{c|c|c|c|c}
  Time Slot & $\mathbf{V}_{1}(\cdot)$ & $\mathbf{V}_{2}(\cdot)$ & $\mathbf{V}_{3}(\cdot)$ & $\mathbf{V}_{4}(\cdot)$\\
            & (incumbent)             &                         &                        & ($\argmax{i}V_i(t)=4$) \\
  \specialrule{2pt}{2pt}{0pt}
    $t$ & $[0,9,11,12]$ & $[0,9,11,12]$ & $[0,9,11,12]$ & $[0,9,11,12]$ \\
    \hline
    $t+1$ & $[0,10,12,13]$ & $[0,10,12,13]$ & $[0,10,12,13]$ & $[0,10,12,13]$ \\
    \hline
    $t+2$ & $[0,11,13,14]$ & {\color{red}$[1^\ast,11,13,0]$} & $[0,11,13,14]$ & $[0,11,13,14]$ \\
     &  & {\color{blue}$[0,11,13,14]$} &  & \\
    \hline
    $t+3$ & $[0,12,14,15]$ & $[0,12,14,15]$ & $[0,12,14,15]$ & {\color{red}$[1^*,12,14,0]$} $\leftarrow$ Queue~4 nonempty \\
    \hline
    $\vdots$ \\
    \hline
    $t'$ & $[0,11,13,14]$ & $[0,11,13,14]$ & $[0,11,13,14]$ & {\color{red}$[1^*,11,13,0]$} $\leftarrow$ Queue~4 empty \\
     &  &  &  & {\color{blue}$[0,11,13,14]$} \\
    \specialrule{2pt}{2pt}{0pt}
  \end{tabular}
  \caption{Illustrating the three types of CCA errors and the resulting $V$-vector misalignment. The network here comprises 4 nodes and the 4 columns represent \emph{local} copies of the $V$-vector (i.e., $\mathbf{V}_i(t)$ is the copy of the $V$-vector at Node~$i$). In the situation considered here, Node~1 is the incumbent and $\argmax{1\leq i\leq4}V_i(t)=4$. The coordinate at which CCA error occurs is denoted by an asterisk ($\ast$). A \emph{\color{red}red} to \emph{\color{blue}blue} transition implies the node detects (infers) a misalignment and also corrects it within the same time slot.}
  \label{table:cca-errors-v-misalignments}
\end{table*}
It can be shown that the proof of Prop.~\ref{propSLEQunequal} remains valid even when the service time of Queue $i$ is $B_i$ and $\sum_{i=1}^N\lambda_i\mathbb{E}B_i<1$. Setting $\mathbb{E}B_i=\frac{1}{p_i}$ we see that the LEQ policy, clearly, is still throughput optimal. Note that the LEQ scheduler \emph{does not} need to know the fading probabilities $\{p_1,p_2,\cdots,p_N\}$ to stabilize the system.
\subsection{Handling Alarm Traffic: The QZMAC$_a$ protocol}\label{appendix:DesigningQZwithAlarms}
Alarm traffic is particularly important in \emph{in-network processing}, wherein the data generated by the sensors is partially processed within the network before reaching the sink etc. \cite{liu-etal2003collaborative-in-network-processing,yao-gehrke02cougar-approach-in-network-query-processing,ye-etal2002energy-efficient-mac-in-network}. This is done, as mentioned earlier, to detect abnormal behavior in the processes begin monitored by the WSN and when such conditions are detected, the network generates alarm packets. 
To study the effects of applications generating alarm packets that need to be delivered with utmost priority, we first modify the system described in Sec.~\ref{secSysMod}. 
 \begin{figure*}[ht]
\begin{subfigure}{.5\textwidth}
\centering
\tikzset{every picture/.style={line width=0.75pt}} 
\resizebox{7.50cm}{2.3cm}{
\begin{tikzpicture}[x=0.5pt,y=0.5pt,yscale=-1,xscale=1]

\draw[very thick,black]    (181,111) -- (391.26,111) ;
\draw[very thick,black]    (391.26,111) -- (391.26,213.35) ;
\draw[very thick,black]    (181,213.35) -- (391.26,213.35) ;
\filldraw[very thick, fill=blue!60!white!40, draw=black]  (373.74,115.93) .. controls (380.26,115.88) and (385.58,121.12) .. (385.64,127.64) -- (386.17,196.09) .. controls (386.22,202.6) and (380.98,207.93) .. (374.47,207.98) -- (339.06,208.26) .. controls (332.55,208.31) and (327.22,203.07) .. (327.17,196.55) -- (326.63,128.1) .. controls (326.58,121.58) and (331.82,116.26) .. (338.34,116.21) -- cycle ;
\filldraw[very thick, fill=blue!60!white!40, draw=black] (308.74,115.39) .. controls (315.26,115.34) and (320.58,120.58) .. (320.63,127.1) -- (321.17,195.55) .. controls (321.22,202.07) and (315.98,207.39) .. (309.46,207.44) -- (274.06,207.72) .. controls (267.54,207.77) and (262.22,202.53) .. (262.16,196.01) -- (261.63,127.57) .. controls (261.58,121.05) and (266.82,115.72) .. (273.34,115.67) -- cycle ;
\filldraw[very thick, fill=blue!60!white!40, draw=black]  (242.73,115.86) .. controls (249.25,115.81) and (254.58,121.05) .. (254.63,127.57) -- (255.16,196.01) .. controls (255.22,202.53) and (249.97,207.86) .. (243.46,207.91) -- (208.05,208.19) .. controls (201.54,208.24) and (196.21,203) .. (196.16,196.48) -- (195.62,128.03) .. controls (195.57,121.51) and (200.81,116.19) .. (207.33,116.14) -- cycle ;
\draw   (106.25,85.13) .. controls (161.69,86.12) and (134.81,129.47) .. (179.87,132.28) ;
\draw [shift={(181.26,132.35)}, rotate = 182.44] [color={rgb, 255:red, 0; green, 0; blue, 0 }  ][line width=0.75]    (10.93,-3.29) .. controls (6.95,-1.4) and (3.31,-0.3) .. (0,0) .. controls (3.31,0.3) and (6.95,1.4) .. (10.93,3.29)   ;
\draw    (107.25,232.13) .. controls (162.69,233.12) and (136.79,189.24) .. (181.87,190.31) ;
\draw [shift={(183.26,190.35)}, rotate = 182.44] [color={rgb, 255:red, 0; green, 0; blue, 0 }  ][line width=0.75]    (10.93,-3.29) .. controls (6.95,-1.4) and (3.31,-0.3) .. (0,0) .. controls (3.31,0.3) and (6.95,1.4) .. (10.93,3.29)   ;
\draw    (105.25,158.13) -- (174.25,158.13) ;
\draw [shift={(176.25,158.13)}, rotate = 180] [color={rgb, 255:red, 0; green, 0; blue, 0 }  ][line width=0.75]    (10.93,-3.29) .. controls (6.95,-1.4) and (3.31,-0.3) .. (0,0) .. controls (3.31,0.3) and (6.95,1.4) .. (10.93,3.29)   ;
\draw    (468,160) -- (511.25,160.13) ;
\draw [shift={(513.25,160.13)}, rotate = 180.17] [color={rgb, 255:red, 0; green, 0; blue, 0 }  ][line width=0.75]    (10.93,-3.29) .. controls (6.95,-1.4) and (3.31,-0.3) .. (0,0) .. controls (3.31,0.3) and (6.95,1.4) .. (10.93,3.29)   ;
\draw[very thick,black]   (418,160) .. controls (418,146.19) and (429.19,135) .. (443,135) .. controls (456.81,135) and (468,146.19) .. (468,160) .. controls (468,173.81) and (456.81,185) .. (443,185) .. controls (429.19,185) and (418,173.81) .. (418,160) -- cycle ;
\draw   (80.65,178.57) .. controls (80.64,176.6) and (82.22,175) .. (84.19,175) .. controls (86.16,175) and (87.77,176.6) .. (87.79,178.57) .. controls (87.81,180.54) and (86.23,182.13) .. (84.26,182.13) .. controls (82.29,182.13) and (80.67,180.54) .. (80.65,178.57) -- cycle ;
\draw   (80.65,194.57) .. controls (80.64,192.6) and (82.22,191) .. (84.19,191) .. controls (86.16,191) and (87.77,192.6) .. (87.79,194.57) .. controls (87.81,196.54) and (86.23,198.13) .. (84.26,198.13) .. controls (82.29,198.13) and (80.67,196.54) .. (80.65,194.57) -- cycle ;
\draw   (80.65,209.57) .. controls (80.64,207.6) and (82.22,206) .. (84.19,206) .. controls (86.16,206) and (87.77,207.6) .. (87.79,209.57) .. controls (87.81,211.54) and (86.23,213.13) .. (84.26,213.13) .. controls (82.29,213.13) and (80.67,211.54) .. (80.65,209.57) -- cycle ;

\draw (55,146.4) node [anchor=north west][inner sep=0.75pt]    {$A_{2}( t)$};
\draw (52,220.4) node [anchor=north west][inner sep=0.75pt]    {$A_{N}( t)$};
\draw (55,76.4) node [anchor=north west][inner sep=0.75pt]    {$A_{1}( t)$};
\draw (425,105.4) node [anchor=north west][inner sep=0.75pt]    {$Geo^{[ X]} /D/1$};

\end{tikzpicture}
}
\caption{Operation of the idealized full knowledge scheduler.}
\label{figFullKnowledgeScheduler}
\end{subfigure}
\begin{subfigure}{.5\textwidth}
\hspace{1.00cm}
\centering
\input{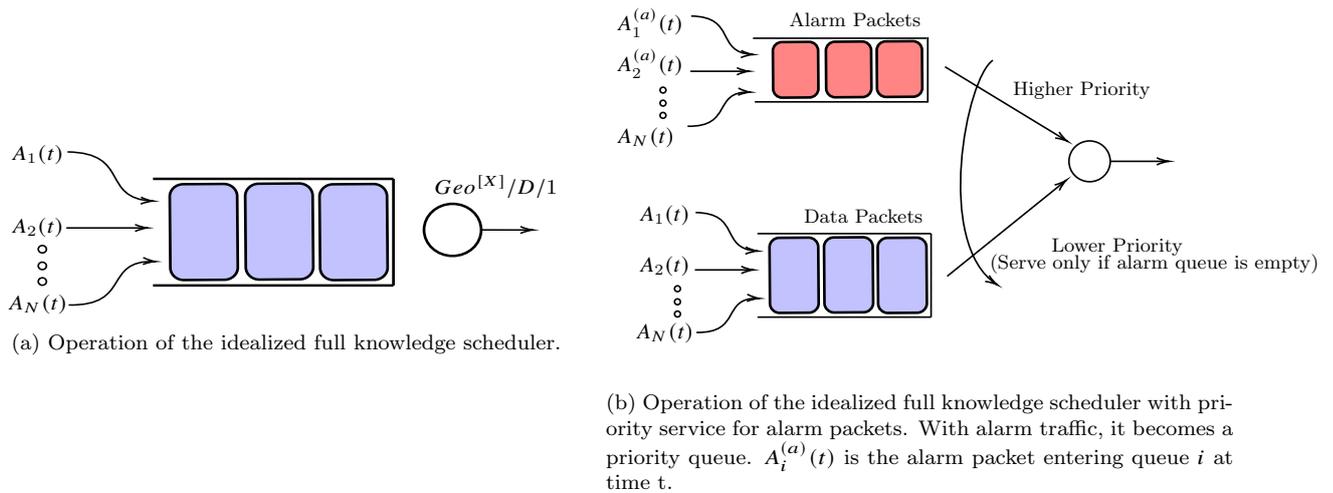}
\caption{Operation of the idealized full knowledge scheduler with priority service for alarm packets. With alarm traffic, it becomes a priority queue. $A^{(a)}_i(t)$ is the alarm packet entering queue $i$ at time t.}
\label{figFullKnowledgeSchedulerWithAlarms}
\end{subfigure}
\caption{Operation of the idealized full knowledge scheduler with priority service for alarm packets.}
\end{figure*}
Instead of a single arrival stream, every queue now has two arrival streams. One carrying non alarm traffic (which we will call data traffic) and another carrying alarm packets. The Bernoulli $\lambda_i$ input process is split (in a memoryless manner) with a fraction $1-\alpha_i$ entering the data stream and a fraction $\alpha_i$ entering the alarm packet stream. Given that alarm packets are generated very rarely 
$\alpha_i$ is much smaller than $1-\alpha_i$. In this case, we modify QZMAC as follows. We include an alarm minislot ($T_a$) at the beginning of every slot. Queues with alarm packets can use this slot to send out a high power jamming signal to inform the others that alarm packets are in the system. With this, we modify QZMAC to accommodate such traffic and call the resulting protocol QZMAC$_a$. 
So, at the beginning of slot $t$
\begin{itemize}
\item The slot begins with the alarm minislot $T_a$. Every queue that has an alarm packet sends out a jam signal.
\item If the queues \emph{do not} sense any power in $T_a$, QZMAC$_a$ operates exactly like QZMAC in slot $t$.
\item If the queues \emph{do} sense power in $T_a$, normal operation according to QZMAC is suspended for slot $t$.
\begin{itemize}
\item All queues with alarm packets back off as usual over $\{T_a+T_p+1,\cdots,T_a+T_p+T_c\}$ and if a winner emerges, it transmits its alarm packet.
\item If no winner emerges, the system still has alarm packets, which means the protocol goes into alarm mode in slot $t+1$ as well. 
\item This continues until there are no alarm packets left in the system.
\end{itemize}
\end{itemize}
\subsubsection{Performance of QZMAC with Alarm Traffic}\label{secQZMAC_aSimulationResults}
As can be seen from Fig.~\ref{figAlarmTraffic_Light}, when alarm traffic is light QZMAC$_a$ provides the same delay to alarm packets as the full knowledge scheduler (the plain red and plan blue curves overlap significantly). This is because alarm packets are given maximum priority in QZMAC$_a$. Queues with alarm packets contend for the channel and since their arrival rate ($\alpha_i\lambda_i$) is small, these packets essentially see a service time of 1 slot. Moreover, the delay to data traffic (blue-with-diamonds curve) is comparable to that without any alarm traffic (dashed green curve), i.e., plain QZMAC. However, if the alarm traffic load becomes significant compared to data traffic, the performance delivered to alarm packets (solid red curve at the bottom of Fig.~\ref{figAlarmTraffic_Heavy}) is not affected much due to priority queueing (this delay does show an increase near saturation from 1.005 slots to 1.357 slots, but this is still acceptable). However, the delay of data packets (blue-with-diamonds curve) degrades considerably and, as Fig.~\ref{figAlarmTraffic_Heavy} shows, is much worse than that of QZMAC without any alarm traffic (dashed green curve). 
%
The curves in red are obtained by the full-knowledge scheduler with priorities (see Fig.~\ref{figFullKnowledgeSchedulerWithAlarms}). 
\begin{figure*}[tbh]
\begin{subfigure}{.5\textwidth}
\centering
\includegraphics[height=5.00cm, width=8.25cm]{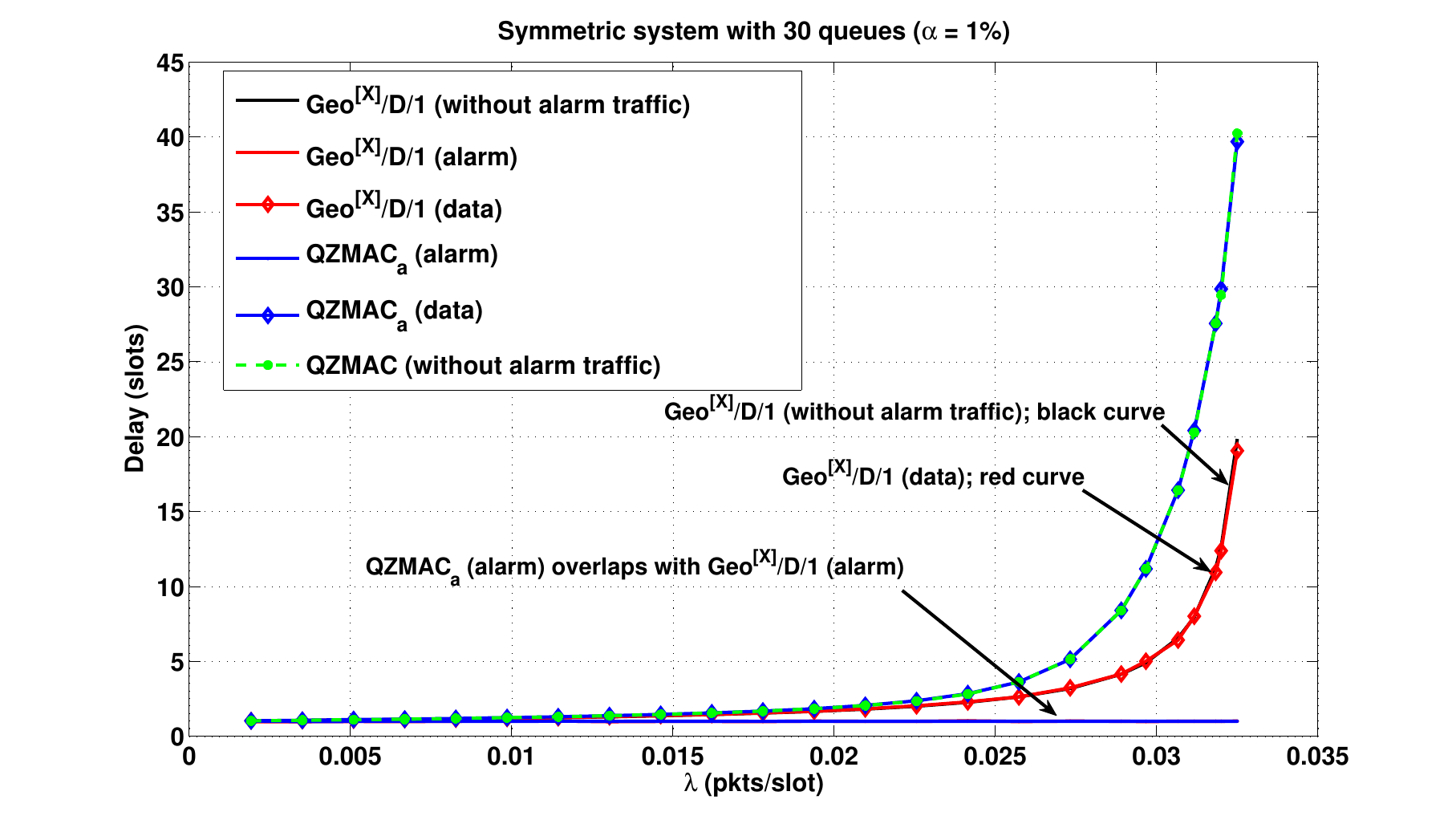}
\caption{Performance of QZMAC$_a$ with $\alpha_i=0.01$ for all $i$, i.e., 1$\%$ alarm traffic.}
\label{figAlarmTraffic_Light}
\end{subfigure}
\begin{subfigure}{.5\textwidth}
\centering
\includegraphics[height=5.00cm, width=8.25cm]{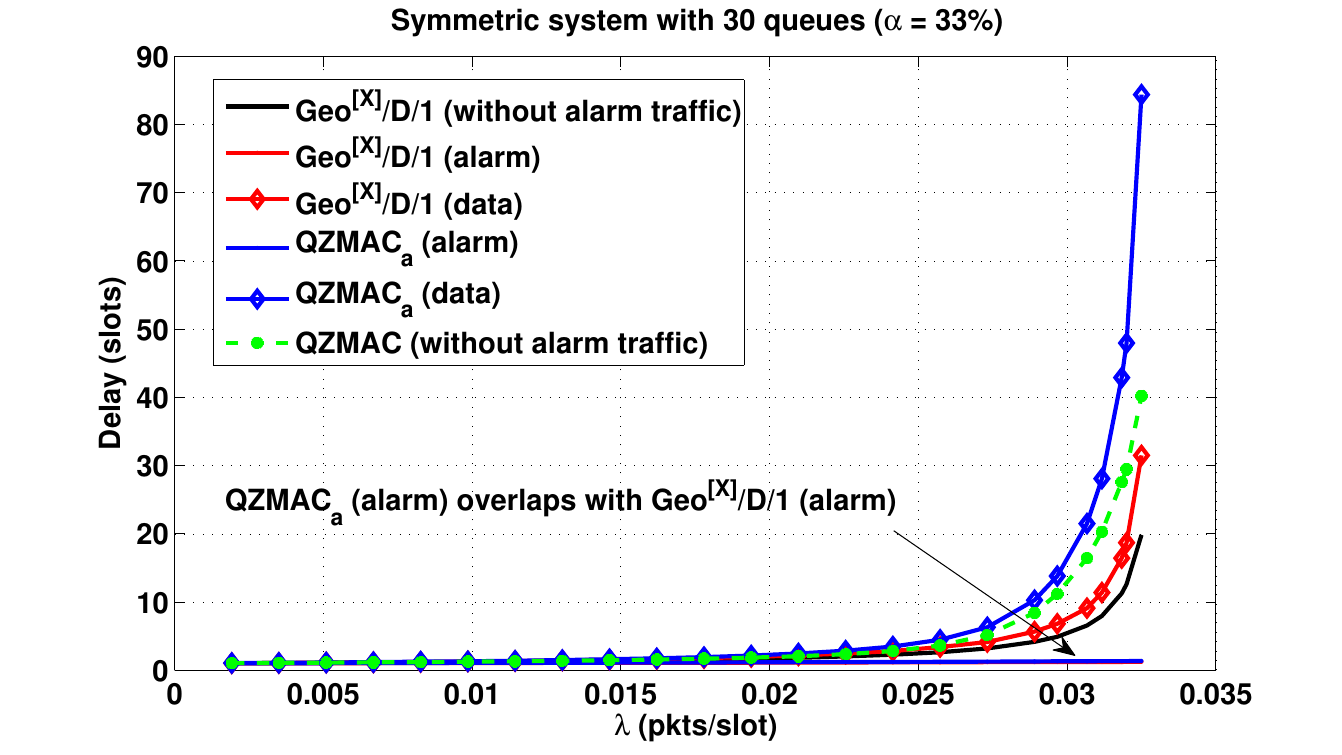}
\caption{Performance of QZMAC$_a$ with $\alpha_i=0.33$ for all $i$, i.e., 33$\%$ alarm traffic.}
\label{figAlarmTraffic_Heavy}
\end{subfigure}
\caption{Performance of QZMAC$_a$ with different alarm traffic arrival rates.}
\end{figure*}
\subsection{Effects of Nonidealities in Clear Channel Assessment}\label{appendix:CCAinQZMAC}
As mentioned earlier, in standard hypothesis testing parlance, the event when a CCA fails despite there being no activity on the channel is called a \enquote{Fale Alarm,} while the event when the CCA is successful even when a node in the network is transmitting is called a \enquote{Miss.} We denote the corresponding probabilities by $p_{FA}$ and $p_{miss}$ respectively. Our extensive experimentation, as described in Sec.~\ref{sec:Experiments}, has revealed the following facts
\begin{itemize}
    \item For networks with diameter up to $8$ meters, $p_{FA}\approx0$. This happens because the noise floor of modern sensors is low enough (-110 dBm for a 2MHz bandwidth) that no CCA failure triggers when the channel is free, and
    \item In any given time slot (not minislot), at most \emph{one} CCA \enquote{Miss} occurs.
\end{itemize}
Therefore, going forward, our analysis will assume that only \enquote{Miss} events occur and at most one node's CCA can register a spurious failure in any time slot. Table~\ref{table:cca-errors-v-misalignments} illustrates the three types of misalignment that can occur. In the sequel, we will refer to them as \enquote{M1,} \enquote{M2} and \enquote{M3} errors. We now describe them in detail and explain how each type is resolved. Recall that the node currently transmitting is called the \emph{incumbent} and in time slot $t$, $i^*=\argmax{i\in[N]}V_i(t).$
\begin{itemize}
    \item[M1.] This is the situation in slot $t+2$ in Table~\ref{table:cca-errors-v-misalignments}. The CCA of Node~$j$ ($j\neq i^*$) succeeds, and it assumes that Node~$i^*$ will now transmit. However, upon decoding the header, the node discovers that the incumbent is nonempty, infers that a CCA Miss has occurred and corrects its copy of the $\mathbf{V}$-vector.
    \item[M2.] This is the situation in slot $t+3$ in Table~\ref{table:cca-errors-v-misalignments}. The $i^*$ Node assumes that the incumbent is empty and, being \emph{nonempty,} begins transmitting. The transmissions from the two nodes now collide and further provisions are now required within QZMAC to extricate the network from this state. The subroutine discussed in Sec.~\ref{sec:handling-CCA-errors-subroutine} is designed to accomplish this.
    \item[M3.] This is the situation in slot $t'$ in Table~\ref{table:cca-errors-v-misalignments}. This case is similar to M2., but now, Node~$i^*$ is empty. The node then decodes the header of the packet, detects the CCA miss and corrects its copy of the $\mathbf{V}$-vector.
\end{itemize}
We thus see that two of the three potential types of misalignment can be detected and rectified quite easily, while only one necessitates modification to QZMAC. 
\\
\indent The polling portion is modified as follows. 
We set a threshold $K_{thr}$ such that any time the transmitting node(s) perceive $K_{thr}$ time outs, the node(s) assume it is due to CCA errors. These transmitting nodes then enter a state we term \enquote{{\tt COLL}} and a \texttt{RESET} Subroutine is triggered. Informally, the subroutine does the following.
\texttt
{
\begin{enumerate}
   \item[\textbf{\color{blue}[1]}]
   Until the transmitting nodes receive a reset beacon (denoted {\tt RSTBCN}), the first $T_p$ minisolt of every slot following the $K_{thr}$ timeouts is '{\tt BUSY}', ensuring no other queues attempt transmissions during those slots.
   \item[\textbf{\color{blue}[2]}]
   In every such slot, the transmitting nodes perform a random backoff and transmit a \texttt{RESET RQST} packet to the Base Station (BS).
   \item[\textbf{\color{blue}[3]}]
   Upon reception of a \texttt{RESET RQST} packet, the BS broadcasts a \texttt{RESET} beacon.
\end{enumerate}
}
  \begin{algorithm}[tb]
  \small
    \floatname{algorithm}{Protocol}
      \allowdisplaybreaks
      	\caption{{\bf RESET} (Subroutine)} \label{protocol:RESET}
      	\begin{algorithmic}[1]
      	\State \textbf{Input:} {\tt THRSLD, NDST, RSTBCN}
       	
      	\State \textbf{Init:} \texttt{RSTBCN == FALSE} {\color{blue}\Comment{Keeps track of whether a RESET beacon has been received from the BS.}}
      	$\tau\leftarrow1,$  {\color{blue}\Comment{$\tau$ keeps track of the minislot number}}
     	\State \textbf{Init:} $\forall~j\in[N], V_k(t)\leftarrow k$ 
      \If{\texttt{NDST == COLL}} {\color{blue}\Comment{Resetting involves only colliding nodes.}}\label{resetSubroutine:RstRqstBegins}
            \While{{\tt NDST == COLL \&\& RSTBCN == FALSE}}
      	    \State Tx \texttt{DUMMY PKT} {\color{blue}\Comment{To jam all $T_p$ minislots.}}
      	    \State {\bf Over} $\tau\in\{4,5,\cdots,T_c+3\}$ {\bf do}
      	    \State $U_j\sim Unif\left(\lbrace4,5,\cdots,T_c+3\rbrace\right)$
      	    \State Wait for $U_j-4$ minislots. 
      	    \State If \texttt{CCA = SUCCESS,} Tx \texttt{RSTRQST} pkt {\color{blue}\Comment{Tx a Reset Request to the BS.}}
      	    \If{\texttt{RSTBCN == TRUE}} {\color{blue}\Comment{The BS transmits a RESET beacon.}}
      	        \State{\texttt{NDST $\leftarrow$ NOCOLL}}
      	    \EndIf
      	    \EndWhile
      \EndIf\label{resetSubroutine:RstRqstEnds}
      \If{\texttt{RSTBCN == TRUE}}{\color{blue}\Comment{Every node in the network resets.}}\label{resetSubroutine:RsttingBegins}
            \State{$\forall k\in[N],~V_k(t)\leftarrow k.$}
      	\EndIf
      	\label{resetSubroutine:RsttingEnds}
      	\end{algorithmic} 
 \end{algorithm}
The network then simply resets, i.e., the copy of the vector $\mathbf{V}(t)$ at each node is reset (without jettisoning any existing packets) and the protocol starts afresh. Note that while only the nodes in a \texttt{COLL} state trigger a reset (Steps~\ref{resetSubroutine:RstRqstBegins}-\ref{resetSubroutine:RstRqstEnds}), \emph{all} nodes are required to reset their copies of the $\mathbf{V}(t)$ vector upon receiving a \texttt{RESET} beacon from the BS (Steps~\ref{resetSubroutine:RsttingBegins}-\ref{resetSubroutine:RsttingEnds}).
As simulation results in Fig.~\ref{fig:perfectCCAversusPMiss3eMinus6}
show, this simple modification resolves the CCA problem but does not harm delay performance adversely.\\
Furthermore, as the simulation results in Fig.~\ref{fig:perfectCCAversusPMiss3eMinus6} show, even at a miss rate $p_{miss} = 3\times10^{-4},$ which is $100$ times larger than the measures miss probability, the largest increase in mean delay due to the modification is $1.5266$ slots ($\approx4.32\%$), which is quite minimal.
\begin{figure}[tbh]
\centering
\includegraphics[height=4.5cm, width=9.0cm]{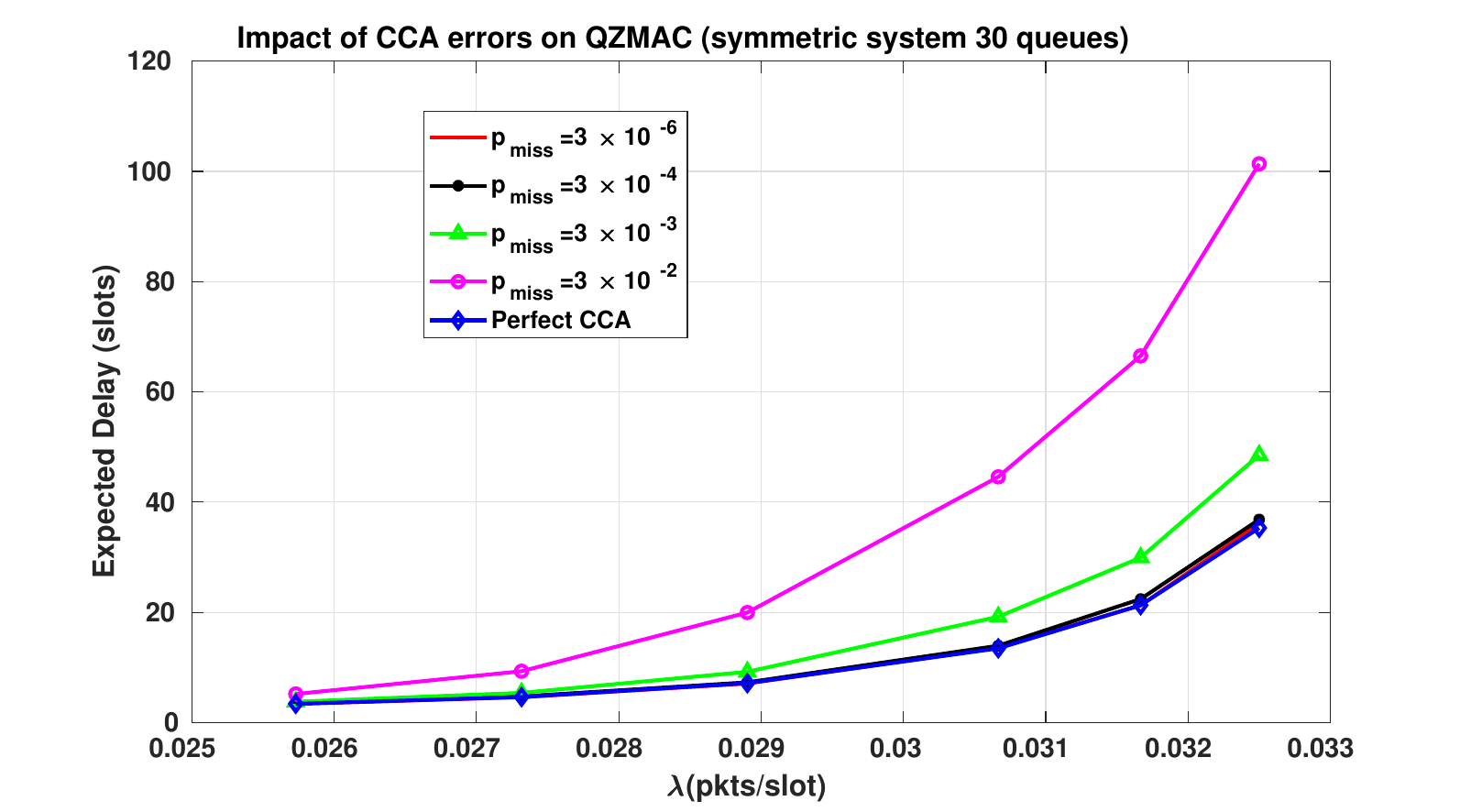}
\caption{Comparing mean delay with and without CCA errors. We ran QZMAC for $5.25\times10^6$ slots with $K_t=5$ slots. The curve in red shows the performance of QZMAC with the RESET subroutine. Mean delay loss with due to imperfect CCA is at most $1.5266$ slots (occurring at $\lambda=0.0325$) even up to a CCA miss probability of $3\times10^{-4}$, which is quite negligible.}
\label{fig:perfectCCAversusPMiss3eMinus6}
\end{figure}
\subsection{Channel Utilization}
In Table~\ref{table:channel-utilization-simulations-qz_ez_z}, we show the results of simulating QZMAC, EZMAC and ZMAC in a system containing 30 queues with \emph{equal} arrival rates. This comparison is necessary, because ZMAC is defined assuming a system with equal arrival rates. Please note that these are simulation results, whereas Table~\ref{table:channel-utilization-comparison} reports the results of an experiment. Since in the experiment, the CCA loss anomaly was rare (see Sec.~\ref{sec:ccaErrorsInferenceAndHandling}), a simulation captures the actual performance almost exactly. 
    
    Since the system in our simulation has 30 queues (with equal arrival rates), the per queue arrival rate cannot exceed $\frac{1}{30} \approx 0.033.$ We compute the channel utilization as defined before, i.e., as
    \begin{eqnarray*}
 \zeta^{\pi}(\mathbf{Q}(0)) &:=& \lim_{t\rightarrow\infty}\mathbb{P}^{\pi}_{\mathbf{Q}(0)} \left(\sum_{j=1}^N \mathbb{I}_{\{Q_j(t)>0\}}D_j(t)>0\bigg|\sum_{j=1}^N Q_j(t)>0\right)\nonumber\\
 &=& \lim_{t\rightarrow\infty} \frac{\sum_{s=0}^{t-1}\sum_{j=1}^N \mathbb{I}_{\{Q_j(s)>0\}}\mathbb{I}_{\{D_j(s)>0\}}}{\sum_{s=0}^{t-1}\mathbb{I}_{\{\sum_{j=1}^N Q_j(s)>0\}}}.
\end{eqnarray*}
Table~\ref{table:channel-utilization-simulations-qz_ez_z} clearly shows that our protocols, QZMAC and EZMAC waste the channel less often and hence, clear the system of packets more efficiently than ZMAC.

    \begin{table}[tb]
        \centering
        \begin{tabular}{c|c|c|c}
        Algorithm & QZMAC & EZMAC & ZMAC\\
        \specialrule{2pt}{2pt}{0pt}
            $\lambda = 0.032$ & 0.9541 & 0.9414 & 0.9224 \\
            \hline
            $\lambda=0.018$ & 0.9271 & 0.9110 & 0.8362 \\
            \specialrule{2pt}{2pt}{0pt}
        \end{tabular}
        \caption{Comparing the channel utilization of QZMAC, EZMAC and ZMAC for different arrival rates.
        As before, $\zeta^{QZMAC}>\zeta^{EZMAC}>\zeta^{ZMAC}$, showing that QZMAC and EZMAC waste the channel less often, thereby clearing the network of packets more efficiently.}
        \label{table:channel-utilization-simulations-qz_ez_z}
    \end{table}


%

%
%
%
%
%

\ifCLASSOPTIONcaptionsoff
  \newpage
\fi

\end{document}